\documentclass[sigconf, nonacm]{acmart}
\AtBeginDocument{%
  \providecommand\BibTeX{{%
    \normalfont B\kern-0.5em{\scshape i\kern-0.25em b}\kern-0.8em\TeX}}}

\usepackage{amsmath,amsthm}
\usepackage{bm}
\usepackage{framed}
\usepackage{algorithm}
\usepackage{algorithmicx}
\usepackage[noend]{algpseudocode}
\usepackage{xspace}
\usepackage{ulem}
\usepackage{mathbbol}
\usepackage{subfig}
\usepackage{enumitem}
\newlength\myindent
\newcommand{\eat}[1]{}

\newcommand{\goal}[1]{\textcolor{blue}{Goal: #1}\PackageWarning{Goal:}{#1!}}
\newcommand{\xh}[1]{\noindent{ \textcolor{orange}{\footnotesize[{\bf Xi}: #1]}}}

\newcommand{\miti}[1]{{#1}} 

\theoremstyle{plain}
\newtheorem{thm}{Theorem}[section]
\theoremstyle{definition}
\newtheorem{defn}{Definition}[section]
\newtheorem{example}{Example}[section]
\newtheorem{prop}{Proposition}[section]
\newtheorem{lemma}{Lemma}[section]

\setcopyright{none}
\renewcommand\footnotetextcopyrightpermission[1]{}
\usepackage{fixltx2e}
\usepackage{eqparbox}
\newdimen{\algindent}
\algnewcommand\LeftComment[2]{%
\hspace{#1\algindent}$\blacktriangledown$ \eqparbox{COMMENT}{#2} \hfill %
 }

\newcommand\sysname{\textit{CacheDP}\xspace}
\newcommand\Attribute{\ensuremath{\mathcal{A}}}

\newcommand\Schema{\ensuremath{\mathcal{R}}}

\newcommand\Predicate{\ensuremath{\phi}}
 
\newcommand\TotalPrivacyBudget{\ensuremath{\mathcal{B}}}

\newcommand\ConsumedPrivacyBudget{\ensuremath{B}_c}

\newcommand\PrivacyBudgetPrecision{\ensuremath{\epsilon}_{\bot}}

\newcommand\LaplaceMechanismNew{\ensuremath{\mathcal{L}_b}}

\newcommand\StrategyRaw{\ensuremath{\mathbb{A}}}
\newcommand\StrategyFull{\ensuremath{\mathbf{A}}}
\newcommand\WorkloadRaw{\ensuremath{\mathbb{W}}}
\newcommand\WorkloadFull{\ensuremath{\mathbf{W}}}
\newcommand\DataVectorRaw{\ensuremath{\mathbb{x}}}
\newcommand\DataVector{\ensuremath{\mathbf{x}}}

\newcommand\BucketMatrix{\ensuremath{\mathbb{T}}}

\newcommand\WorkloadResponse{\ensuremath{z}}

\newcommand\WorkloadQuery{\ensuremath{Q}}
\newcommand\WorkloadMatrix{\WorkloadRaw}

\newcommand\BList{\ensuremath{\mathbf{b}}}

\newcommand\LooseNoiseParameter{\ensuremath{b_L}}

\newcommand\TightNoiseParameter{\ensuremath{b_T}}
\newcommand\PaidNoiseParameter{\ensuremath{b_{\StrategyMatrixPaid}}}
\newcommand\PaidEpsilon{\ensuremath{\epsilon_{\StrategyMatrixPaid}}}
\newcommand\CandidatePaidNoiseParameter{\ensuremath{b_{P0}}}

\newcommand\StrategyRawPaid{\mathbb{P}}
\newcommand\StrategyRawProactive{\Delta\StrategyRawPaid}

\newcommand\StrategyMatrixFree{\ensuremath{\mathbf{F}}}
\newcommand\StrategyMatrixPaid{\ensuremath{\mathbf{P}}}

\newcommand\StrategyMatrixGlobal{\StrategyRaw^*}
\newcommand\StrategyFullGlobal{\StrategyFull^*}
\newcommand\StrategyMatrixGlobalPerAttribute[1]{\ensuremath{\StrategyRaw{#1}^*}}

\newcommand{\noisev}{\ensuremath{\bm \eta}}

\newcommand\GroundTruthResponse{\ensuremath{y}}

\newcommand\NoisyResponse{\ensuremath{\tilde{\GroundTruthResponse}}}
\newcommand\NoisyResponseList{\ensuremath{\tilde{\mathbf{\GroundTruthResponse}}}}

\newcommand\SizeOfWorkload{\ensuremath{\ell}}

\newcommand\DomainSize{\ensuremath{n}}
\newcommand\DomainVector{\DataVector}

\newcommand{\NoisyWorkloadResponse}{\ensuremath{\tilde{\mathbf{z}}}}
\newcommand\NoisyStrategyResponse{\ensuremath{\tilde{\mathbf{y}}}}
\newcommand{\HistoricalPrivacyBudgetRP}{\ensuremath{\epsilon_\Cache}}

\newcommand\NoisyProactiveResponse{\NoisyResponseList_{\StrategyRawProactive{}}}

\newcommand\Cache{\ensuremath{\mathcal{C}}}
\newcommand\StratExpanderLimit{\ensuremath{\lambda}}

\algdef{SE}[VARIABLES]{ModuleState}{EndModuleState}
{\algorithmicvariables}
{}
\algnewcommand{\algorithmicvariables}{\textbf{Module State}}

\algnewcommand{\LineComment}[1]{\State \(\triangledown\) #1}

\usepackage{tikz}
\usepackage[framemethod=TikZ]{mdframed}
\mdfdefinestyle{MyFrame}{%
	linecolor=black,
	outerlinewidth=1pt,
	roundcorner=0pt,
	innerrightmargin=5pt,
	innerleftmargin=5pt,
	backgroundcolor=white
}

\usepackage{xcolor}
\usetikzlibrary{arrows,positioning}
\usetikzlibrary{trees}
\usetikzlibrary{backgrounds,fit}

\tikzset{
  mynode/.style = {circle, minimum size=0.8cm, align=center, inner sep=0pt, text centered, font=\sffamily},
  normalnode/.style = {mynode, text=black},
  strategy1node/.style = {mynode, draw=black, thin, text=blue},
  strategy2node/.style = {mynode, draw=black, thick,  text=magenta},
  strategy1and2node/.style = {mynode, draw=black, very thick, text=violet}, 
  proactive2node/.style = {mynode, draw=black, dashed, text=black},
}

\usepackage{versions}

\newcommand{\squishlist}{
	\begin{list}{$\bullet$}
		{
			\setlength{\itemsep}{0pt}
			\setlength{\parsep}{3pt}
			\setlength{\topsep}{3pt}
			\setlength{\partopsep}{0pt}
			\setlength{\leftmargin}{1.5em}
			\setlength{\labelwidth}{1em}
			\setlength{\labelsep}{0.5em} } }

	\newcommand{\squishend}{
\end{list}  }

\usepackage{hyperref}
\usepackage{csquotes}
\renewcommand{\mkbegdispquote}[2]{\itshape}
\usepackage{multirow,multicol}
\includeversion{techreport}
\excludeversion{vldbpaper}
\newcommand\vldbdoi{XX.XX/XXX.XX}

\newcommand\vldbavailabilityurl{https://gitlab.uwaterloo.ca/m2mazmud/cachedp-public}

\begin{document}

\title{Cache Me If You Can: Accuracy-Aware Inference Engine for Differentially Private Data Exploration}
\begin{techreport}
\titlenote{This is an extended version of our paper that will appear in VLDB'23~\cite{dpcache-vldb}. Our artifact is available online~\cite{own-artifact}.}
\end{techreport}

\author{Miti Mazmudar}
\email{miti.mazmudar@uwaterloo.ca}
\affiliation{%
	\institution{University of Waterloo}
}

\author{Thomas Humphries}
\email{thomas.humphries@uwaterloo.ca}
\affiliation{%
	\institution{University of Waterloo}
}

\author{Jiaxiang Liu}
\email{j632liu@uwaterloo.ca}
\affiliation{
	\institution{University of Waterloo}
}

\author{Matthew Rafuse}
\email{matthew.rafuse@uwaterloo.ca}
\affiliation{%
	\institution{University of Waterloo}
}

\author{Xi He}
\email{xi.he@uwaterloo.ca}
\affiliation{%
	\institution{University of Waterloo}
}

\begin{abstract}
Differential privacy (DP) allows data analysts to query databases that contain users' sensitive information while providing a quantifiable privacy guarantee to users. 
Recent interactive DP systems such as APEx provide accuracy guarantees over the query responses, but fail to support a large number of queries with a limited total privacy budget, as they process incoming queries independently from past queries. 
We present an interactive, accuracy-aware DP query engine, \sysname, which 
utilizes a differentially private cache of past responses, to answer the current workload at a lower privacy budget, while meeting strict accuracy guarantees.
We integrate complex DP mechanisms with our structured cache, through novel cache-aware DP cost optimization.
Our thorough evaluation illustrates that \sysname can accurately answer various workload sequences, while lowering the privacy loss as compared to related work. 
\end{abstract}

\settopmatter{printfolios=true}
\maketitle


\begin{vldbpaper}
\ifdefempty{\vldbavailabilityurl}{}{
	\vspace{.3cm}
	\begingroup\small\noindent\raggedright\textbf{PVLDB Artifact Availability:}\\
	The source code, data, and/or other artifacts have been made available at \url{\vldbavailabilityurl}.
	\endgroup
}
\end{vldbpaper}

\section{Introduction}
Organizations often collect large datasets that contain users' sensitive data and permit
data analysts to query these datasets for aggregate statistics. 
However, a curious data analyst may use these query responses to infer a user's record. 
Differential Privacy (DP)~\cite{dwork_differential_2006, privacy_book} allows organizations to 
provide a guarantee to 
their users that the presence or absence of their record in the dataset will only change the distribution of the query response by a small factor, given by the privacy budget. 
This guarantee is typically achieved by perturbing the query response with noise that is inversely proportional to the privacy budget. 
Thus, DP systems face an accuracy-privacy trade-off: they should provide accurate query responses, while reducing the privacy budget spent.
DP has been deployed at the US Census Bureau~\cite{nthemap08}, Google~\cite{googledp} and Microsoft~\cite{msr17}.

Existing DP deployments~\cite{nthemap08,prochlo17,msr17,privatesql} mainly consider a non-interactive setting, where the analyst provides all queries in advance. 
Whereas in interactive DP systems~\cite{pinq,flex,googledp,opendp}, data analysts supply queries one at a time. 
These systems have been difficult to deploy as they often assume an analyst has DP expertise.
First, data analysts need to choose an appropriate privacy budget per query. 
Second, data analysts require each DP noisy query response to meet a specific accuracy criterion, 
whereas DP systems only seek to minimize the expected error over multiple queries. 
Ge et al.'s APEx~\cite{ge_apex:_2017} eliminates these two drawbacks, as data analysts need only specify accuracy bounds in the form of an error rate $\alpha$ and a probability of failure $\beta$. 
APEx chooses an appropriate DP mechanism and calibrates the privacy budget spent on each workload, to fulfill the accuracy requirements. 
However, interactive DP systems may run out of privacy budget for a large number of queries.

We observe that we can further save privacy budget on a given query, by exploiting \textit{past}, related noisy responses, and thereby, we can answer a larger number of queries interactively. 
The DP post-processing theorem allows arbitrary computations on noisy responses without affecting the DP guarantee.
Hay et al.~\cite{hay2010boosting} have applied this theorem to enforce consistency constraints among noisy responses to related range queries, thereby improving their accuracy, through \textit{constrained inference}.
Peng et al. have proposed caching noisy responses and reusing them to answer future queries in Pioneer~\cite{pioneer}.
However, their cache is unstructured and only operates with simple DP mechanisms such as the Laplace mechanism. 


We design a usable interactive DP query engine, \sysname, with a built-in differentially private cache, to support data analysts in answering data exploration workloads accurately, without requiring them to have any knowledge of DP.
Our system is built on top of an existing non-private DBMS and interacts with it through standard SQL queries. 
\sysname meets the analysts' $(\alpha, \beta)$ accuracy requirements on each workload, while minimizing the privacy budget spent per workload. 
We note that a similar reduction in privacy budget could be obtained if an expert analyst planned their queries, however our system removes the need for such planning.

Our contributions address four main challenges in the design of our engine. 
\emph{First}, we structure our cache to maximize the possible reuse of noisy responses by DP mechanisms (Section~\ref{sec:system_design}). 
Our cache design fully harnesses the post-processing theorem in the interactive setting, for cached noisy responses. 
\emph{Second}, we integrate existing DP mechanisms with our cache, namely Li et al.'s Matrix Mechanism~\cite{li2015matrix} (Section~\ref{sec:mmm-module}), and Koufogiannis et al.'s Relax Privacy mechanism~\cite{relaxed_privacy} (Section~\ref{sec:rp-module}). 
In doing so, we address technical challenges that arise due to the need to maintain accuracy requirements over cached responses while minimizing the privacy budget, and thus, we provide a novel privacy budget cost estimation algorithm. 

\emph{Third}, we extend our cache-aware DP mechanisms with two modules, which further reduce the privacy budget (Section~\ref{sec:strategy-module}). 
Specifically, we apply DP sensitivity analysis to proactively fill our cache, and we apply constrained inference to increase cache reuse. 
We note that \sysname internally chooses the DP module with the lowest privacy cost per workload, removing cognitive burden on data analysts. 
\emph{Fourth}, we develop the design of our cache 
to handle queries with multiple attributes efficiently (Section~\ref{sec:multi-attribute-case}). 

\emph{Finally}, we conduct a thorough evaluation of our \sysname against related work (APEx, Pioneer), in terms of privacy budget consumption and performance overheads (Section~\ref{sec:evaluation}). 
We find that it consistently spends lower privacy budget as compared to related work, for a variety of workload sequences, while incurring modest performance overheads. 
Through an ablation study, we deduce that our standard configuration with all DP modules turned on, is optimal for the evaluated workload sequences. Thus, researchers implementing our system need not tinker with our module configurations. 
\begin{vldbpaper}
\textbf{This paper contains several theorems and lemmas; their proofs can be found in the extended version of the paper~\cite{dpcacheextended}}. 
\end{vldbpaper}

\section{Background}~\label{sec:background}
We consider a single-table relational schema \Schema~ across $d$ attributes: $\Schema(\Attribute_1,\ldots \Attribute_d)$.
The domain of an attribute $\Attribute_i$ is given by $dom(\Attribute_i)$ and the full domain of $\Schema$ is $dom(\Schema)=dom(\Attribute_1)\times \cdots \times dom(\Attribute_d)$. Each attribute $\Attribute_i$ has a finite domain size $|dom(\Attribute_i)| = n_i$. The full domain has a size of $n=\prod_i n_i$.
A database instance $D$ of relation $\Schema$ is a multiset whose elements are values in $dom(\Schema)$.

A predicate $\Predicate:dom(\Schema)\rightarrow \{0,1\}$ is an indicator function  specifying which database rows we are interested in (corresponds to the \texttt{WHERE} clause in SQL). 
A \emph{linear or row counting query (RCQ)}
takes a predicate $\Predicate$ and returns the number of tuples in $D$ that satisfy $\Predicate$, i.e., $\Predicate(D)=\sum_{t\in D} \Predicate(t)$.
This corresponds to querying \texttt{SELECT COUNT(*) FROM $D$ WHERE \Predicate} in SQL.
We focus on RCQs for this work as they are primitives
that can be used to express histograms, multi-attribute range queries, marginals, and data cubes.

In this work, we express  RCQs
 as a matrix. Consider $dom(\Schema)$ to be an ordered list.  We represent a database instance $D$ by a data (column) vector $\DataVectorRaw$ of length $n$, where $\DataVectorRaw[i]$ is the count of $i$th value from $dom(\Schema)$ in $D$. 
After constructing $\DataVectorRaw$, we  represent any RCQ
as a length-$\DomainSize$ vector 
$\mathbb{w}$ with $\mathbb{w}[i]\in \{0,1\}$ for $i=1,\ldots,n$.
To obtain the ground truth response for a RCQ 
$\mathbb{w}$, we can simply compute $\mathbb{w}\cdot \DataVectorRaw$.
Hence, we can represent a workload of $\SizeOfWorkload$ RCQs
as an $\SizeOfWorkload \times \DomainSize$ matrix $\WorkloadRaw$ and answer this workload  by matrix multiplication, as $\WorkloadRaw \DataVectorRaw$. 

When we partition the full domain $dom(\Schema)$ into a set of $\DomainSize'$ disjoint buckets, the data vector $\DataVectorRaw$ and the workload matrix $\WorkloadRaw$ over the full domain $dom(\Schema)$ can be mapped to a vector $\DataVector$ of size $\DomainSize'$ and a matrix $\WorkloadFull$ of size $\SizeOfWorkload \times \DomainSize'$, respectively.  
We also consider a workload matrix $\WorkloadFull$ as a \emph{set} of RCQs,
and hence applying a set operator over a workload matrix is equivalent to applying this operator over a set of RCQs. For example, $\WorkloadFull'\subseteq \WorkloadFull$ means the set of RCQs in $\WorkloadFull'$ is a subset of the RCQs in $\WorkloadFull$.
We follow a differential privacy model with a trusted data curator.

\begin{defn}[$\epsilon$-Differential Privacy (DP) \cite{dwork_differential_2006}]\label{DP}
A randomized mechanism $M: \mathcal{D}\rightarrow \mathcal{O}$ satisfies $\epsilon$-DP if
for any output sets $O\subseteq \mathcal{O}$,
and any \emph{neighboring} database pairs $(D,D')$, i.e., $|D\backslash D'\cup D'\backslash D| = 1$, 
\begin{equation}\label{def:dp}
\Pr[M(D) \in O] \leq e^{\epsilon} \Pr[M(D')\in O].
\end{equation}
\end{defn}

The privacy parameter $\epsilon$ is also known as privacy budget. A classic mechanism to achieve DP is the Laplace mechanism. We present the matrix form of Laplace mechanism here. 
\begin{thm}[Laplace mechanism~\cite{dwork_differential_2006,li2015matrix}] Given an $l\times n$ workload matrix $\WorkloadFull$ and a data vector $\DomainVector$, the Laplace Mechanism $\LaplaceMechanismNew$ outputs  $\LaplaceMechanismNew(\WorkloadFull,\DomainVector)=\WorkloadFull\DomainVector+Lap(b)^l$ where $Lap(b)^l$ is a vector of $l$ i.i.d. samples from a Laplace distribution with scale $b$. If $b\geq \frac{\|\WorkloadFull\|_1}{\epsilon}$, where $\|\WorkloadFull\|_1$ denotes the $L_1$ norm of  $\WorkloadFull$,
then  $\LaplaceMechanismNew(\WorkloadFull,\DomainVector)$ satisfies $\epsilon$-DP. 
\end{thm}

\begin{table}[t!]
    \centering
    \caption{Notation}\label{table:notation} \vspace{-3mm}
    \begin{tabular}{r p{6.5cm}}
        \textbf{Notation} & \textbf{Description}\\ \hline
        
        
        $\DataVectorRaw, \mathbb{w}, \WorkloadRaw, \StrategyRaw$ & raw data vector, query vector, query workload matrix, strategy matrix over full domain $dom(\Schema)$ \\
        
        $\DataVector, \mathbf{w}, \WorkloadFull, \StrategyFull$ & mapped data vector, query vector, query workload matrix, strategy matrix
        over a partition of $dom(\Schema)$ \\
        
        $\alpha, \beta$ & accuracy parameters for $\WorkloadRaw$ \\

        $\TotalPrivacyBudget, \ConsumedPrivacyBudget, \epsilon$ & total budget, consumed budget, workload budget \\

        $\StrategyMatrixGlobal, \Cache_{\StrategyMatrixGlobal}$
        & global strategy matrix, its cache over $dom(\Schema)$\\
        \eat{
        $\StrategyFullGlobal, \Cache_{\StrategyFullGlobal}$
        & global strategy matrix, its cache over a partition of $dom(\Schema)$ \\
        }
        
        $b, \NoisyResponse$ & a scalar noise parameter, a scalar noisy response  \\ 
        
        $\BList$ & a vector of noise parameters \\ 
        
        $\NoisyStrategyResponse, \NoisyWorkloadResponse$ & a vector of noisy responses to the strategy $\StrategyFull$~ or \WorkloadFull. \\

        ($\mathbb{a}$, $b$, $\NoisyResponse$, $t$) & a cache entry for a strategy query $\mathbb{a} \in {\StrategyMatrixGlobal}$ stored at timestamp $t$. See Definition~\ref{defn:cache}. \\
        
        $\StrategyMatrixFree,\StrategyMatrixPaid$ & free strategy matrix, paid strategy matrix \\

        \hline
    \end{tabular} 
\end{table}

Li et al.~\cite{li2015matrix} present the matrix mechanism, which first applies a DP mechanism, $M$,  on a new strategy matrix $\StrategyFull$, and then post-processes the noisy answers to the queries in $\StrategyFull$ to estimate the queries in $\WorkloadFull$.  This mechanism aims to achieve a smaller error than directly applying the mechanism $M$ on $\WorkloadFull$. We will use the Laplace mechanism $\LaplaceMechanismNew$ to illustrate matrix mechanism. 
\begin{defn}[Matrix Mechanism (MM)~\cite{li2015matrix}] \label{def:mm}
Given an $l\times n$ workload matrix $\WorkloadFull$, a $p\times n$ strategy matrix $\StrategyFull$, and the Laplace mechanism $\LaplaceMechanismNew(\StrategyFull,\DomainVector)$ that answers $\StrategyFull$ on $\DomainVector$, the matrix mechanism $\mathcal{M}_{\StrategyFull,\LaplaceMechanismNew}$ outputs the following answer: $\mathcal{M}_{\StrategyFull,\LaplaceMechanismNew}(\WorkloadFull,\DomainVector) = \WorkloadFull\StrategyFull^+\LaplaceMechanismNew(\StrategyFull,\DomainVector)$ is the Moore-Penrose pseudoinverse of $\StrategyFull$.
\end{defn}
Intuitively, each workload query in $\WorkloadFull$ can be represented as a linear combination of strategy queries in $\StrategyFull$, i.e., $\WorkloadFull\DomainVector = \WorkloadFull\StrategyFull^+(\StrategyFull\DomainVector)$. 
We denote $\LaplaceMechanismNew(\StrategyFull,\DomainVector)$ by $\NoisyStrategyResponse$ and $\mathcal{M}_{\StrategyFull,\LaplaceMechanismNew}$ by $\NoisyWorkloadResponse$.
As the matrix mechanism post-processes the output of a DP mechanism~\cite{privacy_book}, it also satisfies the same level of privacy guarantee. 

\begin{prop}[\cite{li2015matrix}] \label{prop:mmdp}
If $b\geq \frac{\|\StrategyFull\|_1}{\epsilon}$, then $\mathcal{M}_{\StrategyFull,\LaplaceMechanismNew}$ satisfies $\epsilon$-DP.
\end{prop}

For data analysts who may not be able to choose an appropriate budget for a DP mechanism, we would like to allow them to specify their accuracy requirements for their queries. We consider two popular types of error specification for DP mechanisms. 

\begin{defn} Given a $l\times n$ workload matrix $\WorkloadFull$ and a DP mechanism $M$, (i) the $\alpha^2$-expected total squared error bound~\cite{li2015matrix} is 
\begin{equation}
    \mathbb{E}[\|\WorkloadFull\DomainVector - M(\WorkloadFull,\DomainVector)\|_2^2] \leq \alpha^2
\end{equation}
and 
(ii) the $(\alpha,\beta)$-worst error bound~\cite{ge_apex:_2017} is defined as  \begin{equation}
    \Pr[\|\WorkloadFull\DomainVector - M(\WorkloadFull,\DomainVector)\|_{\infty} \geq \alpha ] \leq \beta.
\end{equation}
\end{defn}


The error for the matrix mechanism is  $\|\WorkloadFull\StrategyFull^+ Lap(b)^l\|$, which is
independent of the data. This
allows a direct estimation of the error bound without running the algorithm on the data. For example, 
Ge et al.~\cite{ge_apex:_2017} 
provide a loose bound for the noise parameter in the matrix mechanism to achieve an $(\alpha, \beta)$-worst error bound.

\begin{thm}[\cite{ge_apex:_2017}]\label{thm:loosebound}
The matrix mechanism $\mathcal{M}_{\StrategyFull,\LaplaceMechanismNew}$ satisfies the $(\alpha, \beta)$-worst error bound, if 
\begin{equation}
    \label{eq:mm-loose-noise-param-bound}
    b \leq \LooseNoiseParameter = \frac{\alpha\sqrt{\beta/2}}{\|\WorkloadFull \StrategyFull^+\|_F} 
\end{equation}
where $\| \cdot \|_F$ is the Frobenius norm. 
\end{thm}
When we set $b$ to this loose bound $\LooseNoiseParameter$, the privacy budget consumed by this mechanism is $\frac{\|\StrategyFull\|_1}{\LooseNoiseParameter}$. To minimize the privacy cost, Ge et al.~\cite{ge_apex:_2017}  conduct a continuous binary search over noise parameters larger than $\LooseNoiseParameter$. The filtering condition for this search is the output of a Monte Carlo (MC) simulation for the error term $\|\WorkloadFull\StrategyFull^+ Lap(b)^l\|_{\infty}$ (i.e., if the sampled error exceeds $\alpha$ with a probability $\leq \beta$). 


\section{System Design}\label{sec:system_design}
\begin{techreport}
\begin{figure}
    \centering
    \includegraphics[scale=0.45]{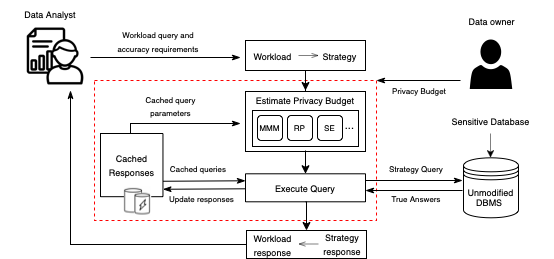}
    \caption{System diagram.}
    \label{fig:system-diagram}
\end{figure}
\end{techreport}

We design an interactive inference engine with a built-in cache, \sysname, that supports data analysts in answering data exploration queries with sufficient accuracy, without requiring them to have any differential privacy knowledge. 
\begin{techreport}
The system architecture is given in Figure~\ref{fig:system-diagram}.
\end{techreport}
The data owner instantiates an unmodified relational DBMS such as MySQL, with a database that includes sensitive data. 
To complete the setup stage, the data owner also provides a total privacy budget $\TotalPrivacyBudget$ to our system. 
At runtime, the data analyst inputs a workload query $\WorkloadRaw$, and an $(\alpha, \beta)$ accuracy requirement that the query should satisfy, to \sysname. 
Our system interacts with the DBMS, via an SQL interface, and a cache $\Cache$, to return a differentially private workload response $\NoisyWorkloadResponse$, 
which satisfies this accuracy requirement, to the analyst. 
Each workload response consumes a privacy budget $\epsilon$, out of \TotalPrivacyBudget, and the goal of \sysname~ is to reduce $\epsilon$ by using our cache, which stores historical noisy responses. 
We provide an overview of our system design in this section, while motivating our description through design challenges. 
Our system follows a \emph{modular design}, in order to enable DP experts to develop new cache-aware, problem-specific modules in the future. 

\subsection{Cache Structure Overview}
\label{subsec:cache-structure}
Our cache stores previously released noisy DP responses and related parameters; it does not store any sensitive ground truth data. Moreover, the cache does not interact directly with the DBMS at all. Therefore, the cache design evolves independently of the DBMS or other alternative data storage systems. We consider two design questions: (i) which queries and their noisy responses should be stored in the cache; and (ii) what other parameters are needed? 

A naive cache design simply stores all historical workloads, their accuracy requirements and noisy responses $[(\WorkloadRaw_1, \alpha_1, \beta_1, \NoisyWorkloadResponse_1),$ $\ldots, (\WorkloadRaw_t,\alpha_t, \beta_t, \NoisyWorkloadResponse_t)]$. 
When a new workload $(\WorkloadRaw_{t+1},\alpha_{t+1}, \beta_{t+1})$ comes in, the system first infers  a response $\NoisyWorkloadResponse'_{t+1}$ from the cache and its error bound $\alpha'_{t+1}$. If its error bound is worse than the accuracy requirement, i.e.,  $\alpha'_{t+1}\geq \alpha_{t+1}$, then additional privacy budget $\epsilon_{t+1}$ needs to be spent to improve $\NoisyWorkloadResponse'_{t+1}$ to $\NoisyWorkloadResponse_{t+1}$. This additional privacy cost $\epsilon_{t+1}$ should be smaller than a DP mechanism that does not use historical query answers. 

This cache design is used in Pioneer~\cite{pioneer}, but it has several drawbacks.  First, this design results in a cache size that linearly increases with the number of workload queries. 
Second, we will not be able to \emph{compose and reuse} cached past responses to overlapping workloads ($\WorkloadRaw_{t-k} \cap \WorkloadRaw_{t} \ne \emptyset$).
Simply put, this design works with only simple DP mechanisms, which answer the data analyst-supplied workloads directly with noisy responses.
For instance, Pioneer~\cite{pioneer} considers only single query workloads and the Laplace mechanism. 
We seek to design a reusable cache that can work with complex DP mechanisms, and in particular, the matrix mechanism. 
Thus, we need to \emph{structure} our cache such that cached queries and their noisy responses can be reused efficiently, 
in terms of the additional privacy cost and run time, 
while limiting the cache size. 

\eat{
In our cache design, we consider a set of queries that are answered directly with a noise-addition DP mechanism. 
For instance, for the matrix mechanism, we only store the set of queries in the strategy matrix $\StrategyRaw$, and their noisy responses, instead of the workload matrix $\WorkloadRaw$. 
Each of these queries in the strategy matrix  $\mathbb{a}\in \StrategyRaw$ is only associated with a single noise parameter $b$. 
This design leads to a cache $\Cache=[(\StrategyRaw_1,b_1,\NoisyResponseList_1),\ldots, (\StrategyRaw_t,b_t, \NoisyResponseList_t)]$.
If all the strategy matrices share a similar structure, in other words, many similar queries, then we can track only a limited set of queries and their ``best'' noise parameters and noisy responses. 
}

Our key insight is that the strategy matrices in Matrix Mechanism (MM) in Def~\ref{def:mm} can be chosen from a structured set. 
So, we store noisy responses to the matrix that the mechanism answers directly (the strategy matrix),
instead of storing noisy responses that are post-processed and returned to the data analyst (the workload matrix).
\eat{
That is, we store noisy responses to the matrix that the mechanism answers through the addition of DP noise (the strategy matrix in MM) \xh{add ``directly without post processing''}, since it \xh{change ``it'' to ``the strategy matrix''} can be structured appropriately, instead of storing noisy responses that are returned to the data analyst (the workload matrix). 
}
If all the strategy matrices share a similar structure, in other words, many similar queries, then we need to only track a limited set of queries in our cache. 
Relatedly, since the $(\alpha, \beta)$ accuracy requirements for different workload matrices can only be composed through a loose union bound, we instead track the noise parameters that are used to answer the associated strategy matrices. 
Thus, in our cache, we store the strategy queries, the noisy strategy query responses and the noise parameter. \eat{$\Cache=[(\StrategyRaw_1,\BList_1,\NoisyResponseList_1),\ldots, (\StrategyRaw_t,\BList_t, \NoisyResponseList_t)]$.}

This cache design motivates us to consider a global strategy matrix $\StrategyMatrixGlobal$ for the cache that can support all possible workloads.
Importantly, for a given workload matrix $\WorkloadRaw$, we present a strategy transformer (ST) module to generate an instant strategy matrix, denoted by $\StrategyRaw$, such that each instant strategy matrix is contained in the global strategy matrix,  i.e., $\StrategyRaw \subseteq \StrategyMatrixGlobal$. 
\eat{In this design, the cache only tracks the latest (and hence the most accurate) noisy response to each answered query in $\StrategyMatrixGlobal$ with its noise parameter, and the timestamp for the update. }
In this design, the cache tracks each strategy entry $\mathbb{a} \in \StrategyMatrixGlobal$, with its noisy response, its noise parameter, and the timestamp.

\begin{defn}[Cache Structure]\label{defn:cache}
Given a global strategy matrix  $\StrategyMatrixGlobal$ over the full domain $dom(\Schema)$,  a cache for differentially private counting queries is defined  as 
\begin{equation}
    \Cache_{\StrategyMatrixGlobal} = \{\ldots,(\mathbb{a},b,\NoisyResponse,t),\ldots | \mathbb{a}\in \StrategyMatrixGlobal \},
\end{equation}
where $b$ and $\NoisyResponse$ are the latest noise parameter and noisy response for the strategy query $\mathbb{a}$, and $t$ is the time stamp for the latest update of $\mathbb{a}$. At beginning, all entries are initialized as $(\mathbb{a},-,-,0)$, where `$-$' denotes invalid values. We use $\Cache$ to represent the set of entries with valid noisy responses and $t>0$. 
\end{defn}

In this work, we consider a hierarchical structure, or $k$-ary tree, for $\StrategyMatrixGlobal$, which is a popular and effective strategy matrix for MM~\cite{li2015matrix} with an expected worst error of $O(\log^3n)$, where $n$ is the domain size. 
Figure~\ref{fig:binary-tree-strategy-gen-proactive} shows the global strategy matrix as a binary tree decomposition of a small integer domain $[0,8)$.

\subsection{Strategy Transformer (ST) Overview}
\label{subsec:strategy-transformer-overview}
We outline the Strategy Transformer (ST) module, which is commonly used by all of our cache-aware DP modules. 
The ST module consists of two components: a Strategy Generator (SG) and a Full-rank Transformer (FT).
Prior work~\cite{li2015matrix} uses the global strategy $\StrategyMatrixGlobal$, which has a high $\|\StrategyMatrixGlobal\|_1$.
Given an input $\WorkloadRaw$, the SG selects a basic instant strategy $\StrategyRaw \subseteq \StrategyMatrixGlobal$, with a low  $\|\StrategyMatrixGlobal\|_1$, among other criteria. 
\eat{
The guidelines for selecting an $\StrategyRaw$ out of $\StrategyMatrixGlobal$ are that: (i) $\StrategyRaw$ should support $\WorkloadRaw$, (ii) it should have a small $\|\StrategyRaw\|_1$ and hence a small privacy cost (based on Proposition~\ref{prop:mmdp}). }
Though the cache $\Cache_{\StrategyMatrixGlobal}$ is structured based on $\StrategyMatrixGlobal$, the cache is not searched while generating $\StrategyRaw$. 
We present two example workloads and the instant strategies generated for these workloads next. 

\begin{techreport}
\begin{figure}[t!]
    \centering
    \caption{A global strategy $\StrategyMatrixGlobal$ in a binary tree decomposition for an integer domain $[0,8)$.     Workload queries include $\WorkloadRaw_1 = \{[0,7)\}$, $\WorkloadRaw_2 = \{[2,6), [3,7)\}$.
    Nodes present only in strategy $\StrategyRaw_1$ are shown in blue text, nodes only in $\StrategyRaw_2$ are in magenta text, nodes in both $\StrategyRaw_1$ and $\StrategyRaw_2$ are shown in purple text. 
    The dashed nodes are output by the Proactive Querying (PQ) module (Section~\ref{subsec:proactive}), for $\StrategyRaw_2$; the value annotations for $r$ and $s$ refer to Algorithm~\ref{algo:proactive}. 
    }  \label{fig:binary-tree-strategy-gen-proactive}
    \resizebox{\dimexpr\linewidth}{!}{%
    \begin{tikzpicture}[->,>=stealth',level/.style={sibling distance = 5.0cm/#1,
      level distance = 1.0cm},level 3/.style={sibling distance = 1.0cm,
      level distance = 1.0cm}] 
    \node [normalnode] {
        \shortstack[c]{[0,8) \\ \tiny{$r:2$} \\ \tiny{$s:2$} }
         }
        child{ node [strategy1node,] {\shortstack[c]{[0,4) \\ \tiny{$r:2$} \\ \tiny{$s:2$} }} 
                child{ node [proactive2node] {\shortstack[c]{[0,2) \\ \tiny{$r:2$} \\ \tiny{$s:0$} }} 
                	child{ node [proactive2node] {\shortstack[c]{[0,1) \\ \tiny{$r:1$} \\ \tiny{$s:0$} }} }
    				child{ node [proactive2node] {\shortstack[c]{[1,2) \\ \tiny{$r:1$} \\ \tiny{$s:0$} }} }
                }
               child{ node [strategy2node] {\shortstack[c]{[2,4) \\ \tiny{$r:2$} \\ \tiny{$s:2$} }} 
                	child{ node [proactive2node] {\shortstack[c]{[2,3) \\ \tiny{$r:1$} \\ \tiny{$s:0$} }} }
    				child{ node [strategy2node] {\shortstack[c]{[3,4) \\ \tiny{$r:1$} \\ \tiny{$s:1$} }} }
                }                        
        }
        child{ node [proactive2node] {\shortstack[c]{[4,8) \\ \tiny{$r:2$} \\ \tiny{$s:1$} }}
                child{ node [strategy1and2node] {\shortstack[c]{[4,6) \\ \tiny{$r:1$} \\ \tiny{$s:1$} }} 
                	child{ node [normalnode] {\shortstack[c]{[4,5) \\ \tiny{$r:0$} \\ \tiny{$s:0$} }} }
    				child{ node [normalnode] {\shortstack[c]{[5,6) \\ \tiny{$r:0$} \\ \tiny{$s:0$} }} }
                }
               child{ node [normalnode] {\shortstack[c]{[6,8) \\ \tiny{$r:1$} \\ \tiny{$s:1$} }} 
                	child{ node [strategy1and2node] {\shortstack[c]{[6,7) \\ \tiny{$r:1$} \\ \tiny{$s:1$} }} }
    				child{ node [normalnode] {\shortstack[c]{[7,8) \\ \tiny{$r:1$} \\ \tiny{$s:0$} }} }
                }                        
        }
    ; 
    \end{tikzpicture}}
\end{figure}
\end{techreport}

\begin{vldbpaper}
\begin{figure}[t!]
    \centering
    \caption{A global strategy $\StrategyMatrixGlobal$ in a binary tree decomposition for an integer domain $[0,8)$.     
    Workloads include $\WorkloadRaw_1 = \{[0,7)\}$, $\WorkloadRaw_2 = \{[2,6), [3,7)\}$.
    Strategy nodes unique to $\StrategyRaw_1$ are in blue text, those unique to $\StrategyRaw_2$ are in magenta text, whereas those in both $\StrategyRaw_1$ and $\StrategyRaw_2$ are in purple text.
    The dashed nodes are output by the PQ module (Section~\ref{subsec:proactive}), for $\StrategyRaw_2$.
    }  \label{fig:binary-tree-strategy-gen-proactive}
    \resizebox{\dimexpr\linewidth}{!}{%
    \begin{tikzpicture}[->,>=stealth',level/.style={sibling distance = 5.0cm/#1,
      level distance = 1.0cm},level 3/.style={sibling distance = 1.0cm,
      level distance = 1.0cm}] 
    \node [normalnode] {
        \shortstack[c]{[0,8)  }
         }
        child{ node [strategy1node,] {\shortstack[c]{[0,4) }} 
                child{ node [proactive2node] {\shortstack[c]{[0,2)  }} 
                	child{ node [proactive2node] {\shortstack[c]{[0,1) }} }
    				child{ node [proactive2node] {\shortstack[c]{[1,2) }} }
                }
               child{ node [strategy2node] {\shortstack[c]{[2,4)  }} 
                	child{ node [proactive2node] {\shortstack[c]{[2,3)  }} }
    				child{ node [strategy2node] {\shortstack[c]{[3,4)  }} }
                }                        
        }
        child{ node [proactive2node] {\shortstack[c]{[4,8)  }}
                child{ node [strategy1and2node] {\shortstack[c]{[4,6) }} 
                	child{ node [normalnode] {\shortstack[c]{[4,5)  }} }
    				child{ node [normalnode] {\shortstack[c]{[5,6)  }} }
                }
               child{ node [normalnode] {\shortstack[c]{[6,8)  }} 
                	child{ node [strategy1and2node] {\shortstack[c]{[6,7) }} }
    				child{ node [normalnode] {\shortstack[c]{[7,8)  }} }
                }                        
        }
    ; 
    \end{tikzpicture}}
\end{figure}
\end{vldbpaper}

\begin{example} \label{eg:globalstrategy}
    In Figure~\ref{fig:binary-tree-strategy-gen-proactive}, for an integer domain $[0,8)$, we show a binary tree decomposition for its global strategy $\StrategyMatrixGlobal$. This strategy 
    consists of $(2^3+2^2+2^1+1)$
    row counting queries (RCQs), where each RCQ corresponds to the counting query with the predicate range indicated by a node in the tree.  We use $\StrategyMatrixGlobal_{[a,b)}$ to denote the RCQ with a range $[a,b)$ in the global strategy matrix. 

    The first workload $\WorkloadRaw_1$ consists of a single query with a range predicate $[0,7)$. 
    Its answer can be composed by summing over noisy responses to three RCQs in the global strategy matrix, ($\StrategyMatrixGlobal_{[0,4)}$, $\StrategyMatrixGlobal_{[4,6)}$, $\StrategyMatrixGlobal_{[6,7)}$). 
    The second workload $\WorkloadRaw_2$ has two queries with range predicates ($[2,6),[3,7)$). 
    It can be answered using $\StrategyRaw_2$ = ($\StrategyMatrixGlobal_{[2,4)}$, $\StrategyMatrixGlobal_{[4,6)}$, $\StrategyMatrixGlobal_{[3,4)}$, $\StrategyMatrixGlobal_{[6,7)}$).
    We detail the strategy generation in Example~\ref{eg:basicstrategy}. 

    We observe that the RCQs $\StrategyMatrixGlobal_{[4,6)}$ and $\StrategyMatrixGlobal_{[6,7)}$ are common to both  $\StrategyRaw_1$ and  $\StrategyRaw_2$, thus our cache-aware DP mechanisms can potentially reuse their noisy responses to answer $\StrategyRaw_2$.
\qedsymbol
\end{example}

The accuracy analysis of the matrix mechanism only holds over full rank strategy matrices, however, the instant strategy $\StrategyRaw$ may be a very sparse matrix over the full domain, and thus, may not be full rank.
\eat{
The instant strategy $\StrategyRaw$ may be a very sparse  matrix over the full domain, and it may not be full rank. The matrix mechanism requires it to be full rank to optimize the accuracy
~\cite[Section 4]{li2015matrix}. 
}
We address this challenge in the FRT module, 
by mapping the instant strategy $\StrategyRaw$, workload $\WorkloadRaw$, data vector $\DataVectorRaw$, 
to a compact, full-rank, efficient representation, resulting in $\WorkloadFull$, $\StrategyFull$, and $\DataVector$ respectively. 
Thus for an input $\WorkloadRaw, \DataVectorRaw$, the ST module outputs $(\StrategyFull, \WorkloadFull, \DataVector)$. 
Since the cache entries should be uniquely addressable, the raw data vector $\DataVectorRaw$ and strategy $\StrategyRaw$ are used to index the cache. 

\subsection{Cache-aware DP Modules}
\label{subsec:module-design}
Our system supports two novel classes of cache-aware DP mechanisms:  \emph{Modified Matrix Mechanism} (MMM) and the \emph{Relax Privacy} Mechanism (RP). 
These two cache-aware DP mechanisms commonly use the ST module to transform an input $\WorkloadRaw, \DataVectorRaw$ to $(\WorkloadFull, \StrategyFull, \DataVector)$. 
Each cache-aware DP mechanism implements two interfaces (similar to APEx~\cite{ge_apex:_2017}) using the mapped representations $\StrategyFull, \WorkloadFull, \DataVector$, as well as the cache $\Cache_{\StrategyMatrixGlobal}$:

\squishlist 
    \item {
    The \emph{\textsc{answerWorkload}} interface answers a workload $\WorkloadFull$ using the cache $\Cache_{\StrategyMatrixGlobal}$ and an instant strategy $\StrategyFull$ to derive fresh noisy strategy responses, using the ground truth from the DB. 
    Each implementation of this interface also updates the cache $\Cache_{\StrategyMatrixGlobal}$.  
    }
    \item {
    The \emph{\textsc{estimatePrivacyBudget}} interface estimates the minimum privacy budget $\epsilon$ required by the \textsc{answerWorkload} interface to achieve 
    the $(\alpha, \beta)$ accuracy requirement. 
    }
\squishend

For the first cache-aware DP mechanism, MMM, we have two additional optional modules, namely \emph{Strategy Expander} (SE) and \emph{Proactive Querying} (PQ), which modify the instant strategy $\StrategyRaw$ output by the basic ST module, for different purposes.
The SE module expands the basic  $\StrategyRaw$
with related, cached, accurate strategy rows in $\Cache_{\StrategyMatrixGlobal}$ to exploit constrained inference as discussed by Hay et al.~\cite{hay2010boosting}. 
The goal of this module is to further reduce the privacy cost of the basic instant strategy to answer the given workload $\WorkloadRaw$. 
\eat{
On the other hand, the PQ module is designed to fill the cache, without incurring any additional privacy budget over the MMM module. 
It expands $\StrategyRaw$ with disjoint strategy queries that are absent from $\Cache_{\StrategyMatrixGlobal}$. 
Therefore, it reduces the privacy cost of future queries.}
On the other hand, the PQ module is designed to fill the cache proactively, for later use by the MMM, MMM+SE, and RP mechanisms. 
It expands $\StrategyRaw$ with strategy queries that are absent from $\Cache_{\StrategyMatrixGlobal}$, without incurring any additional privacy budget over the MMM module. 
Therefore, it reduces the privacy cost of future workload queries.


 \begin{algorithm}[t]
     \caption{\sysname Overview} 
     \label{algo:end-to-end}
    \begin{algorithmic}[1]
         \Require  Dataset $D$, Total privacy budget \TotalPrivacyBudget.
         \State Initialize privacy loss $\ConsumedPrivacyBudget=0$,
         cache $\Cache_{\StrategyMatrixGlobal} = \{(\mathbb{w},-,-,0) |\mathbb{w}\in \StrategyMatrixGlobal \}$
         \Repeat
             \State Receive $(\WorkloadQuery, \alpha, \beta)$ from analyst
             \State $\WorkloadMatrix$ $\gets$ \Call{getMatrixForm}{$\WorkloadQuery, \DomainVector$}
             \State $\StrategyRaw, \StrategyFull, \WorkloadFull$ $\gets $ \Call{generateStrategy}{$\WorkloadMatrix, \StrategyMatrixGlobal$}\label{line:main_ST}
             
             \State $(b, \epsilon_1) \gets$ \Call{MMM.estimatePrivacyBudget}{$\Cache, \StrategyFull, \WorkloadFull, \alpha,\beta$}\label{line:main_est_mmm}

            \State $\epsilon_2 \gets$ \Call{RP.estimatePrivacyBudget}{$\Cache,\StrategyFull, \WorkloadFull,\alpha,\beta$}

            \State $\StrategyRaw_e, \StrategyFull_e, \gets$ 
            \Call{SE.generateExpandedStrategy}{$\StrategyRaw,\Cache,b$} 
            
              \State $\epsilon_3 \gets$ \Call{MMM.estimatePrivacyBudget}{$\Cache, \StrategyFull_e, \WorkloadFull, \alpha,\beta$}\label{line:main_est_se}
            
            
             \State Pick $(\hat{M}, \hat{\StrategyFull})$ from (\textsc{MMM/RP}, $\StrategyFull/\StrategyFull_e$) that has smallest $\epsilon_i$\label{line:main_choose}

             \If{$\epsilon_i + \ConsumedPrivacyBudget \ge \TotalPrivacyBudget$}
                \State Answering \WorkloadQuery~ satisfying ($\alpha, \beta$) will exceed \TotalPrivacyBudget. Reject \WorkloadQuery.
             \EndIf

            \State  $z \gets$ \Call{$\hat{M}$.answerWorkload}{$\Cache,\hat{\StrategyFull},\WorkloadFull, \epsilon_i, \DomainVector$}\label{line:main_answer} 
            
            \State \Return $z$ to data analyst.
             \State $\ConsumedPrivacyBudget \gets \ConsumedPrivacyBudget + \epsilon_i$\label{line:main_budget_update} 
         \Until{no more $\WorkloadQuery$ from the analysts}
     \end{algorithmic}
 \end{algorithm}

\underline{Putting it all together}, we state the end-to-end algorithm in Algorithm~\ref{algo:end-to-end}. First, for an input workload $(\WorkloadFull, \alpha, \beta)$, our system first uses the ST module to generate a full-rank instant strategy matrix $\StrategyFull$ (line~\ref{line:main_ST}), and then executes the \textsc{estimatePrivacyBudget} interface, with the input tuple $(\WorkloadFull, \StrategyFull, \alpha, \beta)$, for the MMM, MMM+SE, and RP mechanisms (line~\ref{line:main_est_mmm}-\ref{line:main_est_se}). 
We choose the mechanism that returns the lowest privacy cost $\epsilon_i$ (line~\ref{line:main_choose}).
If the sum of this privacy cost with the consumed privacy budget is 
smaller than the total privacy budget, then the system executes the \textsc{answerWorkload} interface for the chosen mechanism, with the input tuple $(\WorkloadFull, \hat{\StrategyFull}, \epsilon_i)$ (line~\ref{line:main_answer}). 
The consumed privacy budget will increase by $\epsilon_i$ (line~\ref{line:main_budget_update}). 
(The PQ module does not impact the cost estimation for MMM, it only extends the strategy matrix $\StrategyFull$ to be answered.)
We present the MMM in Section~\ref{sec:mmm-module}, the common ST module and the MMM optional modules (SE, PQ) in Section~\ref{sec:strategy-module}, and the RP mechanism in Section~\ref{sec:rp-module}.
\begin{theorem}\label{thm:end-to-end_privacy}
\sysname, as defined in Algorithm~\ref{algo:end-to-end}, satisfies \TotalPrivacyBudget-DP.
\end{theorem}

\section{Modified Matrix Mechanism (MMM)}\label{sec:mmm-module}
In this section, we focus on our core cache-aware DP mechanism, namely the Modified Matrix Mechanism (MMM). We would like to 
answer a workload $\WorkloadFull$ with an $(\alpha,\beta)$-accuracy requirement using a given cache $\Cache_{\StrategyFullGlobal}$ and an instant strategy $\StrategyFull \subseteq \StrategyFullGlobal$, while minimizing the privacy cost. We will first provide intuition for the design of this mechanism. Then, we will describe the first interface \textsc{answerWorkload}  that answers a workload $\WorkloadFull$ using the instant strategy $\StrategyFull$ with the best set of parameters derived from the second interface \textsc{EstimatePrivacyBudget}. We then present how the the second interface arrives at an optimal privacy budget.

\subsection{MMM Overview}


The cacheless matrix mechanism  (Definition~\ref{def:mm}) perturbs the ground truth response to the strategy, that is $\StrategyFull \DomainVector$, with the noise vector freshly drawn from  $Lap(b)^{|\StrategyFull|}$ to obtain $\NoisyResponseList = \StrategyFull\DomainVector+Lap(b)^{ |\StrategyFull|}$. 
An input workload is then answered using $\WorkloadFull \StrategyFull^+ \NoisyResponseList$. 
As we discussed in the background, in an accuracy-aware DP system such as APEx~\cite{ge_apex:_2017}, the noise parameter $b$ is calibrated, first through a loose bound $\LooseNoiseParameter$ and then to a tighter noise parameter $\TightNoiseParameter$, such that the workload response above meets the $(\alpha,\beta)$-accuracy requirement.
This spends a privacy budget $\frac{\|\StrategyFull\|_1}{\TightNoiseParameter}$ (Proposition~\ref{prop:mmdp}).

In MMM, we seek to reduce the privacy budget spent by using the cache $\Cache$. 
Given an instant strategy matrix $\StrategyFull \subseteq \StrategyFullGlobal$, we first lookup the cache for any rows in the strategy matrix $\StrategyFull$. 
Note that not all rows in $\StrategyFull$ have their noisy responses in the cache. 
The cache may contain noisy responses for some rows of $\StrategyFull$, given by $\Cache \cap \StrategyFull$, whereas other rows in $\StrategyFull$ may not have cached responses.  
A preliminary approach would be to simply reuse all cached strategy responses, and obtain noisy responses for non-cached strategy rows by expending some privacy budget through naive MM. 
However, some cached responses may be too noisy and thus including them will lead to a higher privacy cost than the cacheless MM.

Our key insight is that by reusing noisy responses for \emph{accurately cached} strategy rows, MMM can ultimately use a smaller privacy budget for all other strategy rows as compared to MM without cache 
while satisfying the accuracy requirements.
Thus, out of all cached strategy rows $\Cache\cap \StrategyFull$, MMM identifies a subset of accurately cached strategy rows $\StrategyMatrixFree\subseteq \Cache\cap \StrategyFull$ that can be directly answered using their cached noisy responses, without spending any privacy budget. 
MMM only spends privacy budget on the remaining strategy rows, namely on $\StrategyMatrixPaid = \StrategyFull-\StrategyMatrixFree$. 
We refer to $\StrategyMatrixFree$ and $\StrategyMatrixPaid$ as the \emph{free strategy matrix} and the \emph{paid strategy matrix} respectively. 
MMM consists of two interfaces as indicated by Algorithm~\ref{algo:mmm}: (i) \textsc{answerWorkload}  and (ii) \textsc{estimatePrivacyBudget}. 
The second interface seeks the best pair of free and paid strategy matrices $(\StrategyMatrixFree,\StrategyMatrixPaid)$ that use the smallest privacy budget $\epsilon$ to achieve $(\alpha,\beta)$-accuracy requirement. 
The first interface will make use of this parameter configuration $(\StrategyMatrixFree,\StrategyMatrixPaid, \epsilon)$ to generate noisy responses to the workload.

\begin{algorithm}[t]
    \caption{MMM main interfaces and supporting functions}
    \label{algo:mmm}
    {
    \small \begin{algorithmic}[1]
        
        \Function{answerWorkload}{
        $\Cache, \StrategyFull, 
        \WorkloadFull, \epsilon, \DataVector$}  
        
        \State 
        $(\StrategyMatrixFree,\StrategyMatrixPaid,b_{\StrategyMatrixPaid}, \epsilon)$ from pre-run  \Call{EstimatePrivacyBudget}{$\Cache, \StrategyFull, 
        \WorkloadFull, \alpha,\beta$}  
        \label{line:answer-mmm-getparameter}

        \State (Optional) Expand $\StrategyMatrixPaid$ with PQ module (Section~\ref{subsec:proactive}) 
        

        \State $\NoisyResponseList_{\StrategyMatrixPaid} \gets \StrategyMatrixPaid \DomainVector + Lap(b_{\StrategyMatrixPaid})^{|\StrategyMatrixPaid|}$ \Comment{we have $b_{\StrategyMatrixPaid} = \frac{\|\StrategyMatrixPaid\|_1}{\epsilon}$} \label{line:answer-mmm-get-paid-responses}
        
        \State Update cache $\Cache_{\StrategyMatrixGlobal}$ with $(\StrategyMatrixPaid, b_{\StrategyMatrixPaid}, \NoisyResponseList_{\StrategyMatrixPaid},t=\text{current time})$
        \label{line:answer-mmm-insert-to-cache}

		 \State $\NoisyResponseList_{\StrategyMatrixFree} \gets  [(\mathbf{w},b,\NoisyResponse,t) \in \Cache | \mathbf{w} \in \StrategyMatrixFree]$   \Comment{free cached responses for $\StrategyMatrixFree$}  \label{line:answer-mmm-get-cached-responses}
        
		\State $\NoisyResponseList \gets \NoisyResponseList_{\StrategyMatrixFree} \| \NoisyResponseList_{\StrategyMatrixPaid}$ \Comment{concatenate noisy responses for $\StrategyFull$.} \label{line:answer-mmm-return-full-response}
            \State \Return $\WorkloadFull\StrategyFull^+\NoisyResponseList$, $\epsilon$  
        \EndFunction
        \Statex

        \Function{EstimatePrivacyBudget}{$\Cache, \StrategyFull, 
        \WorkloadFull, \alpha,\beta$}
            \State Set upper bound 
            $b_{\top} = \frac{\|\StrategyFull\|_1}{\PrivacyBudgetPrecision}$ \Comment{$\PrivacyBudgetPrecision$ is the budget precision }
 \label{line:answer-mmm-upperbound}
 
            \State Set loose bound $\LooseNoiseParameter = \frac{\alpha\sqrt{\beta/2}}{\|\WorkloadFull \StrategyFull^+\|_F} $ \Comment{Theorem~\ref{thm:loosebound} (without cache)}
            \label{line:answer-mmm-loosebound}

            \State $\BList \gets 
            [(\mathbf{w},b,\NoisyResponse,t) \in \Cache  ~|~ \mathbf{w}\in \StrategyFull \cap \Cache, b>\LooseNoiseParameter] 
            \cup [b_L]$ 
            \label{line:answer-mmm-discretespace}
            
            \State $b_D \gets$ 
            \textsc{binarySearch}(sort($\BList$),
             \Call{checkAccuracy}{$\cdot, \Cache, \StrategyFull, \WorkloadFull, \alpha,\beta$}) 
             \Comment{Search $b_D$ in the discrete space}
        \label{line:answer-mmm-discretesearch}
        
        \State $\StrategyMatrixFree \gets [c.\mathbf{a} \in \Cache  ~|~ c.\mathbf{a} \in \StrategyFull \cap \Cache, c.b<b_{\StrategyMatrixPaid}]$ 
        and 
        $\StrategyMatrixPaid \gets \StrategyFull- \StrategyMatrixFree$ 
        
            \State $b_{\StrategyMatrixPaid} \gets$ 
            \textsc{binarySearch}($[b_D, b_{\top}]$,
             \Call{checkAccuracy}{$\cdot, \Cache, \StrategyFull, \WorkloadFull, \alpha,\beta$}) 
             \Comment{Search $b_{\StrategyMatrixPaid}$ in a continuous space}
      \label{line:answer-mmm-contsearch}

            \State \Return ( $\StrategyMatrixFree, \StrategyMatrixPaid, b_{\StrategyMatrixPaid}, \frac{\|P\|_1}{b_{\StrategyMatrixPaid}}$)
        \EndFunction
                
\begin{techreport}
         \Statex
        \Function{checkAccuracy}{$b_{\StrategyMatrixPaid}, \Cache, \StrategyFull, \WorkloadFull, \alpha,\beta$}
        \State $(\StrategyMatrixFree, \BList_{\StrategyMatrixFree}) \gets [(\mathbf{w},b,\NoisyResponse,t) \in \Cache  ~|~ \mathbf{w} \in \StrategyFull \cap \Cache,  b<b_{\StrategyMatrixPaid}]$ 
        and 
        $\StrategyMatrixPaid \gets \StrategyFull- \StrategyMatrixFree$ 
       \label{line:answer-mmm-setF}
        
        \State Sample size $N=10000$ and failure counter $n_f = 0$  
        \label{line:answer-mmm-mcstart}

        \For{$i=1,\ldots,N$}
            \State $n_f$++ if $\|\WorkloadFull\StrategyFull^+ Lap(\BList_{\StrategyMatrixFree} || \BList_{\StrategyMatrixPaid})\|_{\infty} >\alpha$
            
        \EndFor
        \State $\beta_e = n_f/N$, $p=\beta/100$
        \State $\delta\beta = z_{1-p/2}\sqrt{\beta_e(1-\beta_e)/N}$ 
        \State \Return $(\beta_e+\delta\beta+p/2)<\beta$
        \label{line:answer-mmm-mcend}
 
        \EndFunction
        \end{techreport}
        \end{algorithmic}}
\end{algorithm}

\subsection{Answer Workload Interface}
We present the first interface \textsc{answerWorkload} for the MMM. We recall that this interface is always called after the \textsc{estimatePrivacyBudget} interface which computes the best combination of free and paid strategy matrices and their corresponding privacy budget $(\StrategyMatrixFree,\StrategyMatrixPaid, b_{\StrategyMatrixPaid}, \epsilon)$. As shown in Algorithm~\ref{algo:mmm}, the \textsc{answerWorkload} interface first calls 
the proactive module (Section~\ref{subsec:proactive}). If this module is turned on, $\StrategyMatrixPaid$ will be expanded for the remaining operations. Then this interface will answer the paid strategy matrix $\StrategyMatrixPaid$
using Laplace mechanism with the noise parameter $\PaidNoiseParameter$. We have $\PaidNoiseParameter=\frac{\|\StrategyMatrixPaid\|_1}{\epsilon}$, to ensure $\epsilon$-DP (Line~\ref{line:answer-mmm-get-paid-responses}).   Then, it updates the corresponding entries in the cache $\Cache_{\StrategyMatrixGlobal}$ (Line~\ref{line:answer-mmm-insert-to-cache}). In particular, for each query $\mathbf{w}\in \StrategyMatrixPaid$, we update its corresponding noisy parameter, noisy response, and timestamp in $\Cache_{\StrategyMatrixGlobal}$ to $b_{\StrategyMatrixPaid}, \NoisyResponse$, and the current time.   
After obtaining the fresh noisy responses $\NoisyResponseList_{\StrategyMatrixPaid}$ for the paid strategy matrix, this interface pulls the cached responses $\NoisyResponseList_{\StrategyMatrixFree}$ for the free strategy matrix from the cache and concatenate them into $\NoisyResponseList$ according to their order in the instant strategy $\StrategyFull$  (Lines~\ref{line:answer-mmm-get-cached-responses}-\ref{line:answer-mmm-return-full-response}). Finally, this interface returns a noisy response to the workload $\WorkloadFull\StrategyFull^+\NoisyResponseList$, and its privacy cost $\epsilon$. 

\begin{prop}\label{prop:privacy_mmm}
The \textsc{AnswerWorkload} interface of MMM (Algorithm~\ref{algo:mmm}) satisfies $\epsilon$-DP, where $\epsilon$ is the output of this interface. 
\end{prop}

As the final noisy response vector $\NoisyResponseList$ to the strategy $\StrategyFull$ is concatenated from 
$\NoisyResponseList_{\StrategyMatrixFree}$ and $\NoisyResponse_{\StrategyMatrixPaid}$,  its distribution is equivalent to a response vector perturbed by a vector of Laplace noise with parameters: $\BList = \BList_{\StrategyMatrixFree} || \BList_{\StrategyMatrixPaid}$, where 
$\BList_{\StrategyMatrixFree}$ is a vector of noise parameters for the cached entries in $\StrategyMatrixFree$ with length $|\StrategyMatrixFree|$ and $\BList_{\StrategyMatrixPaid}$ is a vector of the same value $b_{\StrategyMatrixPaid}$  with length $|\StrategyMatrixPaid|$.
This differs from the standard matrix mechanism with a single scalar noise parameter. 
We derive its error term next.

\begin{prop} \label{prop:mmm-error}
Given an instant strategy $\StrategyFull=(\StrategyMatrixFree||\StrategyMatrixPaid)$ with a vector of 
$k$ noise parameters $\BList=\BList_{\StrategyMatrixFree} || \BList_{\StrategyMatrixPaid}$, the error to a workload $\WorkloadFull$ using the \textsc{AnswerWorkload} interface of MMM (Algorithm~\ref{algo:mmm}) is 
\begin{equation}
\|\WorkloadFull\StrategyFull^+ Lap(\BList)\|     \end{equation}
where $Lap(\BList)$ draws independent noise from $Lap(\BList[1])$, $\ldots, Lap(\BList[k])$ respectively.
We can simplify its expected total square error as
\begin{equation}
    \|\WorkloadFull\StrategyFull^+ diag(\BList)\|_F^2
\end{equation} where $diag(\BList)$ is a diagonal matrix with $diag(\BList)[i,i] = \BList[i]$. 
\eat{
    The $\alpha^2$-expected total square error bound is 
    \begin{equation}
        \|\WorkloadFull\StrategyFull^+ diag(\BList)\|_2^2 \leq \alpha^2
    \end{equation} where $diag(\BList)$ is a diagonal matrix with $diag(\BList)[i,i] = \BList[i]$. 
    The $(\alpha,\beta)$-worst error bound is 
    \begin{equation}
        \Pr[\|\WorkloadFull\StrategyFull^+ Lap(\BList)\|_{\infty} \geq \alpha ] \leq \beta
    \end{equation}
}
\end{prop}

\subsection{Estimate Privacy Budget Interface}
The second interface \textsc{EstimatePrivacyBudget} chooses the free and paid strategy matrices and the privacy budget to run the first interface for MMM. This corresponds to the following questions:
\squishlist
    \item[(1)] Which cached strategy rows out of $\Cache\cap \StrategyFull$ should be included in the free strategy matrix $\StrategyMatrixFree$? The choice of $\StrategyMatrixFree$ directly determines the paid strategy matrix $\StrategyMatrixPaid$ as $\StrategyFull-\StrategyMatrixFree$. 
    \item[(2)]  Given $\StrategyMatrixPaid$ and $\PaidNoiseParameter$, the privacy budget paid by MMM is given by $\epsilon = \|\StrategyMatrixPaid\|_1 / \PaidNoiseParameter =\|\StrategyFull-\StrategyMatrixFree\|_1 / \PaidNoiseParameter$. To minimize this privacy budget, what is the maximum noise parameter value $\PaidNoiseParameter$ that can be used to answer $\StrategyMatrixPaid$ while meeting the accuracy requirement? 
\squishend 

A baseline approach to the first question is to simply set $\StrategyMatrixFree = \Cache\cap\StrategyFull$, that is, we reuse all cached strategy responses. 
This approach may reuse inaccurate cached responses with large noise parameters, which results in a larger $\epsilon$ (or a smaller $\PaidNoiseParameter$) to achieve the given accuracy requirement than answering the entire $\StrategyFull$ by resampling new noisy responses without using the cache. 

\begin{example}
Continuing with Example~\ref{eg:globalstrategy}, we have an instant strategy $\StrategyFull$ for the workload $\WorkloadRaw_1$ with range predicate $[0,7)$ mapped to a partitioned domain $\{[0,4), [4,6), [6,7)\}$. The mapped workload and instant strategy are shown in Figure~\ref{fig:mmm-use-inaccurate-cache-entries_simple}. 
For simplicity, we use the expected square error to illustrate the drawback of the baseline approach, but the same reasoning applies to $(\alpha,\beta)$-worst error bound.  Without using the cache, when we set $\BList = [10,10,10]$, we achieve an expected error  $\|\WorkloadFull\StrategyFull^+ diag(\BList)\|^2_F$ = 300 for the workload $\WorkloadFull$. 
Suppose the cache has an entry for the first RCQ $[0,4)$ of the strategy and a noise parameter $b_c=15$. Using this cached entry, the noise vector becomes $\BList = [15, b_{\StrategyMatrixPaid}, b_{\StrategyMatrixPaid}]$, and the expected square error is $\|\WorkloadFull\StrategyFull^+ diag(\BList)\|^2_F=15^2+2b_{\StrategyMatrixPaid}^2$.  
To achieve the same or a smaller error than the cacheless MM, we need to set  $b_{\StrategyMatrixPaid}\leq \sqrt{(300-15^2)/2} \approx 6.12$ for the remaining entries in the strategy. This tighter noise parameter $b_{\StrategyMatrixPaid}$ corresponds to a larger privacy budget. 
\qedsymbol
\end{example}
\begin{figure}
    \centering
    \[
    \WorkloadFull = 
    \begin{bmatrix}
        1 & 1  & 1
    \end{bmatrix}
    ,\ \StrategyFull = 
    \left[
    \begin{array}{cccc}
        1 & 0  & 0 \\ 
        \hline 
        0 & 1 & 0\\
        0 & 0 & 1 
    \end{array}
    \right]
     ,\ \BList = 
     \left[
    \begin{array}{c}
     b_c \\
      \hline 
      b \\
      b 
    \end{array}
    \right]
     ,\ \DomainVector_{1} = 
    \left[
    \begin{array}{c}
       \DataVectorRaw{\left[0,4\right)} \\
       \hline
        \DataVectorRaw{\left[4,6\right)} \\
         \DataVectorRaw{\left[6,7\right)} \\
    \end{array}
    \right]
    \] \vspace{-4mm}
    \caption{
    Consider $\WorkloadRaw_1 = \{[0,7)\}$ with its corresponding mapped workload matrix, instant strategy, noise vector, and data vector. Reusing a cached response for the first row with noise parameter $b_c$ requires a smaller noise parameter $b$ (and hence a bigger privacy budget) for the other rows than the cacheless MM to achieve the same accuracy level. 
    }
    \label{fig:mmm-use-inaccurate-cache-entries_simple}
\end{figure}



\subsubsection{Privacy Cost Optimizer}
We formalize the two aforementioned questions as an optimization problem, subject to the accuracy requirements, as follows. \\
\begin{mdframed}[style=MyFrame]
    \textbf{Cost estimation (CE) problem:} 
    Given a cache $\Cache$ and an instant strategy matrix $\StrategyFull$, 
    determine $\StrategyMatrixFree \subseteq (\StrategyFull \cap \Cache)$ (and $\StrategyMatrixPaid=\StrategyFull-\StrategyMatrixFree)$ and $b_{\StrategyMatrixPaid} \in [\LooseNoiseParameter, b_{\top}]$
    that minimizes the paid privacy budget 
    $\epsilon =\frac{ \|\StrategyMatrixPaid\|_1}{\PaidNoiseParameter}$     subject to accuracy requirement: 
    \begin{center}
    $\|\WorkloadFull\StrategyFull^+ diag(\BList_{\StrategyMatrixFree}||\BList_{\StrategyMatrixPaid})\|_F^2 \leq \alpha^2$
or
            $\Pr[\|\WorkloadFull \StrategyFull^+ Lap(\BList_{\StrategyMatrixFree}||\BList_{\StrategyMatrixPaid}) \|_{\infty} \geq \alpha ] \leq \beta$. 
    \end{center}
    \eat{
    where
     $\BList_{\StrategyMatrixFree}= [c.b \in \Cache  ~|~ c.\mathbf{a} \in \StrategyMatrixFree]$ 
     \xh{This notation $c.b$ and $c.a$ for $\BList_{\StrategyMatrixFree}$ is not used or explained anywhere else. It may be good to introduce it in Def. 3.1 or here when it is first time used.}
     and $\BList_{\StrategyMatrixPaid}=[b_{\StrategyMatrixPaid}~|~ \mathbf{a} \in \StrategyMatrixPaid]$.}
\end{mdframed}

In this optimization problem, the lower bound for $\PaidNoiseParameter$ is the loose bound for the cacheless MM (Equation~\eqref{eq:mm-loose-noise-param-bound}), and the upper bound $b_{\top}$ is $\frac{\|A\|_1}{\PrivacyBudgetPrecision}$, where $\PrivacyBudgetPrecision$ is the smallest possible privacy budget. 



In a brute-force solution to this problem, we can search over all possible pairs of $\StrategyMatrixFree \subseteq (\StrategyFull \cap \Cache)$ and $\PaidNoiseParameter \in [\LooseNoiseParameter,b_{\top}]$, and check whether every possible pair of $(\StrategyMatrixFree, \PaidNoiseParameter)$ can lead to an accurate response.
In this solution, the search space for $\StrategyMatrixFree$ will be $O(2^{|\StrategyFull \cap \Cache|})$ and thus the total search space will be $O\left(2^{|\StrategyFull \cap \Cache|}\cdot \log_2(|[\LooseNoiseParameter,b_{\top}]|)\right)$ if we apply binary search within $[\LooseNoiseParameter,b_{\top}]$. 
Hence, we need another way to efficiently determine optimal values for $(\StrategyMatrixFree, \PaidNoiseParameter)$.


\subsubsection{Simplified Privacy Cost Optimizer}
\label{subsubsec:simplified-priv-optimizer}
We present a simplification to arrive at a much smaller search space for  ($\StrategyMatrixFree$, $\PaidNoiseParameter$), while ensuring that $\PaidNoiseParameter$ improves over the noise parameter of the cacheless MM.  We observe that, if we perturb the paid strategy matrix with noise parameter $b_{\StrategyMatrixPaid}$ and choose cached entries with noise parameters smaller than $b_{\StrategyMatrixPaid}$, we will have a smaller error than a cacheless MM with a noise parameter $b=b_{\StrategyMatrixPaid}$ for all the queries in the strategy matrix. This motivates us to consider the following search space for  $\StrategyMatrixFree$. 
When given  $\PaidNoiseParameter$, we choose a free strategy matrix fully determined by this noise parameter: 
\begin{equation} \label{eq:Fbp}
\StrategyMatrixFree_{\PaidNoiseParameter} = \{c.\mathbf{a} \in \Cache ~|~ c.\mathbf{a} \in \Cache \cap  \StrategyFull, c.b \leq \PaidNoiseParameter\},
\end{equation}
and formalize a simplified optimization problem. 


\begin{mdframed}[style=MyFrame]
    \textbf{Simplified CE problem:} 
       Given a cache $\Cache$ and an instant strategy matrix $\StrategyFull$, 
    determine $\PaidNoiseParameter\in [b_L,b_{\bot}]$ (and $\StrategyMatrixFree=\StrategyMatrixFree_{\PaidNoiseParameter}$,
    $\StrategyMatrixPaid=\StrategyFull-\StrategyMatrixFree$)
    that minimizes the paid privacy budget 
    $\epsilon = \frac{\|\StrategyMatrixPaid\|_1}{\PaidNoiseParameter}$ 
    subject to: 
    \begin{center}
    $\|\WorkloadFull\StrategyFull^+ diag(\BList_{\StrategyMatrixFree}||\BList_{\StrategyMatrixPaid})\|_F^2 \leq \alpha^2$
or
            $\Pr[\|\WorkloadFull \StrategyFull^+ Lap(\BList_{\StrategyMatrixFree}||\BList_{\StrategyMatrixPaid}) \|_{\infty} \geq \alpha ] \leq \beta$. 
    \end{center}
    \eat{
    where 
     $\BList_{\StrategyMatrixFree}= [c.b \in \Cache  ~|~ c.\mathbf{a} \in \Cache \cap \StrategyFull, c.b \leq \PaidNoiseParameter]$ and $\BList_{\StrategyMatrixPaid}=[b_{\StrategyMatrixPaid}~|~ \mathbf{a} \in \StrategyMatrixPaid]$.}
\end{mdframed}

\begin{thm} \label{thm:simplifiedce}
The optimal solution to simplified CE problem incurs a smaller privacy cost $\epsilon$ than the privacy cost $\epsilon_{\StrategyMatrixFree = \emptyset}$ of the matrix mechanism without cache, i.e., MMM with $\StrategyMatrixFree=\emptyset$. 
\end{thm}

\begin{techreport}
\begin{figure}
    \centering
    \includegraphics[scale=0.15]{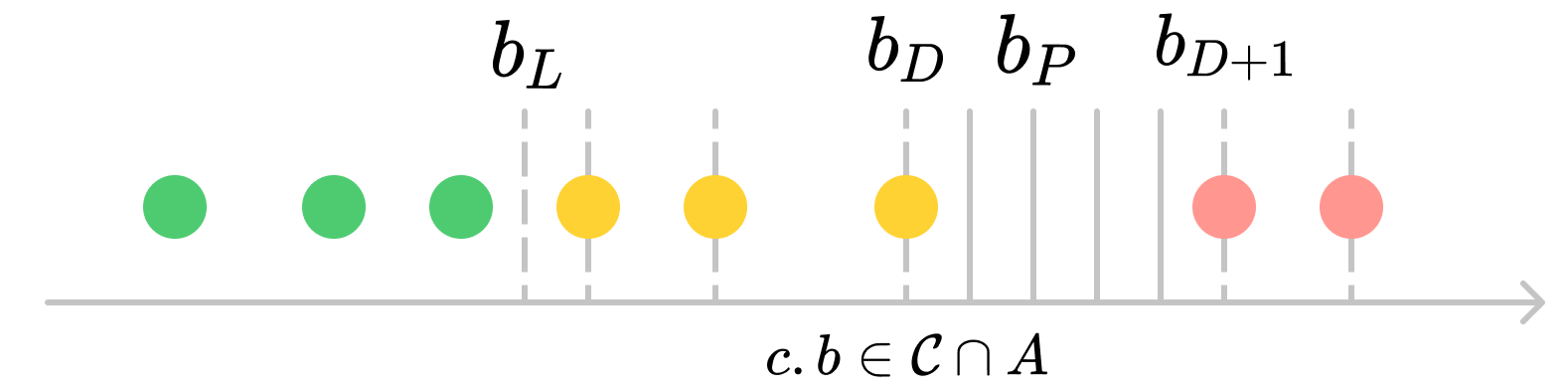}
    \caption{An example of cached noise parameters $c.b \in \StrategyFull\cap \Cache$ (dots) and the discrete (dashed lines) and continuous (full lines) binary searches through these parameters. 
    Parameters in green are accurate enough. The discrete search scans over $c.b \geq \LooseNoiseParameter$ and outputs $b_D$; the free matrix includes all cache entries in green and yellow. 
    The continuous search scans over the interval $[b_D,b_{D+1}]$ and identifies an optimal $\PaidNoiseParameter$. 
    }
    \label{fig:noise-params}
\end{figure}
\end{techreport}

\subsubsection{Algorithm for Simplified CE Problem}
\label{subsub:disc-cont-searches}
We present our search algorithm to find the best solution to the simplified CE problem, shown in the \textsc{estimatePrivacyBudget} function of Algorithm~\ref{algo:mmm}.  
\begin{vldbpaper}
    In our extended paper, we visualize our searches through the cached noise parameters.   
\end{vldbpaper}
\begin{techreport}
We visualize our searches through the cached noise parameters in Figure~\ref{fig:noise-params}. 
\end{techreport}
First, we setup the upper and lower bounds for the noise parameter $\PaidNoiseParameter$ for the simplified CE problem (Lines~\ref{line:answer-mmm-upperbound}-\ref{line:answer-mmm-loosebound}). 

{\bf Step 1: Discrete search for $\PaidNoiseParameter$.} We first search $\PaidNoiseParameter$ from the existing noise parameters in the cached strategy rows $\StrategyFull \cap \Cache$  that are greater than $\LooseNoiseParameter$ (Line~\ref{line:answer-mmm-discretespace}). We also include $\LooseNoiseParameter$ in this noise parameter list $\BList$.  Next, we sort the noise parameter list $\BList$ and conduct a binary search in this sorted list to find the largest possible $b_D\in \BList$ that meets the accuracy requirement (Line~\ref{line:answer-mmm-discretesearch}). 
During this binary search,  to check if a given $\PaidNoiseParameter$ achieves $(\alpha,\beta)$-accuracy requirement, we run the function \textsc{checkAccuracy}\begin{techreport}
, defined in Algorithm~\ref{algo:mmm}
\end{techreport}. This function first places all the cached entries with noise parameter smaller than $\PaidNoiseParameter$ into $\StrategyMatrixFree$ and the remaining entries of the strategy into $\StrategyMatrixPaid$ \begin{techreport}
(Line~\ref{line:answer-mmm-setF}) 
\end{techreport}. Then it runs an MC simulation \begin{techreport}(Lines~\ref{line:answer-mmm-mcstart}-\ref{line:answer-mmm-mcend})\end{techreport} of the error $\WorkloadFull\StrategyFull^+Lap(\BList_{\StrategyMatrixFree}|| \BList_{\StrategyMatrixPaid})$ 
(Proposition~\ref{prop:mmm-error}). If a small number of the simulated error vectors have a norm bigger than $\alpha$, then this paid noise vector $\PaidNoiseParameter$ achieves $(\alpha,\beta)$-accuracy guarantee.  
This MC simulation differs from a traditional one~\cite{ge_apex:_2017} which makes no use of the cache and has only a single scalar noise value for all entries of the strategy. 
On the other hand, if the accuracy requirement is $\alpha^2$-expected total square error, we simply check if $\|\WorkloadFull\StrategyFull^+ diag(\BList_{\StrategyMatrixFree}||\BList_{\StrategyMatrixPaid})\|_2^2 \leq \alpha^2$. 

\textbf{Step 2: Refining $\PaidNoiseParameter$ in a continuous space.}
We observe that we may further increase $\PaidNoiseParameter$, by examining the interval between $b_D$, which is the output from the discrete search, and the next largest cached noise parameter, denoted by $\top_C=b_{D+1}$. 
If $\top_C$ does not exist, then we set $\top_C=b_{\top}$. 
We conduct a binary search in a continuous domain $[b_D,\top_C]$ (Line~\ref{line:answer-mmm-contsearch}).  
This continuous search does not impact the free strategy matrix $\StrategyMatrixFree$ obtained from the discrete search, as the chosen noise parameter will be strictly smaller than $b_{D+1}$. 
\begin{techreport}
The continuous search is depicted through full lines in Figure~\ref{fig:noise-params}. 
\end{techreport}
This search outputs a noise parameter $\PaidNoiseParameter$.
Finally, this function returns $\PaidNoiseParameter$, the privacy budget $\epsilon=\frac{\|\StrategyMatrixPaid\|_1}{\PaidNoiseParameter}$, as well as the free and paid strategy matrices outputted from the discrete search. 

The search space for this simplified CE problem is $O(\log_2(|[\LooseNoiseParameter,$ $ b_{\top}]|))$. We only need to sort the cached matrix once, which costs  
$O(n_c \cdot \log(n_c))$, where $n_c=|\StrategyFull \cap \Cache|$. Hence, this approach significantly improves the brute-force search solution for the CE problem.

\begin{vldbpaper}
    \vspace{-1em}
\end{vldbpaper}

\section{Strategy Modules}\label{sec:strategy-module}
In this section, we first present the strategy transformer (ST), which is used by all of our cache-aware DP mechanisms.  We then present two optional modules for MMM: the Strategy Expander (SE) and Proactive Querying (PQ). 
\begin{vldbpaper}
Due to space constraints, all detailed algorithms for this section are included in the full paper~\cite{dpcacheextended}.
\end{vldbpaper}

\subsection{Strategy Transformer} \label{subsec:strategy-transformer}
The ST module selects an instant strategy from the given global strategy $\StrategyRaw \subseteq \StrategyMatrixGlobal$ based on the workload $\WorkloadRaw$. 
Since our cache-aware MMM and RP modules build on the matrix mechanism, we require a few basic properties for this instant strategy $\StrategyRaw$ to run the former mechanisms, with good utility. 
First, the strategy $\StrategyRaw$ should be \emph{a support} to the workload $\WorkloadRaw$~\cite{li2015matrix}, that is, it must be possible to represent each query in $\WorkloadRaw$ as a linear combination of strategy queries in $\StrategyRaw$. 
In other words, there exists a solution matrix $\mathbb{X}$ to the linear system $\WorkloadRaw = \mathbb{X}\StrategyRaw$. 
Second, $\StrategyRaw$ should have a \emph{low $l_1$ norm}, such that the privacy cost $\epsilon=\frac{\|\StrategyRaw\|_1}{b}$ for running MM is small, for a given a noise parameter $b$ (Proposition~\ref{prop:mmdp}).  
Third, using noisy responses to $\StrategyRaw$ to answer $\WorkloadRaw$ should incur minimal \emph{noise compounding}~\cite{hay2010boosting}. 
We thus present the strategy generator (SG) component, to address all of these requirements. 
The strategy generator only uses the global strategy $\StrategyMatrixGlobal$, 
and does not use the cached responses, 
to generate an instant strategy $\StrategyRaw$ for the workload $\WorkloadRaw$. 

Last, we require that $\StrategyRaw$ must be mapped to \emph{a full rank} matrix $\StrategyFull$, such that $\StrategyFull^+\NoisyResponseList$ is the estimate of the mapped data vector $\DataVector$ that minimizes the total squared error given the noisy observations $\NoisyResponseList$ of the strategy queries $\StrategyFull$~\cite[Section 4]{li2015matrix}. 
We present a full-rank transform (FRT) component to address this last requirement. 
Thus the ST module consists of two components: the strategy generator, and the full-rank transform, run sequentially. 

\subsubsection{Strategy Generator.} 
\label{subsubsec:strategy-generator}
\begin{vldbpaper}
Consider using the global strategy $\StrategyMatrixGlobal$ as follows: to answer the first workload, we obtain the noisy strategy responses for all nodes on the tree, thereby fully populating the cache. Cached noisy responses can be reused for future workloads. 
Though $\StrategyMatrixGlobal$ supports all possible counting queries over $dom(\Schema)$, it 
 has a very high norm $\|\StrategyRaw^*\|$, equal to the tree height $\log_k(n)+1$, where $n$ is the full domain size. 
Thus, answering the first workload would require spending a high upfront privacy budget, which may not be amortized across future workloads, as they may focus on a small part of the domain with higher accuracy requirements. 
\end{vldbpaper}

\begin{techreport}
Our global strategy $\StrategyMatrixGlobal$ is a $k$-ary tree over the full domain $dom(\Schema)$, hence, it supports all possible counting queries on the full domain. 
A baseline instant strategy $\StrategyRaw$ just uses the full global strategy matrix ($\StrategyRaw=\StrategyMatrixGlobal$), thus satisfying the first requirement. 
To answer the first workload, we obtain the noisy strategy responses for all nodes on the tree, thereby fully populating the cache and reusing the cached noisy responses for future workloads. 
However, this instant strategy has a very high norm $\|\StrategyRaw^*\|$, equal to the tree height $\log_k(n)+1$, where $n$ is the full domain size. 
Thus, answering the first workload would require spending a high upfront privacy budget. 
Moreover, this high upfront cost may not be amortized across future workload queries, for example, if the future queries do not require many nodes on this tree. 
Future workload queries may also have higher accuracy requirements, and we would thus need to re-sample noisy responses to the entire tree again, with a lower noise parameter. 
\end{techreport} 

To obtain a low norm strategy matrix, we only choose those strategy queries from $\StrategyMatrixGlobal$ that support the workload $\WorkloadRaw$. 
Intuitively, we wish to fill the cache with noisy responses to as many strategy queries as possible, thus we should bias our strategy generation algorithm towards the leaf nodes of the strategy tree. 
However, the DP noisy responses for the strategy nodes would be added up to answer the workload, and summing up responses to a large number of strategy leaf nodes compounds the DP noise in the workload response~\cite{hay2010boosting}. 
Thus, for each query in the workload $\WorkloadRaw$, we apply a top-down tree traversal to fetch the \emph{minimum} number of nodes in the strategy tree (and the corresponding queries in $\StrategyMatrixGlobal$) required to answer this workload query. 
Then we include all these queries into the instant strategy $\StrategyRaw$ for this workload $\WorkloadRaw$. 
The $L_1$ norm of the output strategy matrix is then simply the maximum number of nodes in any path of the strategy tree, and it is upper-bounded by the tree height. 
We present an example strategy generation below.

\begin{example}
\label{eg:basicstrategy}
    We continue with Example~\ref{eg:globalstrategy} shown in  Figure~\ref{fig:binary-tree-strategy-gen-proactive}, for an integer domain $[0,8)$.
    For the single workload query $\WorkloadRaw_1=\mathbb{w}=[0,7)$, the first iteration of our SG workload decomposition algorithm computes the overlap of $\mathbb{w}$ with its left child $c_1=\StrategyMatrixGlobal_{[0,4)}$ as $\mathbb{w}_{c1}=[0,4)$ and the overlap with its right child
    $c_2=\StrategyMatrixGlobal_{[4,8)}$ as $\mathbb{w}_{c2}=[4,7)$.
    The function only iterates once for the left child $c_1$, directly outputs that child's range $\StrategyMatrixGlobal_{[0,4)}$, as the base condition is satisfied\begin{techreport} (Line~\ref{line:decompose-workload-base-condition})\end{techreport}.
    In the next iteration for the right child $c_2$, the overlaps with both of its children are non-null ($[4,6)$ with $\StrategyMatrixGlobal_{[4,6)}$ and $[6,7)$ with $\StrategyMatrixGlobal_{[6,8)}$), and the corresponding strategy nodes are returned in subsequent iterations. 
    Since $\StrategyRaw_1$ has no overlapping intervals, $\|\StrategyRaw_1\|_1= 1 < \|\StrategyMatrixGlobal\|_1$. 
    \begin{techreport}
    
    The second workload $\WorkloadRaw_2$ has two queries with range predicates ($[2,6),[3,7)$). 
    The first workload query predicate requires the strategy nodes 
    $\StrategyMatrixGlobal_{[2,4)}$ and $\StrategyMatrixGlobal_{[4,6)}$, whereas the second query requires the following three nodes:  $\StrategyMatrixGlobal_{[3,4)}$,$\StrategyMatrixGlobal_{[4,6)}$ and $\StrategyMatrixGlobal_{[6,7)}$. 
    Hence, the second instant strategy $\StrategyRaw_2$ is a set of all of these strategy nodes.
    

    The global strategy has an $L_1$ norm $\|\StrategyMatrixGlobal\|_1=4$.
    The matrix forms of $\StrategyRaw_1$ and $\StrategyRaw_2$ can be generated as shown in Example~\ref{ex:frt}. Both strategy matrices improve over the global strategy in terms of their $L_1$ norms: $\StrategyRaw_1= 1 < \|\StrategyMatrixGlobal\|_1$ and $\StrategyRaw_2= 2 < \|\StrategyMatrixGlobal\|_1$. 
    \end{techreport}
    We observe that though $\StrategyMatrixGlobal$ is full-rank, due to the removal of strategy queries that do not support the workloads, both $\StrategyRaw_1$ and $\StrategyRaw_2$ are not full rank. 
\qedsymbol
\end{example}

\begin{techreport}
We formalize our strategy generation algorithm in the recursive function \textsc{workloadDecompose} given in Algorithm~\ref{algo:strategy-transformer-functions}.
This function takes as input a single workload query range interval $\mathbb{w}$ 
and a node $v$ on the tree $\mathcal{T}$. 
It first checks if the input predicate matches the range interval for the node $v$. 
If it does, it returns that range interval (Line~\ref{line:decompose-workload-base-condition}). 
Otherwise, for each child $c$ of node $v$, it computes the \emph{overlap} $\mathbb{w}_c$ of the range interval of that child with the interval $\mathbb{w}$ (Line~\ref{line:decompose-workload-compute-overlap}). 
For instance, the overlap of the range interval $[2,6)$ with $[0,4)$ is given by $[\text(max)\{(0,2)\}, \text{min}\{(4,6)\}) = [2,4)$.
For each child with a non-null range interval overlap, the function is called recursively with that overlap $\mathbb{w}_c$ (Line~\ref{line:decompose-workload-iterative-condition}). 
This function is called with the root of the tree $\mathcal{T}$, as the second argument, and returns with the decomposition of $\mathbb{w}$ over all child nodes in the tree ($\mathbb{a}$). 
It is run for each workload RCQ $\mathbb{w} \in \WorkloadRaw$, and $\StrategyRaw$ simply includes the union of each workload decomposition.

\begin{algorithm}[t]
\caption{Strategy Transformer (ST)} 
    \label{algo:strategy-transformer-functions}
    {
    \small
    \begin{algorithmic}[1]

     \Function{decomposeWorkload}{$\mathbb{w}$, node $v$}
        \If{$v$.query == $\mathbb{w}$ }
            \Return $v$ \label{line:decompose-workload-base-condition}
        \EndIf
        \State $\mathbb{a} \gets \emptyset$
        \If{$v$ has children}
            \For{child $c$ of node $v$}
                \State $\mathbb{w}_{c}$ $\gets$ \Call{OverlappingRCQ}{$\mathbb{w}$, node $c$.query} 
                \label{line:decompose-workload-compute-overlap}
                \If{$\mathbb{w}_{c} \ne \emptyset$}
                    \State $\mathbb{a} \gets \mathbb{a} $ $\cup$ \Call{decomposeStrategy}{$\mathbb{w}_{c}$, \text{node }$c$} 
                    \label{line:decompose-workload-iterative-condition}
                \EndIf
            \EndFor
        \EndIf
        \State \Return $\mathbb{a}$
     \EndFunction
    \end{algorithmic}
    }
\end{algorithm}
\end{techreport}

\subsubsection{Full Rank Transformer (FRT)} \label{sec:fullranktransform}
We transform an instant strategy matrix $\StrategyRaw$ to a full rank matrix $\StrategyFull$ by mapping the full domain $dom(\Schema)$ of size $n$ to a new partition of the full domain of $n'\leq n$ non-overlapping counting queries or buckets.
The resulting partition should still support all the queries in the instant raw strategy $\StrategyRaw$ output by our SG.
For efficiency, the partition should have the smallest possible number of buckets such that the transformed strategy $\StrategyFull$ will be full rank. First, we define a domain transformation matrix $\BucketMatrix$ of size $n'\times n$ that transforms the data vector $\DataVectorRaw$ over the full domain to the partitioned data vector $\DataVector$, such that $\DataVector= \BucketMatrix\DataVectorRaw$. Using $\BucketMatrix$, we can then transform a raw $\StrategyRaw$ to a full-rank $\StrategyFull$. 

\begin{defn}[Transformation Matrix]
    Given a partition of $n'$ non-overlapping buckets over the full domain $dom(\Schema)$, if the $i$th value in $dom(\Schema)$ is in the $j$th bucket, $\BucketMatrix[j,i]=1$; else, $\BucketMatrix[j,i]=0$.
\end{defn}
\begin{techreport}
 We elaborate exactly how we create a transformation matrix $\BucketMatrix$ to support strategy $\StrategyRaw$, as presented in \textsc{getTransformationMatrix} in Algorithm~\ref{algo:strategy-transformer-functions}. 
    To create $\BucketMatrix$, \textsc{getTransformationMatrix} starts with the first row of $\StrategyRaw$ (line~\ref{line:tm_init}), then iterates through the remaining rows (line~\ref{line:tm_loop}) updating the transformation matrix $\BucketMatrix$ as needed.
    If row $i$ is disjoint from all buckets, we simply add this row as a new bucket (line~\ref{line:tm_if_disjoint}).
    Otherwise, we construct a new bucket matrix $\BucketMatrix'$. 
    To do this, we first copy all rows of $\BucketMatrix$ that do not intersect with the current row of $\StrategyRaw$ (line~\ref{line:tm_copy_buckets}).
    Then, for the buckets that do intersect, we remove the intersection from that bucket and add a new bucket containing the intersection (line~\ref{line:tm_add_intersect}).
    Finally, if the row of $\StrategyRaw$ is a super-set of some buckets, we add the part of the row that is not covered by the buckets as a new bucket (line~\ref{line:tm_add_remain}).
    
    The \textsc{transformStrategy} function in  Algorithm~\ref{algo:strategy-transformer-functions}
    transforms $\StrategyRaw$ to $\StrategyFull$, using the transformation matrix $\BucketMatrix$ to determine which buckets are used for each query. 
    We first initialize $\StrategyFull$ to be the zero matrix (line~\ref{line:st_zero_init}).
    For row $i$ of $\StrategyRaw$ and row $j$ of $\BucketMatrix$, we compute whether bucket $j$ is needed to answer row $i$ by checking if $\BucketMatrix[j] \subseteq \StrategyRaw[i]$ (line~\ref{line:st_need_bucket}). 
    If bucket $j$ is needed, we set the corresponding entry of $\StrategyFull$ to $1$ (line~\ref{line:st_mark_bucket}).

\begin{algorithm}[t]
\caption{Full-rank transformer (Section~\ref{sec:fullranktransform})} 
\label{algo:strategy-transformer-functions}
    {
    \small
    \begin{algorithmic}[1]
    \Function{getTransformationMatrix}{$\StrategyRaw$}  
    \State $\BucketMatrix = \{\StrategyRaw[0]\}$\label{line:tm_init}
    \For{$i$ in range(1,$|\StrategyRaw|$) and $\StrategyRaw[i]\notin \BucketMatrix$
    }\label{line:tm_loop}
        \If{$\mathbb{t} \cdot \StrategyRaw[i]=0$ for all $\mathbb{t} \in \BucketMatrix $
        }\label{line:tm_if_disjoint}
        \State $\BucketMatrix \gets \BucketMatrix \cup \{\StrategyRaw[i]\} $ \Comment{Add a disjoint bucket $\StrategyRaw[i]$}
    
        \Else 
        \State $\BucketMatrix' \gets \{  \mathbb{t} \in \BucketMatrix ~|~
            \mathbb{t} \cdot \StrategyRaw[i] = 0
        \}$\label{line:tm_copy_buckets}
        
        \For{$\mathbb{t}$ in $\BucketMatrix$ and 
        $\mathbb{t} \cdot \StrategyRaw[i] \neq 0$
        }
            \State $\BucketMatrix' \gets \BucketMatrix' \cup (\mathbb{t}\cap  \StrategyRaw[i])
            \cup (\mathbb{t} - \StrategyRaw[i])
            $ \label{line:tm_add_intersect}
            
            \label{line:add_diff}
        \EndFor
        \State $\BucketMatrix' \gets \BucketMatrix' \cup (\StrategyRaw[i] - \sum_{\mathbb{t}\in \BucketMatrix' \wedge \mathbb{t}\cdot \StrategyRaw[i]\neq 0} \mathbb{t}) $\label{line:tm_add_remain}
        \State $\BucketMatrix = \BucketMatrix'$
        \EndIf
    \EndFor
    \State \Return{$\BucketMatrix$}
    \EndFunction
    \Statex
    
    \Function{transformStrategy}{$\StrategyRaw$}  
     \State $\BucketMatrix \gets $ \textsc{getTransformationMatrix}($\StrategyRaw$) \label{line:st-call-to-tm}
         \State Initialize  
        $\StrategyFull$ as a $|\StrategyRaw|\times |\BucketMatrix|$ zero-valued matrix\label{line:st_zero_init}
    \For{$i$ in range($|\StrategyRaw|$)
    and $j$ in range($|\BucketMatrix|$)}\label{line:st_transform_loop}
            \If {$\BucketMatrix[j] \subseteq \StrategyRaw[i]$}\label{line:st_need_bucket}
            \State Set $\StrategyFull[i,j] = 1$ \Comment{Bucket $j$ is contained in query $j$}\label{line:st_mark_bucket}
            \EndIf
    \EndFor
    \State \Return{$\StrategyFull, \BucketMatrix$}
    \EndFunction
    \end{algorithmic}
    }
\end{algorithm}
\end{techreport}
\begin{techreport}
\begin{example}\label{ex:frt}
We consider $\StrategyRaw_2$ from example~\ref{eg:globalstrategy}. The domain vector \DataVectorRaw~consists of the leaves of the tree depicted in Figure~\ref{fig:binary-tree-strategy-gen-proactive}. We get the following raw matrix form for $\StrategyRaw_2$.
\[
    \StrategyRaw_2 = \begin{bmatrix}
0 & 0 & 1 & 1 & 0 & 0 & 0 & 0 \\
0 & 0 & 0 & 0 & 1 & 1 & 0 & 0 \\
0 & 0 & 0 & 0 & 0 & 1 & 0 & 0 \\
0 & 0 & 0 & 0 & 0 & 0 & 1 & 0 
\end{bmatrix},  
\StrategyFull_2 =  \begin{bmatrix}
1 & 0 & 0 & 0 \\
0 & 1 & 1 & 0 \\
0 & 0 & 1 & 0 \\
0 & 0 & 0 & 1 
\end{bmatrix}  
\]
    We generate the full-rank form $\StrategyFull_2$ above using $\BucketMatrix$:
    \[
    \BucketMatrix =  \begin{bmatrix}
    0 & 0 & 1 & 1 & 0 & 0 & 0 & 0 \\
    0 & 0 & 0 & 0 & 1 & 0 & 0 & 0 \\
    0 & 0 & 0 & 0 & 0 & 1 & 0 & 0 \\
    0 & 0 & 0 & 0 & 0 & 0 & 1 & 0 
    \end{bmatrix}  
    \]

\end{example}
\end{techreport}
\begin{thm} \label{thm:fullrank}
    Given a global strategy $\StrategyMatrixGlobal$ in a $k$-ary tree structure, and an instant strategy $\StrategyRaw \subseteq \StrategyMatrixGlobal$,  \textsc{transformStrategy} outputs a strategy $\StrategyFull$ that is  full rank  and supports $\StrategyRaw$. 
\end{thm}
\begin{vldbpaper}
    We present an example FRT in our full paper.
\end{vldbpaper}
The ST module finally outputs $\StrategyRaw$, $\StrategyFull$, as well as the transformation matrix, as it can be used to transform $\WorkloadRaw$.
We use the full-rank versions $\WorkloadFull$, $\StrategyFull$ for all invocations of the matrix mechanism (i.e. computing $\WorkloadFull\StrategyFull^+$). 

\subsection{Strategy Expander} \label{subsec:strategy-expander}
We recall that our goal with \sysname is to use cached strategy responses, in order to save privacy budget on new strategy queries. 
 Section~\ref{sec:mmm-module} shows that MMM achieves this goal by directly reusing accurate  strategy responses from the cache for the basic instant matrix, i.e., by selecting $\mathbb{F} \subseteq \Cache\cap \StrategyRaw$.
In this strategy expander (SE) module, we provide efficient heuristics to include
\emph{additional} cached strategy entries out of $\Cache-\StrategyRaw$, to $\StrategyRaw$ to save more privacy budget.

\begin{techreport}
\begin{example}\label{se-example-setting}
Consider the cache structure $\Cache_{\StrategyMatrixGlobal}$ in Figure~\ref{fig:binary-tree-strategy-gen-proactive} and a new workload  
$\WorkloadRaw_1=\{ [0,1),[0,2),[2,4) \}$
and so $\StrategyRaw_1= \{\StrategyMatrixGlobal_{[0,1)}$ $,\StrategyMatrixGlobal_{[1,2)},$ $\StrategyMatrixGlobal_{[2,4)} \}$. 
The cache includes entries for $\StrategyMatrixGlobal_{[0,1)},\StrategyMatrixGlobal_{[1,2)}$ at noise parameter $b$, as well as $\StrategyMatrixGlobal_{[0,4)}$ at $4b$. 
The MMM module decides to reuse the first two cache entries, and pay for $\StrategyMatrixGlobal_{[2,4)}$ at $5b$, resulting in the noise parameter vector $\BList_{1}$, as depicted in Figure~\ref{fig:strat_expander_coutner_ex}.
The SE problem is deciding which cached responses (such as  $\StrategyMatrixGlobal_{[0,4)}$) can be added to the strategy to reduce it's cost.
\end{example}

Consider a strawman solution to choosing cache entries: we simply add all strategy queries from $\Cache - \StrategyRaw$ to $\StrategyRaw$, in order to obtain an expanded strategy $\StrategyRaw_{e}$. 
Prior work by Li et al.~\cite[Theorem 6]{li2015matrix} suggests that 
adding more queries to \StrategyFull~ always reduces the error of the matrix mechanism.
However, their result hinges on the assumption that all strategy queries are answered using i.i.d draws from the \emph{same} Laplace distribution~\cite{li2015matrix}.
Our cached strategy noisy responses can be drawn at different noise parameters in the past, and thus Li et al's result does not hold. 
In our case, the error term for the expanded strategy is given by: $\WorkloadFull\StrategyFull_{e}^+\text{diag}(\BList_{e})$ (Proposition~\ref{prop:mmm-error}). 
In Figure~\ref{fig:strat_expander_coutner_ex}, we present a counterexample for Li et al.'s result. 

\begin{example}\label{se-example-error-term}
Continuing Example~\ref{se-example-setting}, we expand $\StrategyFull_{1}$ to $\StrategyFull_{1e}$ by adding a row $\StrategyMatrixGlobal_{[0,4)}$.
In Figure~\ref{fig:strat_expander_coutner_ex}, we compute the $\alpha^2$-expected error using both $\StrategyFull_{1}$ and $\StrategyFull_{1e}$ and find that $\StrategyFull_{1e}$ has a larger error term.
\end{example}

We can see that the strawman solution can lead to a strategy with an increased error term. 
Importantly, this figure shows that adding a strategy query results in changed coefficients in $\WorkloadFull\StrategyFull_{e}^+$, that is, this added query changes the weight with which noisy responses to the original strategy queries are used to form the workload response. 
The added strategy query response must also be accurate, since adding a large, cached noise parameter to $\BList_{e}$ will also likely increase the magnitude of the error term (recall the example in Figure~\ref{fig:mmm-use-inaccurate-cache-entries_simple}).
\end{techreport}

\begin{vldbpaper}
Consider a strawman solution to choosing cache entries: we simply add all strategy queries from $\Cache - \StrategyRaw$ to $\StrategyRaw$, in order to obtain an expanded strategy $\StrategyRaw_{e}$. 
The error term for the expanded strategy is given by: $\WorkloadFull\StrategyFull_{e}^+\text{diag}(\BList_{e})$ (Proposition~\ref{prop:mmm-error}). 
In our full version of the paper~\cite{dpcacheextended}, we discuss related work hypothesizing this strawman solution~\cite{li2015matrix}, and we present an example wherein the strawman solution can lead to a strategy with an increased error term.
Intuitively, adding a strategy query results in changed coefficients in $\WorkloadFull\StrategyFull_{e}^+$, that is, this added query changes the weight with which noisy responses to the original strategy queries are used to form the workload response. 
The added strategy query response must also be accurate, since adding a large, cached noise parameter to $\BList_{e}$ will also likely increase the magnitude of the error term (recall the example in Figure~\ref{fig:mmm-use-inaccurate-cache-entries_simple}).

A brute force approach to find the optimal $\StrategyRaw_{e}$ would consider all possible subsets of cache entries from $\Cache-\StrategyRaw$ and check if the error is better than the original strategy. This induces an exponentially large search space of $O(2^{|\Cache|})$ possible solutions for $\StrategyRaw_{e}$. We propose a series of efficient heuristics to obtain a greedy solution. 

First, we search only the strategy queries from $\Cache - \StrategyRaw$ that are accurate enough. Recall that the \textsc{MMM.estimatePrivacyBudget} interface outputs the noise parameter $\PaidNoiseParameter$. 
Just as we used $\PaidNoiseParameter$ to compute $\StrategyMatrixFree$, we can also use it to select cache entries for $\StrategyRaw_{e}$ that are at least as accurate as other entries in \StrategyMatrixFree.  These accurate cached responses will likely improve the accuracy of the workload response. We first sort the cache entries in increasing order of the noise parameters and add each entry to $\StrategyRaw_{e}$ one by one until its noise parameter is greater than $\PaidNoiseParameter$, or, we reach a maximum bound on the cached strategy size. 
This approach reduces the search space from $O(2^{|\Cache|})$ to $O(|\Cache|)$ and ensures that the additional strategy rows do not significantly increase the run-time of \sysname.

Second, we ensure that each query $\mathbb{a}$ added to $\StrategyRaw_{e}$ is a parent or a child of an existing query $\mathbb{a}' \in \StrategyRaw$. Our heirarchical global strategy $\StrategyMatrixGlobal$ structures cache entries, and induces relations between the cached noisy responses. The constrained inference problem focuses on minimizing the error term for multiple noisy responses, while following consistency constraints among them, as described by Hay et al.~\cite{hay2010boosting}. 
For example, if we add the strategy queries corresponding to the siblings and parent nodes of an existing query in $\StrategyRaw$, we obtain an additional consistency constraint which tends to reduce error.  However, if we only added the sibling node, we would not have seen as significant (if any) improvement. This heuristic selects strategy rows that are more likely to reduce the privacy budget compared to MMM (\PaidEpsilon).

The privacy budget for $\StrategyFull_e$ is estimated using the \textsc{MMM.estimate\-Privacy\-Budget} interface. 
We encapsulate SE as a module rather than integrate it with MMM, since our heuristics might fail and $\StrategyFull_{e}$ might cost a higher privacy budget than the \StrategyFull~used by MMM\begin{techreport} (as illustrated in Example~\ref{se-example-error-term})\end{techreport}. 
Since Algorithm~\ref{algo:end-to-end} chooses the \textsc{answerWorkload} interface for the module and strategy with the lowest privacy cost, in the above case, $\StrategyFull_{e}$ is simply not used. 
\miti{We analyze the conditions under which our heuristics result in SE module being selected, in our full paper~\cite{dpcacheextended}.} 
\end{vldbpaper}
\begin{techreport}
In an optimal solution to this problem, one would have to consider adding each combination of cache entries from $\Cache-\StrategyRaw$.
This induces an exponentially large search space of $O(2^{|\Cache|})$ possible solutions for $\StrategyRaw_{e}$. 
Then for each candidate $\StrategyRaw_{e}$ we need to evaluate this error term and compare it to the error for the original strategy.
We provide an example in Figure~\ref{fig:strat_expander_coutner_ex}.
Instead of navigating this large search space for $\StrategyRaw_{e}$, we propose a series of efficient heuristics to obtain a greedy solution to this problem, as presented in Algorithm~\ref{algo:strat_expander}. 
In designing our algorithm, we have three goals: 
\squishlist
    \item[(1)] \textit{Search space:} Reduce the search space from $O(2^{|\Cache|})$ to $O(|\Cache|)$.
    \item[(2)] \textit{Efficiency:} Ensure that the additional strategy rows do not significantly increase the run-time of \sysname.
    \item[(3)] \textit{Greediness:} Select strategy rows that are most likely to reduce the privacy budget from that for MMM (\PaidEpsilon). 
\squishend
We achieve the first goal above by conducting a single lookup over cache entries in $\Cache - \StrategyRaw$ (Line~\ref{line:se-loop}), which would only incur a worst-case complexity of $O(|\Cache|)$. 
We limit the number of selected cached strategy rows to $\StratExpanderLimit$ (Line~\ref{line:acc_threshold}), thereby achieving our second goal of efficiency. 
Our greediness heuristics to select a strategy query are based on the two aforementioned factors that impact the workload error term, namely, the \emph{accuracy} of its cached noisy response, and how the noisy response is \emph{related} to noisy responses to the original strategy.

First, we must ensure that the strategy queries selected from $\Cache - \StrategyRaw$ are accurate enough.
Before conducting our cache lookup, we sort our cache entries in increasing order of the noise parameter, therefore our algorithm greedily prefers more accurate cache entries.
Recall that the \textsc{MMM.estimatePrivacyBudget} interface outputs the noise parameter $\PaidNoiseParameter$. 
Just as we used $\PaidNoiseParameter$ to compute \StrategyMatrixFree, we can also use it to select cache entries for $\StrategyRaw_{e}$ that are at least as accurate as other entries in \StrategyMatrixFree. 
These accurate cached responses will likely improve the accuracy of the workload response.
Thus, out of cache entries in $\Cache - \StrategyRaw$, we only consider cache entries whose noise parameter is lower than $\PaidNoiseParameter$ (Line~\ref{line:acc_threshold}). 

\begin{figure}
    \centering
    \[
    \WorkloadFull_{1} = 
    \begin{bmatrix}
        1 & 0 & 0 \\ 
        1 & 1 & 0 \\
        0 & 0 & 1 
    \end{bmatrix}
    ,\ \StrategyFull_{1} = I_3 = 
    \begin{bmatrix}
        1 & 0 & 0 \\ 
        0 & 1 & 0 \\
        0 & 0 & 1 
    \end{bmatrix}
    ,\ \BList_{1} = 
    \begin{bmatrix}
        b \\
        b \\
        5b
    \end{bmatrix}
    \]

    \[
    \StrategyFull_{1e} = 
    \left[
    \begin{array}{ccc}
        1 & 0 & 0 \\ 
        0 & 1 & 0 \\
        0 & 0 & 1 \\
        \hline
        1 & 1 & 1
    \end{array}
    \right]
    ,\ \BList_{1e} = 
    \left[
    \begin{array}{c}
        b \\
        b \\
        5b\\
        \hline
        4b
    \end{array}
    \right]
    \]
    \[
    \WorkloadFull_{1}\StrategyFull_{1}^+ diag(\BList_{1}) = 
    \left[ 
    \begin{array}{rrr}
1.00 \, b & 0 & 0 \\
1.00 \, b & 1.00 \, b & 0 \\
0 & 0 & 5.00 \, b
\end{array}
\right]
    \]
    \[
    \WorkloadFull_{1}\StrategyFull_{1e}^+ diag(\BList_{1e}) = 
    \left[ 
    \begin{array}{rrrr}
0.750 \, b & -0.250 \, b & -1.25 \, b & 1.00 \, b \\
0.500 \, b & 0.500 \, b & -2.50 \, b & 2.00 \, b \\
-0.250 \, b & -0.250 \, b & 3.75 \, b & 1.00 \,b
\end{array}
\right]
    \]
    \[
    \|\WorkloadFull_{1}\StrategyFull_{1}^+ diag(\BList_{1})\| = 28b^2 ,\
    \|\WorkloadFull_{1}\StrategyFull_{1e}^+ diag(\BList_{1e})\| = 29.1b^2
    \]
    \caption{
        An input instant strategy $\StrategyFull_1$ is expanded to a strategy $\StrategyFull_{1e}$. 
        Using their noise parameter vectors ($\BList_{1}, \BList_{1e}$), we can see that $\StrategyFull_{1e}$ has an error term of larger magnitude. 
    }
    \label{fig:strat_expander_coutner_ex}
\end{figure}

Second, our heirarchical global strategy $\StrategyMatrixGlobal$ structures cache entries, and induces relations between the cached noisy responses. 
The constrained inference problem focuses on minimizing the error term for multiple noisy responses, while following consistency constraints among them, as described by Hay et al.~\cite{hay2010boosting}. 
For example, if we add the strategy queries corresponding to the siblings and parent nodes of an existing query in $\StrategyRaw$, we obtain an additional consistency constraint which tends to reduce error. 
However, if we only added the sibling node, we would not have seen as significant (if any) improvement. 
Thus, our second greedy heuristic, in line~\ref{line:parent_child}, ensures that each query $\mathbb{a}$ added to $\StrategyRaw_{e}$ is a parent or a child of an existing query $\mathbb{a}' \in \StrategyRaw$.

The SE algorithm generates an expanded strategy $\StrategyRaw_{e}$ and transforms it to its full-rank form  $\StrategyFull_{e}$ (Line~\ref{line:se-frt}). 
The privacy budget for $\StrategyFull_e$ is estimated using the \textsc{MMM.estimate\-Privacy\-Budget} interface. 
We encapsulate SE as a module rather than integrate it with MMM, since our heuristics might fail and $\StrategyFull_{e}$ might cost a higher privacy budget than the \StrategyFull~used by MMM (as illustrated in Example~\ref{se-example-error-term}). 
Since Algorithm~\ref{algo:end-to-end} chooses to run the \textsc{answerWorkload} interface for the module and strategy with the lowest privacy cost, in the above case, $\StrategyFull_{e}$ is simply not used. 
We evaluate the success of our heuristics both experimentally and theoretically in Appendix~\ref{sec:se_heuristic_analysis}.

\begin{algorithm}[t]
    \caption{Strategy Expander (SE) (Section~\ref{subsec:strategy-expander})} 
    \label{algo:strat_expander}
    \begin{algorithmic}[1]
        \Function{generateExpandedStrategy}{$\StrategyRaw$, $\Cache$, $\PaidNoiseParameter$}
        \State $\StrategyRaw_{e} \gets \StrategyRaw$
        \For{$(\mathbb{a}, b, \NoisyResponse, t) \in (\Cache-\StrategyRaw)$ in an ascending order of $b$} \label{line:se-loop} 
            \If{$|\StrategyRaw_{e}|\geq |\StrategyRaw| + \StratExpanderLimit$ or $b>\PaidNoiseParameter$ \label{line:acc_threshold}  
            }
             \State Break
            \EndIf 
             \If{$\mathbb{a}' \cdot \mathbb{a} \neq 0$ for some $\mathbb{a}'\in \StrategyRaw$}\label{line:parent_child} 
        \State $\StrategyRaw_{e} \gets \StrategyRaw_{e} \cup \{\mathbb{a}\}$
             \EndIf
        \EndFor
    \State $(\StrategyFull_{e},\BucketMatrix_{e}) \gets$  \textsc{ST.transformStrategy}($\StrategyRaw_e$) 
    \Comment{Section~\ref{subsec:strategy-transformer}} \label{line:se-frt}
        \State \Return{$\StrategyRaw_{e}, \StrategyFull_{e},\BucketMatrix_{e}$} 
        \EndFunction
    \end{algorithmic}
\end{algorithm}
\end{techreport}

\subsection{Proactive Querying}
\label{subsec:proactive}
The proactive querying (PQ) module is an optional module for MMM.
The MMM obtains fresh noisy responses only for the paid strategy matrix $\StrategyMatrixPaid$, and inserts them into the cache.  
The goal of the PQ module is to proactively populate the cache with noisy responses to a subset $\StrategyRawProactive$ out of the remaining, non-cached strategy queries of the global strategy $(\StrategyMatrixGlobal-\Cache-\StrategyRawPaid)$, where $\StrategyRawPaid$ corresponds to the raw, non-full rank form of $\StrategyMatrixPaid$. 
Thus, we run the PQ module in the function MMM.\textsc{answerWorkload}($\cdot$) after obtaining the paid strategy matrix $\StrategyMatrixPaid$. 
Our cache-aware modules, including MMM, RP and SE, can use the cached noisy responses to $\StrategyRawProactive$ to answer future instant strategy queries. 
We wish to satisfy this goal without consuming any additional privacy budget over the MMM.

We first motivate key constraints for the PQ algorithm.
First, we do not assume any knowledge of future workload query sequences.
However, all future workload queries will be transformed into instant strategy matrices, and our cache-aware mechanisms will lookup the cache for cached strategy rows. 
Second, we also do not know the accuracy requirements for future workload queries.
Future workloads may be asked at different accuracy requirements than the current workload.
Thus, we choose to obtain responses to $\StrategyRawProactive$ at the highest possible accuracy requirements without spending any additional privacy budget over that required for $\StrategyMatrixPaid$ by MMM, which is 
$\epsilon=\frac{\|\StrategyMatrixPaid\|_1}{\PaidNoiseParameter}$.  
Our key insight is to generate $\StrategyRawProactive \subseteq (\StrategyMatrixGlobal-\Cache-\StrategyRawPaid)$ such that $\|\StrategyRawPaid \cup \StrategyRawProactive\|_{1} = \|\StrategyRawPaid\|_{1}$. 
Therefore, answering both instant strategies ($\StrategyRawPaid$ and 
$\StrategyRawProactive$) with the Laplace mechanism using $\PaidNoiseParameter$ costs no more privacy budget than simply answering $\StrategyRawPaid$ at $\PaidNoiseParameter$. 
\begin{vldbpaper}
\begin{thm}\label{thm:proactive}
Given a paid strategy matrix $\StrategyRawPaid$ 
our proactive strategy generation algorithm outputs $\StrategyRawProactive$ such that $\|\StrategyRawPaid\cup\StrategyRawProactive\|_1=\|\StrategyRawPaid\|_1$. 
\end{thm}

Our proactive generation algorithm consists of two top-down traversals of the tree.
We illustrate our proactive strategy generation algorithm through the following example.
Our detailed algorithm and theorem proofs are in the full paper~\cite{dpcacheextended}. 
\begin{example}
    \label{ex:proactive}
    In Figure~\ref{fig:binary-tree-strategy-gen-proactive},
    we apply our proactive strategy generation function to 
    to $\StrategyRawPaid_2=\{\StrategyMatrixGlobal_{[2,4)}, \StrategyMatrixGlobal_{[3,4)}\}$ for $\WorkloadRaw_2$ 
    in our example sequence.
    Our algorithm outputs $\StrategyRawProactive_2=$ $\{\StrategyMatrixGlobal_{[4,8)}$, $\StrategyMatrixGlobal_{[0,2)}$, $\StrategyMatrixGlobal_{[0,1)}$, $\StrategyMatrixGlobal_{[1,2)}$, $\StrategyMatrixGlobal_{[2,3)}\}$. 
    ($\StrategyMatrixGlobal_{[7,8)}$ is excluded from $\StrategyRawProactive_2$ since it is cached from $\StrategyRaw_1$ for $\WorkloadRaw_1$.)
    Here, $\|\StrategyRawPaid_2\|_1 = \|{\StrategyRawPaid_2 \cup \StrategyRawProactive_2}\|_1 = 2$. 
    Adding any other nodes from the tree to $\StrategyRawProactive_2$ will increase the number of nodes in $\StrategyRaw \cup \StrategyRawProactive$ that are on the same \emph{path} of the tree from 2 to 3 or 4, or in other words, $\|\StrategyRawPaid_2 \cup \StrategyRawProactive_2\|_1$ might increase. 
    Instead of any node in $\StrategyRawProactive_2$ we could obtain its children nodes, however, our algorithm prefers nodes at the higher layers of the tree, since they are more likely to be reused by other modules.
    Note that $\StrategyRawProactive_2$ does not only consist of disjoint query predicates. For example, $\StrategyMatrixGlobal_{[0,2)}$ and $\StrategyMatrixGlobal_{[0,1)}$ overlap. 
\end{example}
\end{vldbpaper}

\begin{techreport}
We present the function \textsc{searchProactiveNodes} for generating $\StrategyRawProactive$ in Algorithm~\ref{algo:proactive}. 
We formulate this algorithm 
in terms of the the $k$-ary tree representation of the global strategy $\StrategyMatrixGlobal$, denoted by $\mathcal{T}$. 
(We assume this is a directed tree with directed edges from the root to leaves and all paths refer to paths from a node to its leaves.)
For a node $v\in\mathcal{T}$, we define a binary function $\mathcal{M}_{\StrategyRawPaid}(n)$ to indicate if its corresponding query $v$.query is in $\StrategyRawPaid$.  

\begin{definition}\label{def:subtree-norm}
    We define the subtree norm of a node $v$ as the maximum number of marked nodes across all paths $p$ from node $v$-to-leaf in the subtree of $v$, i.e.,
    \begin{equation}\label{eqn:subtree-norm}
        \mathcal{S}_{\StrategyRawPaid}(v) = 
        \max_{p \in \text{subtree}(v)} \sum_{v\in p} \mathcal{M}_{\StrategyRawPaid}(v)
    \end{equation}
\end{definition}

The subtree norm of a node can be computed recursively as the sum of the mark function for that node and the maximum subtree norm of all of its children nodes, if any. 
The proactive module first recursively computes the subtree norm of each node before generating $\StrategyRawProactive$. 
This step requires a single top-down traversal of the strategy decomposition tree. 
Given that the children of each node have non-overlapping ranges, the subtree norm enables us to define the $L_1$ norm of $\StrategyRawPaid$:
\begin{lemma}
    \label{lemma:subtree-norm-root}
    The $L_1$ norm of the $\StrategyRawPaid$~matrix is equal to the subtree norm of the root of the tree with marked nodes corresponding to $\StrategyRawPaid$:
    \begin{equation}
        \mathcal{S}_{\StrategyRawPaid}(\mathcal{T}\text{.root}) = \|\StrategyRawPaid\|_1
    \end{equation}
\end{lemma}

Our proactive strategy generation function, \textsc{generateProactiveStrategy}, is presented in Algorithm~\ref{algo:proactive}.
This function conducts a recursive top-down traversal of the strategy decomposition tree (line~\ref{line:proac-recursion}), and outputs a list of nodes to be fetched proactively 
into $\StrategyRawProactive$, 
such that the sum of the number of marked nodes and proactively fetched nodes for each path in the tree is at most $\|\StrategyRawPaid\|_{1}$.
\begin{lemma}
    \label{lemma:proac-invariant} 
    The proactive strategy $\StrategyRawProactive$ generated by \textsc{generateProactiveStrategy} for an input $\StrategyRawPaid$ satisfies the condition: 
    \begin{equation}
        \label{eqn:proac-invariant}
        \forall \text{ paths } p \in \mathcal{T}, 
        \sum_{v\in p} \mathcal{M}_{\StrategyRawPaid\cup \StrategyRawProactive}(v) 
        \leq     \mathcal{S}_{\StrategyRawPaid}(\mathcal{T}\text{.root}) 
        = \|\StrategyRawPaid\|_{1}
    \end{equation}
\end{lemma}

The second argument $r$ in the function \textsc{searchProactiveNodes}
represents the number of \textit{remaining} nodes that can be fetched proactively for the subtree originating at node $v$. 
Thus in the first call, we pass the root node for the first argument, 
and the second argument is initially set to $\|\StrategyRawPaid\|_{1}$.
We decrement $r$ whenever we encounter a marked node (line~\ref{line:proac-marked-condition-decrement}) or  when we add a node to the proactive output list (line~\ref{line:proac-add-node-to-list}). 
In the latter case, we require that the node is not cached and that $r$ is greater than the subtree norm of the node, $\mathcal{S}_{\StrategyRawPaid}(v)$ or $s$, as seen in line~\ref{line:proac-condition}. 
This condition ensures that we can safely add node $v$ to the proactive list, while achieving Equation~\eqref{eqn:proac-invariant}. (We prove all PQ module lemmas and theorems in Appendix~\ref{app:proacitveproof}.) 

\begin{thm}\label{thm:proactive}
Given a paid strategy matrix $\StrategyRawPaid$ 
Algorithm~\ref{algo:proactive} outputs $\StrategyRawProactive$ such that $\|\StrategyRawPaid\cup\StrategyRawProactive\|_1=\|\StrategyRawPaid\|_1$. 
\end{thm}

\begin{algorithm}[t]
    \caption{Proactive Querying (PQ) (Section~\ref{subsec:proactive})}
    \label{algo:proactive}
    {
    \small \begin{algorithmic}[1]
        \Function{searchProactiveNodes}{node \textit{v}, $r$, $\Cache$, $\StrategyRawProactive$}
            \If{$v$.query $\in \StrategyRawPaid$} \label{line:proac-marked-condition-decrement}
                \State $r \gets r - 1$ 

            \ElsIf{$v$.query $\notin \Cache$ and 
            $v$.s < $r$} \label{line:proac-condition}
                \State $\StrategyRawProactive \gets$ $\StrategyRawProactive \cup \{v.\text{query}\}$
                \label{line:proac-add-node-to-list}
                \State 
                $r \gets r - 1$ 
            \EndIf
            
            \If{$r > 0 $ and $v$ has children}
                \For{child $c$ of node $v$}
                    \State \Call{searchProactiveNodes}{node \textit{c}, $r$,  $\Cache$, $\StrategyRawProactive$} \label{line:proac-recursion}
                \EndFor
            \EndIf
            \State \Return
        \EndFunction
    \end{algorithmic}
    }
\end{algorithm}



\begin{example}
    \label{ex:proactive}
    In Figure~\ref{fig:binary-tree-strategy-gen-proactive},
    we apply our proactive strategy generation function to 
    to $\StrategyRawPaid_2=\{\StrategyMatrixGlobal_{[2,4)}, \StrategyMatrixGlobal_{[3,4)}\}$ for $\WorkloadRaw_2$ 
    in our example sequence.
    We annotate each node with the values of $r$ and $v.s$ from line~\ref{line:proac-condition} in function \textsc{searchProactiveNodes}.
    The tree nodes $\StrategyMatrixGlobal_{[4,8)}$, $\StrategyMatrixGlobal_{[0,2)}$, $\StrategyMatrixGlobal_{[0,1)}$, $\StrategyMatrixGlobal_{[1,2)}$, $\StrategyMatrixGlobal_{[2,3)}$ and $\StrategyMatrixGlobal_{[7,8)}$ satisfy the condition $r > v.s$. 
    All nodes other than $\StrategyMatrixGlobal_{[7,8)}$ are output into $\StrategyRawProactive_2$; the latter node is excluded since it is cached from $\StrategyRaw_1$ for $\WorkloadRaw_1$. 
    Note that $\StrategyRawProactive_2$ does not only consist of disjoint query predicates. For example, $\StrategyMatrixGlobal_{[0,2)}$ and $\StrategyMatrixGlobal_{[0,1)}$ overlap. 
    However, $\mathcal{S}_{\StrategyRawPaid_2}(\mathcal{T}) = \mathcal{S}_{\StrategyRawPaid_2 \cup \StrategyRawProactive_2}(\mathcal{T}) = 2$.
\end{example}

\subsubsection{Integration.} 
The MMM \textsc{answerWorkload} function perturbs $\StrategyRawProactive$ with the same noise parameter as for $\StrategyRawPaid$, namely $\PaidNoiseParameter$, to obtain the noisy responses $\NoisyProactiveResponse$. 
It then updates the cache $\Cache$ with $\{(\mathbb{p'}, \PaidNoiseParameter, \tilde{y}, t)  | $ $\mathbb{p'} \in \StrategyRawProactive, \tilde{y} \in \NoisyProactiveResponse\}$.
We observe that we do not answer the analyst's workload query $\WorkloadFull$ using $\NoisyProactiveResponse$. 
Importantly, this is the reason why we do not incorporate $\StrategyRawProactive$ in estimating $\PaidNoiseParameter$ in our MMM cache-aware cost estimation function.
The PQ module can also be used while the SE module is turned on.
Algorithm~\ref{algo:proactive} can also be applied to multi-attribute strategies, as we discuss in Section~\ref{sec:multi-attribute-case}. 

\end{techreport}


\section{Relax Privacy Mechanism} \label{sec:rp-module}

\begin{algorithm}[t]
	\caption{Relax Privacy (RP) (Section~\ref{sec:rp-module})}
	\label{algo:relax-privacy}
	\small{\begin{algorithmic}[1]
		\Function{answerWorkload}{$\Cache, \StrategyFull, \WorkloadFull, \DataVectorRaw$}
			\State $\noisev_{o} \gets \NoisyResponseList_{o} - \StrategyRaw_o\DataVectorRaw$\label{line:rp_get_old_noise} \Comment{Old noise vector for  $\StrategyRaw_o$.}
		    \State $\noisev \gets$ \Call{lapNoiseDown}{$\noisev_{o}, b_{o}, b$} 
		    \Comment{Koufogiannis et al.~\cite{relaxed_privacy}} \label{line:rp_get_new_noise}
		    \State $\NoisyResponseList \gets \StrategyRaw_o\DataVectorRaw + \noisev$\label{line:rp_get_new_response} \Comment{New noisy responses to $\StrategyRaw_o$} 
		    \State Update cache $\Cache_{\StrategyMatrixGlobal}$ with ($\StrategyRaw_{o}, b, \NoisyResponseList, t$=current time) \label{line:rp-update-cache}
		    \State $\NoisyResponseList' \gets \NoisyResponseList$ for $\StrategyRaw \subseteq \StrategyRaw_{o}$ \Comment{New noisy responses to $\StrategyRaw$} \label{line:rp-filter}
        
            \State \Return  
		        $\WorkloadFull\StrategyFull^+\NoisyResponse'$  \label{line:rp-workload-response}
		\EndFunction
		\Statex

	    \Function{estimatePrivacyBudget}{$\Cache, \StrategyFull, \WorkloadFull, \alpha, \beta$} 
	        \State $b \gets$ 
	        \textsc{MMM.estimatePrivacyBudget}(
	        $\Cache=\emptyset$, $\StrategyFull$, $\WorkloadFull$, $\alpha,\beta$) \label{line:rp-target-cost}
	        \State $\Cache \gets \{\cdots(\StrategyRaw_{t}$, $\NoisyStrategyResponse_{t}$, $b_{t})\}$ \label{line:rp-group-by-ts} \Comment{Group queries in $\Cache$ by timestamp.}
	        \State $S_{RP} \gets \StrategyRaw_{j} \in \Cache_{t=j} | \StrategyRaw_{j} \supseteq \StrategyRaw$ \Comment{Keep only those $\StrategyRaw_t$ that contain $\StrategyRaw$} \label{line:rp_possible_strats} 
	        \If{$S_{RP} = \emptyset$} \label{line:rp_check_if_run}
			    \State \Return ``RP cannot run for this input $\StrategyRaw$.''
		\EndIf
	        \State $o=\underset{j}{\arg \min}  \quad \epsilon_{RP,j} = \frac{\|\StrategyRaw_{j}\|_1}{b} - \frac{\|\StrategyRaw_{j}\|_1}{b_j}$ \Comment{$\StrategyRaw_o$ has the lowest RP cost} \label{line:rp_choose_min}
		\State \Return $b_o, \quad \epsilon_{RP,o}$ \Comment{Cached noise parameter, RP cost for $\StrategyRaw_{o}$}
		\EndFunction
	\end{algorithmic}}
\end{algorithm}

When exploring a database, a data analyst may first ask a series of workloads at a low accuracy (spending $\epsilon_1$), and then re-query the most interesting workloads at a higher accuracy (spending $\epsilon_2 > \epsilon_1$).
\begin{techreport}
(The repeated workload can also be asked by a different analyst.) 
\end{techreport}
The cumulative privacy budget spent by the MMM will be $\epsilon_1+\epsilon_2$ due to sequential composition. 
The goal of the Relax Privacy module is to spend less privacy budget than MMM on such repeated workloads with higher accuracy requirements.

Koufogiannis et al.~\cite{relaxed_privacy} refine a noisy response at a smaller $\epsilon_1$, to a more accurate response at a larger $\epsilon_2$, using only a privacy cost of $\epsilon_2-\epsilon_1$~\cite{relaxed_privacy,pioneer}. 
However, their framework only operates with the simple Laplace mechanism.
Thus, we achieve the aforementioned goal by closely integrating their framework~\cite{relaxed_privacy} with the matrix mechanism and our DP cache.
\begin{techreport}
    This design of the RP module meets a secondary goal, namely, our RP module handles not only repeated analyst-supplied \textit{workloads}, but also different workloads that result in identical strategy matrices. 
    We exploit the fact that other modules, such as the PQ module, also operate over the matrix mechanism and DP cache.
    Therefore, our RP module can be seen as generalizing these frameworks to operate over more workload sequences. 
\end{techreport}

\subsection{Estimate Privacy Budget Interface}

We first describe the \textsc{estimatePrivacyBudget} interface for RP, and as with the MMM, it estimates the privacy budget required by the RP mechanism. 
The privacy budget required for the RP mechanism is defined as the difference between the new or  \emph{target} privacy budget for the output strategy noisy responses \NoisyStrategyResponse~to meet the accuracy guarantees, and the old or \emph{cached} privacy budget (\HistoricalPrivacyBudgetRP) that cached responses to \StrategyRaw~were obtained at. 
The target noise parameter is the noise parameter required \emph{by the cacheless MM} to achieve an $(\alpha,\beta)$-accuracy guarantee for $\WorkloadFull, \StrategyFull$. 
It can be obtained by running the \textsc{estimatePrivacyBudget} of MMM with an empty cache (line~\ref{line:rp-target-cost}). 
Then the main challenge of this interface is to choose which past strategy entries should be relaxed by the RP mechanism, based on the smallest RP cost as defined above.

Each strategy query $\mathbb{a} \in \StrategyRaw$ may be cached at a different timestamp. 
Relaxing each such set of cache entries across different timestamps, through sequential composition, requires summing over the RP cost for each set, and can thus be very costly. 
For simplicity, we design the RP mechanism to relax the entirety of a past strategy matrix, rather than picking and choosing strategy entries across different timestamps. 
Our RP cache lookup condition groups cache entries by their timestamps 
to form cached strategy matrices (Line~\ref{line:rp-group-by-ts}), and identifies all candidate matrices that include the entire input strategy
(Line~\ref{line:rp_possible_strats}). 
The inclusion condition (instead of an equality) allows proactively fetched strategy entries to be relaxed, at no additional cost to relaxing $\StrategyRaw_{j}$. 
If answering $\StrategyRaw$ using the cache requires: (1) composing cache entries spanning multiple timestamps, or (2) composing cache entries at one timestamp and paid (freshly noised) strategy queries at the current timestamp, then the RP cost estimation interface simply returns nothing (Line~\ref{line:rp_check_if_run}) and \sysname will instead use another module.
\begin{example}
    Suppose that the workloads shown in Figure~\ref{fig:binary-tree-strategy-gen-proactive} have been asked in the past at $\alpha_1$, and have been answered through MMM, as discussed in Example~\ref{ex:proactive}. 
    Now $\WorkloadRaw_3=\{[3,8)\}$ is asked at $\alpha_3<\alpha_1$. 
    We have $\StrategyRaw_3 = \{[3,4), [4,8)\}$. 
    Thus $\StrategyRaw_3 \subset \StrategyRawPaid_2 \cup \StrategyRawProactive_2$, and the RP module relaxes all of $\StrategyRaw_{3,RP} = \StrategyRawPaid_2 \cup \StrategyRawProactive_2$ from $\alpha_1$ to $\alpha_3$. 
\end{example}
For each candidate cached strategy matrix  $\StrategyRaw_{j}$, we compute the RP cost to relax its cached noisy response vector $\NoisyStrategyResponse_j$ from ${b_j}$ to the new target $b$ as $\epsilon_{RP,j}$. 
Lastly, the RP module chooses to relax the candidate past strategy $\StrategyRaw_j$ with the minimum RP cost (Line~\ref{line:rp_choose_min}). 
\begin{techreport}
Importantly, we only seek to obtain accuracy guarantees over $\WorkloadFull\StrategyFull^+$, and not over $\WorkloadFull\StrategyFull_{o}^+$, and thus we compute the target noise parameter $b$ based on the input strategy matrix $\StrategyRaw$ (Line~\ref{line:rp-target-cost}), rather than the optimal cached candidate $\StrategyRaw_{o}$.
\end{techreport}
For the chosen cached strategy matrix $\StrategyRaw_{o}$, we return the cached noise parameter $b_o$ and the RP cost $\epsilon_{RP,o}$.

\subsection{Answer Workload Interface}
The RP \textsc{answerWorkload} interface is a straightforward application of Koufogiannis et al.'s noise down module. 
\begin{techreport}
    The \textsc{estimatePrivacyBudget} interface records the following parameters: the optimal cached or old strategy to relax ($\StrategyRaw_{o}$), its cached noisy response ($\NoisyResponse_{o}$), the cached noise parameter ($b_{o}$) and the target noise parameter ($b$). 
    \end{techreport}
We first compute the Laplace noise vector used in the past $\noisev_{o}$, by subtracting the ground truth for the cached old strategy $\StrategyRaw_{o}\DataVectorRaw$ from the cached noisy response $\NoisyResponse_{o}$ (line~\ref{line:rp_get_old_noise}).
We can now supply Koufogiannis et al.'s noise down algorithm with the old noise vector $\noisev_{o}$, the cached noise parameter $b_{o}$, and the target noise parameter $b$. 
This algorithm draws noise from a correlated noise distribution, and outputs a new, more accurate noise vector at noise parameter $b$ (line~\ref{line:rp_get_new_noise})~\cite[Algorithm 1]{relaxed_privacy}. 
We can simply compute the new noisy response vector to $\NoisyResponse_{o}$ using the ground truth and the new noise vector (line~\ref{line:rp_get_new_response}). 
We then update the cache with the new, more accurate noisy responses, which can be used to answer future strategy queries (line~\ref{line:rp-update-cache}). 
Finally, we do not need the noisy strategy responses to $\StrategyRaw_{o} - \StrategyRaw$ to answer the data analyst's workload, and so we filter them out to simply obtain new noisy responses $\NoisyResponse'$ to $\StrategyRaw$ (line~\ref{line:rp-filter}).
We use $\NoisyResponse'$ to compute the workload response and return it to the analyst (line~\ref{line:rp-workload-response}).


\begin{vldbpaper}
    \vspace{-2.5mm}
\end{vldbpaper}

\section{Multiple attribute workloads}
\label{sec:multi-attribute-case}
We extend \sysname to work over queries with multiple attributes. 
We define a single data vector \DataVectorRaw~over $\text{dom}(\Schema)$ 
as the cross product of $d$ single-attribute domain vectors.
It represents the frequency of records for each value of a marginal over all attributes. 
However, $|\DataVectorRaw|$ and thus $|\Cache|$ could be very large due to the cross product. 

We observe that not all attributes may be referenced by analysts in their workloads. 
Suppose that each workload includes marginals over a set of attributes $S_\Attribute \in R$.
That is, each marginal $\mathbb{w} \in \WorkloadRaw$ includes $|S_\Attribute|=k \le d$ RCQs, with one RCQ over each attribute ($\mathbb{w}=\prod_{j}^{k}\mathbb{w}_j$ ). 
These workloads would share a common, smaller domain and hence a data vector 
\begin{vldbpaper}
$\DataVectorRaw_{S_\Attribute}$.
\end{vldbpaper}
\begin{techreport}
$\DataVectorRaw_{S_\Attribute}=\bigotimes_{i=1}^{k \le d}\DataVectorRaw_{\Attribute_i}$. 
\end{techreport}
Similarly, instead of creating a large cache, we create a set of smaller caches, with one cache $\Cache_{S_\Attribute}$  for each \emph{unique} combination of attributes $S_\Attribute$ encountered in a workload \emph{sequence}. 
Cache entries can thus be reused across workloads that span the same set of attributes. 
The entries of each smaller cache $\Cache_{S_\Attribute}$ are indexed by its associated domain vector $\DataVectorRaw_{_{S_\Attribute}}$.

\begin{techreport}
\begin{example}\label{multi-attr-workload-seq-ex}
    Consider three attributes with $dom(\Attribute_1)=[0,4)$, $dom(\Attribute_2)=[0,8)$, and $dom(\Attribute_3)=[0,2)$. 
    The analyst supplied a workload sequence of RCQ over different sets of attributes: $\WorkloadRaw_1$ over $S_1=\{\Attribute_1,\Attribute_2\}$, $\WorkloadRaw_2$ over $S_2=\{\Attribute_2,\Attribute_3\}$ and then $\WorkloadRaw_3$ over $S_3=S_1=\{\Attribute_1,\Attribute_2\}$. 
    Then $\WorkloadRaw_1$ and $\WorkloadRaw_3$ are answered using the same domain vector $\DataVectorRaw_{S_1}=\DataVectorRaw_{\Attribute_1} \otimes \DataVectorRaw_{\Attribute_2}$ and a cache indexed over it, $\Cache_{S_1}$. 
    However, $\WorkloadRaw_2$ is answered using a different domain vector $\DataVectorRaw_{S_2}=\DataVectorRaw_{\Attribute_2} \otimes \DataVectorRaw_{\Attribute_3}$ and the corresponding cache $\Cache_{S_2}$. 
\end{example}
\end{techreport}
Our cache-aware MMM, SE and RP modules can be extended trivially to the multi-attribute case, since these modules would simply operate on the larger domain vector. 
However, in order to generate \StrategyRaw, the ST module relies on a $k$-ary strategy tree, corresponding to 
$\StrategyMatrixGlobal$ for the single-attribute case. 
Thus, we extend the ST and PQ modules by defining this global strategy tree using marginals over multiple attributes. 
Our extended ST and PQ modules serve as a proof-of-concept that other modules can be extended for other problem domains. 
\begin{vldbpaper}
\miti{
We detail the multi-attribute strategy tree generation in our full paper~\cite{dpcacheextended}. We also discuss 
handling complex SQL queries such as joins, in its future work section. }
\end{vldbpaper}
\begin{techreport}

For a given workload \WorkloadRaw~ that spans a set of attributes $S_\Attribute$, we can use the  single-attribute strategy tree $\mathcal{T}_{i}$ for each attribute $\Attribute_i \in S_\Attribute$ to construct a multi-attribute strategy tree $\mathcal{T}$ for \WorkloadRaw, as follows. 
Intuitively, \DataVectorRaw~ consists of marginals that are constructed by taking each unit value in $dom(\Attribute_i)$, and forming a cross product with each unit value in the $dom(\Attribute_{i+1})$, and so on. 
Unit values in $dom(\Attribute_i)$ are represented by leaf nodes on its strategy tree $\mathcal{T}_{i}$, as illustrated in Figure~\ref{fig:binary-tree-strategy-gen-proactive}. 
Thus, we form a two-attribute strategy tree $\mathcal{T}$ over $\Attribute_{1}$ and $\Attribute_{2}$, by attaching the tree $\mathcal{T}_{2}$ to each \emph{leaf} node for the tree $\mathcal{T}_{1}$. 
(We can then simply extend this definition to $k=\|S_\Attribute\|$ attributes in \WorkloadRaw.) 
Nodes on $\mathcal{T}$ either represent marginals over these attributes or sums of marginals. 
In particular, the leaf nodes on $\mathcal{T}$ represent the unit value strategy marginals over  $\Attribute_{1}$ and $\Attribute_{2}$. 


We can thus decompose each workload marginal $\mathbb{w}$ in terms of nodes on this tree $\mathcal{T}$. 
In the first step, we decompose each single-attribute marginal $\mathbb{w}_i$ following our single-attribute \textsc{decomposeWorkload} function (Algorithm~\ref{algo:strategy-transformer-functions}) to get a set of single-attribute strategy nodes $\mathbb{a}_i$. 
The second step is run for all but the last ($k$th) attribute: in this step, we \emph{further} decompose $\mathbb{a}_i$ in terms of the leaf nodes on tree $\mathcal{T}_{i}$, to obtain $\mathbb{a}_{\text{leaf,}j}$. 
In the final step, we compute the cross-product of the single-attribute strategy nodes $\mathbb{a}_{\text{leaf,}j}$, over all $j$, and a final cross product with the decomposition of the $k$th attribute, $\mathbb{a}_{k}$, to form the strategy marginal $\mathbb{a}$ for the workload marginal $\mathbb{w}$, as follows: 
$\mathbb{a}=\bigotimes_{i=1}^{k-1 \le d} \mathbb{a}_{\text{leaf,}j} \otimes \mathbb{a}_{k}$.

To determine the order in which to decompose the attributes, we first compute the minimum granularity needed for each attribute to be able to answer the workload.
We then sort the attributes in acceding order of granularity. 
The intuition is that this should result in the least amount of noise compounding for the lowest attribute on the tree.
Each decomposition is greedy in that we only go down to the minimum granularity for the given workload and not to the leaf nodes.
This further reduces the noise compounding of our approach.

We design the above strategy transformer such that our PQ module can easily be applied to the structure. The only difference is that the \textsc{generateProactiveStrategy} algorithm is run over the strategy tree for the combination of attributes in the workload ($S_\Attribute$) rather than the global tree structure.
The PQ module can thus add queries for any subset of attributes in $S_\Attribute$, prioritizing nodes higher up in the tree, which correspond to fewer number of attributes. 
In practice, this is preferable as it focuses on expanding the areas an analyst has already shown interest in rather than completely unrelated attributes.

\end{techreport}
 

\begin{techreport}
\begin{example}
\label{example:multi-attr-strategy}
Consider the first workload in the sequence in Example~\ref{multi-attr-workload-seq-ex}. 
It consists of one marginal tuple over attributes $\Attribute_1$ and $\Attribute_2$:  $(\WorkloadRaw=[0,3)_{\Attribute1}, [1,6)_{\Attribute2})$.
We denote the global strategy matrix for attribute $i$ as $\StrategyMatrixGlobalPerAttribute{i}$. 
We decompose the RCQ for each attribute in the marginal, using our single attribute strategy generator. 
Thus: $\StrategyRaw_{\Attribute1}= \{\StrategyMatrixGlobalPerAttribute{1}_{[0,2)}, \StrategyMatrixGlobalPerAttribute{1}_{[2,3)}\}$ and $\StrategyRaw_{\Attribute2}= \{\StrategyMatrixGlobalPerAttribute{2}_{[1,2)}, \StrategyMatrixGlobalPerAttribute{2}_{[2,4)}, \StrategyMatrixGlobalPerAttribute{2}_{[4,6)}\}$. 
We then construct the two-attribute strategy marginals as the cross product of these RCQs: 
$\StrategyRaw=\StrategyRaw_{\Attribute1} \bigotimes \StrategyRaw_{\Attribute2}$
$=\{(\StrategyMatrixGlobalPerAttribute{1}_{[0,2)}, \StrategyMatrixGlobalPerAttribute{2}_{[1,2)}),$ 
$(\StrategyMatrixGlobalPerAttribute{1}_{[0,2)}, \StrategyMatrixGlobalPerAttribute{2}_{[2,4)}),$ 
$(\StrategyMatrixGlobalPerAttribute{1}_{[0,2)}, \StrategyMatrixGlobalPerAttribute{2}_{[4,6)}),$ 
$(\StrategyMatrixGlobalPerAttribute{1}_{[2,3)}, \StrategyMatrixGlobalPerAttribute{2}_{[1,2)}),$ 
$(\StrategyMatrixGlobalPerAttribute{1}_{[2,3)}, \StrategyMatrixGlobalPerAttribute{2}_{[2,4)}),$ 
$(\StrategyMatrixGlobalPerAttribute{1}_{[0,2)}, \StrategyMatrixGlobalPerAttribute{2}_{[4,6)})\}$
\end{example}
\end{techreport}

\begin{vldbpaper}
    \vspace{-2.5mm}
\end{vldbpaper}

\section{Evaluation}
\label{sec:evaluation}
\begin{figure*}[ht]
\centering

\subfloat[BFS - Adult (Age attribute)]{%
  \includegraphics[width=0.30\linewidth]{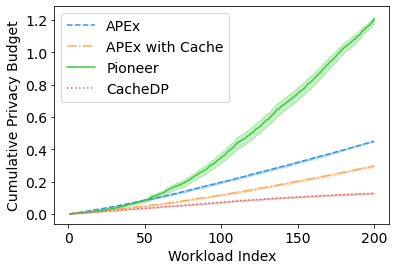}%
  \label{fig:single_bfs}%
}\quad
\subfloat[DFS - Adult (Country attribute)]{%
  \includegraphics[width=0.30\linewidth]{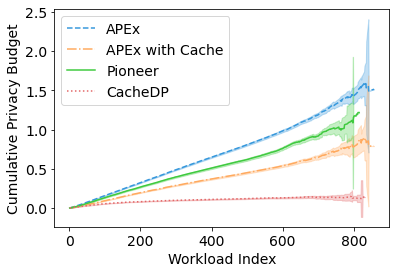}%
  \label{fig:single_dfs}%
}\quad
\subfloat[RRQ - Synthetic (1D)]{%
  \includegraphics[width=0.30\linewidth]{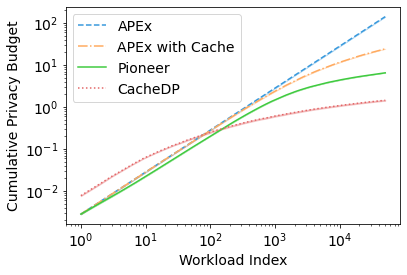}%
  \label{fig:pioneer}%
} 

\subfloat[BFS - Taxi (Lat, Long attributes)]{%
  \includegraphics[width=0.30\linewidth]{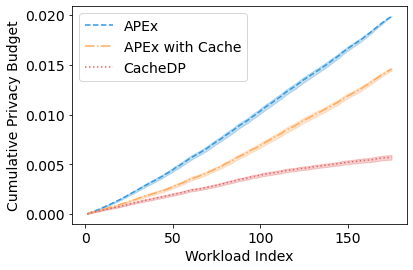}%
  \label{fig:location_bfs}%
}\quad
\subfloat[DFS - Taxi (Lat, Long attributes)]{%
  \includegraphics[width=0.30\linewidth]{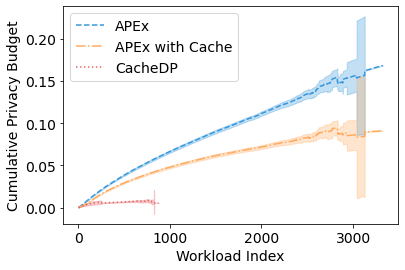}%
  \label{fig:location_dfs}%
}\quad
\subfloat[IDEBench - Planes (multiple attributes)]{%
  \includegraphics[width=0.30\linewidth]{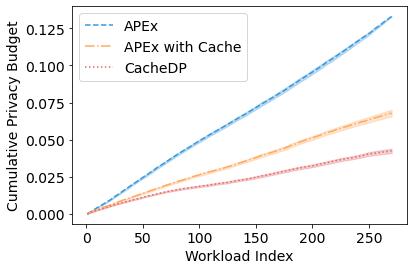}%
  \label{fig:HD_case}%
}
\caption{
Average cumulative privacy budget comparison between CacheDP and baselines (APEx, {APEx with  cache}, Pioneer).
}\label{fig:privacycomparison}
\end{figure*}

We conduct a thorough experimental evaluation of \sysname. 
We focus on our primary goal, namely, reducing the cumulative privacy budget of interactive workload sequences over baseline solutions (Section~\ref{sec:privacy_budget_comp}), while still meeting the accuracy requirements (Section~\ref{sec:accuracy-eval}) and incurring low overheads (Section~\ref{sec:overheads}).
We assess how often each module is used, and quantify its impact on the privacy budget, through our ablation study in Section~\ref{sec:ablation_study}.

\subsection{Experimental Setup}
\label{subsec:experimental-setup}
\subsubsection{Baseline Solutions}
We consider a number of baseline, accuracy-aware solutions from the literature to compare with \sysname.

\textbf{APEx~\cite{ge_apex:_2017}}: \textit{APEx} is a state-of-the-art accuracy-aware interactive DP query engine. 
\begin{techreport}
\textit{APEx} consumes accuracy requirements in the form of an $(\alpha, \beta)$ bound (see Definition~\ref{def:mm}). 
\end{techreport}
\textit{APEx} treats all workload queries separately and has no cache of previous responses.

\textbf{APEx with cache}: We simulate \textit{APEx} with a naive cache of all past workloads and their responses. 
If a client repeats a workload asked in the past by any client, with the same or a lower accuracy requirement, we do not count its privacy budget towards the cumulative budget spent by \textit{APEx with cache}.

\textbf{Pioneer~\cite{pioneer}}: \textit{Pioneer} is a DP query engine that incorporates a cache of previous noisy responses to save the privacy budget on future queries. 
\begin{techreport}
\textit{Pioneer} expects the accuracy over its responses in the form of a target variance.
\end{techreport}
Since \textit{Pioneer} can only answer single range queries, we  decompose all workloads into single queries for our evaluation, and let \textit{Pioneer} answer them sequentially.

\begin{table}[t]
    \centering
    \begin{tabular}{|l|c|c|c|c|} \hline
        {Dataset} & {\textsc{Adult}~\cite{adult_dataset}} & \textsc{Taxi}~\cite{nyc_dataset} &  \textsc{Planes}~\cite{planes_dataset, eichmann2020idebench}\\ \hline
        Size & $48842  \times 14$ & $1028527 \times 19$ &  $500,000 \times 12$ \\ \hline
        {Tasks }& BFS (Age) & BFS (Lat, Long) & IDEBench \\ 
        (Attributes) & DFS (Country) & DFS (Lat, Long) & (8 out of 12) \\ \hline
    \end{tabular}
    \caption{\miti{Datasets, their sizes, and associated tasks, with the attributes or number of attributes used in each task.}} 
    \label{tab:dataset-tasks}
\end{table}
\vspace{-5pt}

\begin{vldbpaper}
\subsubsection{Datasets and Tasks}
\label{subsubsec:experiment-setup}
\miti{
    In Table~\ref{tab:dataset-tasks}, we outline the datasets used and the tasks that each dataset is used in. 
    A common data exploration task involves traversing a decomposition tree over the domain~\cite{Zhang2016_privtree}. We construct our workloads through either a breadth-first search (BFS) or a depth-first search (DFS); both of these tasks are executed over a single attribute or a pair of correlated attributes (Latitude and Longitude from the \textsc{Taxi} dataset). 
    We also replicate the evaluation of \textit{Pioneer}~\cite{pioneer} through randomized single range queries (RRQ) over a single attribute of a synthetic dataset. 
    We use Eichman et al.'s benchmarking tool, namely IDEBench~\cite{eichmann2020idebench}, to construct a sequence of interactive multi-attribute workloads. 
    }

\miti{
    We model multiple data analysts querying each system, as clients. 
    We run the BFS, DFS and IDEBench tasks with multiple clients. 
    We schedule the clients' interactions with each system by randomly sampling clients, with replacement, from the set of clients until no queries remain.
    A client chooses the accuracy requirements and task parameters for each run of the experiment independently and at random; we detail these choices in the full paper~\cite{dpcacheextended}. 
    }
\end{vldbpaper}

\begin{techreport}
\subsubsection{Datasets and Interactive Exploration Tasks}
We use three datasets: the 1994 US Census data in the {\textsc{Adult}} dataset~\cite{adult_dataset} ($48842$ rows $\times 14$ attributes), a log of NYC yellow taxi trip records from 2015 in the \textsc{Taxi} dataset~\cite{nyc_dataset} ($1028527 \times 19$), and US domestic flight records in the \textsc{Planes} dataset~\cite{planes_dataset, eichmann2020idebench} ($500,000 \times 12$).
The \textsc{Taxi} dataset is used to model a single strategy tree over a pair of correlated (\textit{Lat, Long}) attributes. Each node on the tree (or a range query) represents a rectangle of area on a map, and each level on the tree splits each node's rectangle into its quarters.

\textbf{BFS and DFS tasks:}
A common data exploration task involves traversing a decomposition tree over the domain~\cite{Zhang2016_privtree}, by progressively asking more fine-grained queries over a subset of the domain.
In each iteration, the analyst decomposes each query in the past workload whose noisy response satisfies a certain criteria, into $k$ new children queries on the attribute decomposition tree. 
In the BFS task, the analyst explores only past queries with a sufficiently \textit{high} noisy count.
The BFS task thus returns the smallest subsets of a domain that are sufficiently populated.

Whereas in a DFS task, the analyst focuses on past queries with a sufficiently \textit{low non-zero} count (i.e. underrepresented subgroups).
A DFS task terminates if a query's noisy count is non-zero and falls within the low DFS threshold range. 
In the DFS task, when the analyst reaches a leaf node without finding a node with a sufficiently low count, they backtrack a random number of steps up the tree, and resume the search with the second smallest node at that level. 
We group the range queries that satisfy the BFS or DFS criteria, into a single workload per level of the tree.

\textbf{Randomized range queries (RRQ) over synthetic data:}
We replicate the evaluation of \textit{Pioneer}~\cite{pioneer} through $50,000$ randomized range queries over a synthetic dataset. 
We fix a domain size of $m=1,000$ such that each range query is contained in $(0, m)$. 
Each query is in the form of $(s, s+\ell)$ where $s$ and $\ell$ are both selected from a normal distribution with the following mean ($\mu$) and standard deviation ($\sigma$): 
$\mu_s=500$, $\sigma_s=10$, $\mu_\ell=320$, and $\sigma_\ell=10$.
The accuracy requirement is supplied as an expected square error, which is also selected from a normal distribution with $\mu=250000$,$\sigma=25000$. 

\textbf{IDEBench for Multi-Attribute queries:} Eichman et al. develop a benchmarking tool to evaluate interactive data exploration systems~\cite{eichmann2020idebench}. We use this tool to construct a sequence of exploration workloads for a multi-attribute case study over the \textsc{Planes} dataset. 
Specifically, we extract the SQL workloads for their $1:N$ case where each query triggers an additional $N$ dependent queries~\cite{eichmann2020idebench}. 
We simplify these queries for easy integration with our prototype.

\subsubsection{Client Modelling}
We model multiple data analysts querying the aforementioned systems, as clients. 
We schedule the clients' interactions with each system by randomly sampling clients, with replacement, from the set of clients until no queries remain.
A client chooses the accuracy requirements and task parameters for each run of the experiment independently and at random.

\subsubsection{Experiment setup.} 
\label{subsubsec:experiment-setup}
    We run the BFS and DFS tasks with $c=25$ clients and the IDEBench task with $c=10$ clients. 
    (We only run the RRQ task with a single client, to precisely replicate the Pioneer paper.
    Since \textit{Pioneer} can only work with a single attribute, we do not run it for tasks 
    over the \textsc{Taxi} or \textsc{Planes} datasets.) 
    We detail the accuracy requirements and task parameters in the full paper. 
    The BFS and DFS tasks are conducted over the \textit{Age} and \textit{Country} attributes of the \textsc{Adult} dataset respectively.
    Both tasks are also run over the \textit{Latitude (Lat)} and \textit{Longitude (Long)} attributes of the \textsc{Taxi} dataset. 
    Each client randomly draws the minimum threshold for their BFS and the maximum threshold for their DFS. 
    The BFS, DFS and IDEBench tasks consume  $(\alpha, \beta)$ accuracy requirements, with fixed $\beta=0.05$ across all clients. 
    Each client randomly selects $\alpha = \alpha_s|D|$ for $\alpha_s \in [0.01, 0.16]$, with step size 0.05.
\end{techreport}

\subsection{End-to-end Comparison}

\subsubsection{Privacy Budget Comparison}\label{sec:privacy_budget_comp}
We repeat each interactive exploration task $N=100$ times, and we compute the average cumulative privacy budget for our solution \sysname and baselines (\textit{APEx}, \textit{APEx with cache}, \textit{Pioneer}) over all $N$ experiment runs. 
We plot the mean and 95\% confidence intervals in Figure~\ref{fig:privacycomparison}. 
We have two hypotheses: 
\begin{enumerate}[label=H\arabic*, leftmargin=*]
    \item \label{hyp:h1} The baselines arranged in order of increasing cumulative privacy budget should be: \textit{APEx}, \textit{APEx with cache}, Pioneer. 
    \begin{enumerate}[label=H1.\arabic*, leftmargin=*]
        \item \label{hyp:h1-1} \textit{Pioneer} should outperform \textit{APEx with cache},
        since \textit{Pioneer} saves privacy budget over any related workloads, whereas \textit{APEx with cache} only saves privacy budget over repeated workloads. 
        \item \label{hyp:h1-2} Baselines with a cache (\textit{APEx with cache}, Pioneer) should outperform the baseline without a cache (\textit{APEx}).
    \end{enumerate}
    \item \label{hyp:h2} \sysname should outperform all baselines. 
\end{enumerate}

First, we observe that hypothesis~\ref{hyp:h1} holds for all tasks, other than the single-attribute BFS and DFS tasks (Figures~\ref{fig:single_bfs},~\ref{fig:single_dfs}).
Since we decompose each workload into multiple single range queries for Pioneer, this sequential composition causes it to perform worse than \textit{APEx} without a cache in the BFS task, and thus hypothesis~\ref{hyp:h1-1} is violated. 
For the same reason, \textit{APEx with cache} outperforms Pioneer for the single-attribute DFS task, and so, hypothesis~\ref{hyp:h1-2} is violated. 
Though, we note that hypothesis~\ref{hyp:h1-2} holds for the RRQ task (Figure~\ref{fig:pioneer}).
Our Pioneer implementation replicates a similar privacy budget trendline to the original paper~\cite[Figure 15]{pioneer}. 

Hypothesis~\ref{hyp:h2} holds for all tasks, and the cumulative privacy budget spent by \sysname scales slower per query, by \textit{at least} a constant factor, over all graphs.
\miti{We note that in the RRQ task (Figure~\ref{fig:pioneer}), \sysname~spends more privacy budget \textit{upfront} than the other systems, since these systems use the simpler Laplace Mechanism, which is optimal for single range queries over our underlying Matrix Mechanism. 
However, any upfront privacy budget spent by \sysname is used to fill the cache, which yields budget savings over a large number of workloads, as \sysname requires an order of magnitude less cumulative privacy budget than the best baseline (\textit{Pioneer}).} 
We observe that even in the computationally intensive tasks due to larger data vectors for two attributes (Figures~\ref{fig:location_bfs},~\ref{fig:location_dfs}) and multiple attributes (Figure~\ref{fig:HD_case}), \sysname outperforms the best baseline (\textit{APEx with cache}), by at least a factor of 1.5 for Figure~\ref{fig:HD_case}. 

For both DFS tasks  (Figures~\ref{fig:single_dfs},~\ref{fig:location_dfs}), since each experiment can terminate after a different number of workloads have been run, we observe large confidence intervals for higher workload indices for each system. 
\sysname simply returns cached responses to a workload if they meet the  accuracy requirements, whereas our simulation for \textit{APEx with cache} resamples noisy workload responses 
and may traverse the tree again in a possibly different path. 
Relaxed accuracy requirements from different clients can lead to frequent re-use of our cache (Section~\ref{subsubsec:experiment-setup}),  
and thus, we find that in Figure~\ref{fig:location_dfs},  \sysname ends the DFS exploration faster than \textit{APEx with cache}.

\subsubsection{Accuracy Evaluation}
\label{sec:accuracy-eval}
We measured the empirical error of the noisy responses returned by all systems and found that they  meet the the clients' $(\alpha, \beta)$ accuracy requirements. 
Cached responses used by \sysname commonly exceed the accuracy requirements. 
Specifically, when \textit{all} strategy responses are free, \sysname will always return the most accurate cached response for each strategy query, even if the current workload has a poorer $\alpha$.

\subsubsection{Overhead Evaluation}\label{sec:overheads}
We compute the following storage and computation overheads for all systems, averaged over all $N$ experiment runs: (1) \textit{cache size} in terms of total number of cache entries at the end of a run, and (2) \textit{workload runtime}, averaged over all workloads in a run.
In Table~\ref{tab:overhead-eval}, we present these overheads for representative tasks. 
\miti{(Since our simulation for \textit{APEx with cache} only differs from \textit{APEx} by a recalculation of the privacy budget (Section~\ref{subsec:experimental-setup}), the latter has the same runtime as the former.) 
} 
Our cache size is limited by the number of nodes on our strategy tree, and so for the RRQ task, which has $50k$ workloads, \sysname has a smaller cache size than the baselines. 
Whereas, in tasks with fewer workloads, such as 
the IDEBench task, our PQ module inserts more strategy query nodes into the cache, and thus, significantly increases our cache size over the baselines. Nevertheless, since our cache entries only consist of 4 floating points (32B), even a cache with  $\approx 25k$ entries would be reasonably small in size ($\approx 800kB)$. 

\miti{
    In terms of runtime, \sysname~only takes a few seconds per workload for single-attribute tasks, such as the DFS task, thereby matching other cached baselines. 
    As the number of attributes increases in the IDEBench task, 
    the cumulative cache size and runtime of \sysname scales linearly. 
    Specifically, IDEBench workloads require computations over a larger data vector that spans many attributes.} 
    \begin{vldbpaper}
        \miti{(We include a graph for these variables and discuss performance optimizations in the extended paper.)}
    \end{vldbpaper}
    \miti{Yet, non-optimized \sysname only takes around 6 minutes per workload for the IDEBench task, and performs slightly better than APEx with cache.
    }

\begin{techreport}
\begin{figure}
    \centering
    \includegraphics[scale=0.5]{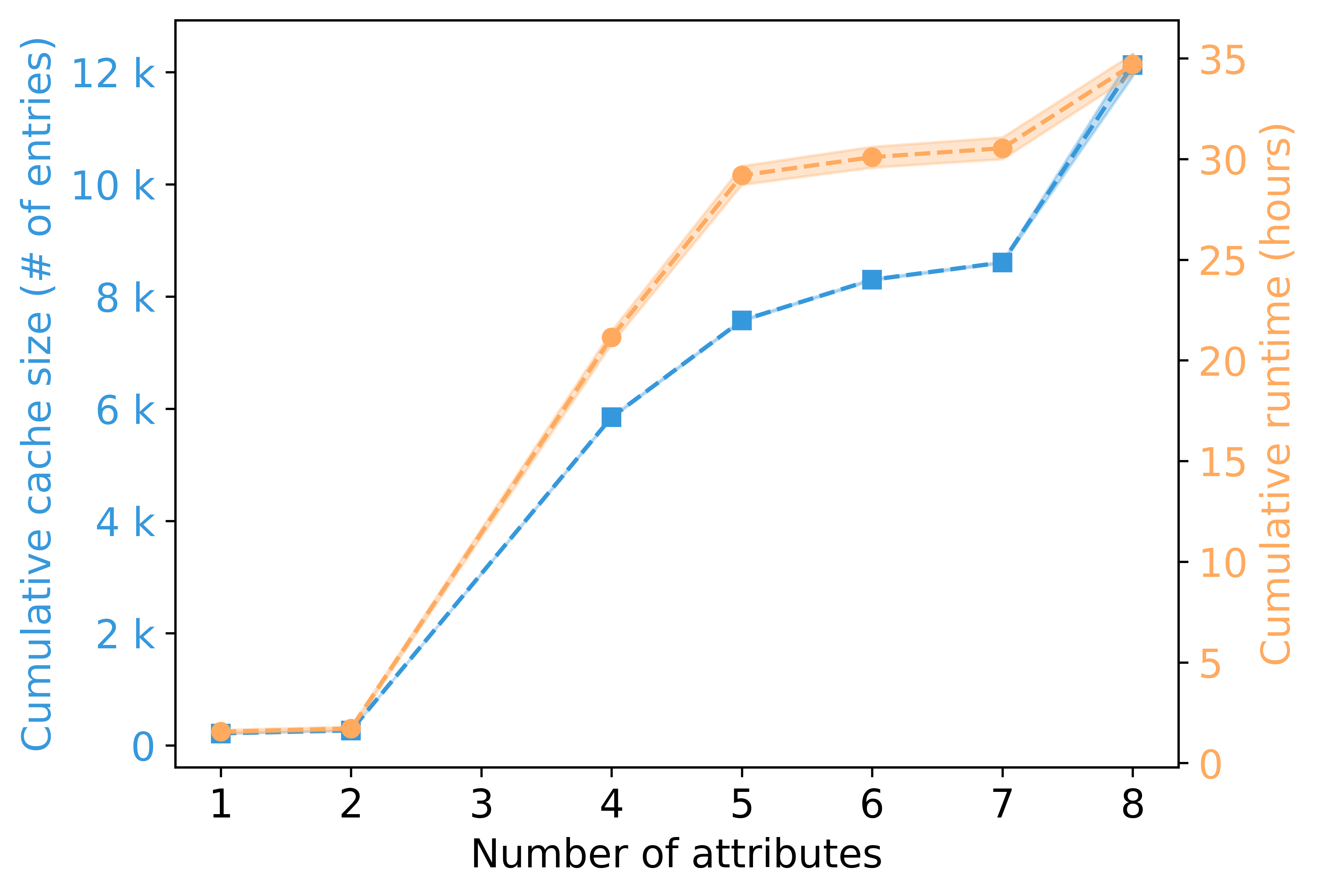}
    \caption{Average cumulative cache size and runtime of \sysname, versus the number of attributes, for the IDEBench task.}
    \label{fig:idebench-performance}
\end{figure}
We illustrate how the cumulative cache size and runtime of \sysname scales with the number of attributes for the IDEBench task in Figure~\ref{fig:idebench-performance}.
We note that the x-axis corresponds to a total of 27 IDEBench workloads. We consider multiple clients in Table~\ref{tab:overhead-eval}, whereas we only consider one client in Figure~\ref{fig:idebench-performance}.\footnote{Since this single client will not experience fully cached, accurate workload responses, we find an order of magnitude worse runtime overheads than simply by multiplying the average in Table~\ref{tab:overhead-eval} by the number of IDEBench workloads. Thus this plot shows worst-case runtimes for the IDEBench task.} We also note that increasing the number of records in the database only impacts the size of the data vector in the pre-processing stage, and does not impact the size of the cache.

We can see that both observed variables scale approximately linearly with the number of attributes.
The cache size can be limited by identifying inaccurate cache entries that can be replaced, while inserting new noisy responses, in each module.
The MC simulation forms the bottleneck for the runtime, when the number of attributes remains small, whereas the matrix multiplication becomes the main bottleneck as the number of attributes increases.
The MC simulation can be avoided entirely, by expressing the desired accuracy through expected variances instead of the $(\alpha, \beta)$ requirements.
Our implementation can be optimized with fast matrix multiplication algorithms to  achieve smaller runtimes. 
\end{techreport}

\begin{vldbpaper}
    \begin{table*}[t]
        \parbox{.48\linewidth}{
        \centering
        \resizebox{\dimexpr\columnwidth}{!}{%
        \begin{tabular}{|l|ll|ll|}
        \hline
            System     & \multicolumn{2}{c|}{Cache size (entries)} & \multicolumn{2}{c|}{Runtime (s)} \\ \cline{2-5}
            & \multicolumn{1}{l|}{RRQ}    & IDEBench  & \multicolumn{1}{l|}{DFS - Adult}  & IDEBench \\ \hline 
            APEx (cache) & \multicolumn{1}{l|}{3118$\pm$4} & 6540$\pm$0  & \multicolumn{1}{l|}{2.71$\pm$0.02} & 456$\pm$70 \\ \hline
            Pioneer & \multicolumn{1}{l|}{3118$\pm$4} & - & \multicolumn{1}{l|}{1.6$\pm$0.2ms}    & - \\ \hline
            CacheDP & \multicolumn{1}{l|}{1998$\pm$0} & 23666$\pm$1000 & \multicolumn{1}{l|}{2.82$\pm$0.2}    &  338$\pm$20 \\ \hline
        \end{tabular}}
        
        \caption{Cache size and workload runtime comparison. }
        \label{tab:overhead-eval}}
        \hfill
\parbox{.48\linewidth}{
\resizebox{\dimexpr\columnwidth}{!}{%
\begin{tabular}{|l|l|l|l|l|}
\hline
         & BFS & DFS & RRQ  & IDEBench \\ \hline
Free & $164.8\pm 0.9$ & $591\pm 6$ & $49801\pm 3$ & $216\pm 1$\\
MMM & $2.1\pm 0.1  $ & $2.25\pm 0.09  $ & $\mathbf{137\pm 3}  $ & $15.1\pm 0.2 $ \\
RP & $14.7\pm 0.5 $ & $21.1\pm 0.5 $ & $1.4\pm 0.1    $ & $\mathbf{31\pm 1} $ \\
SE & $\mathbf{18.5\pm 0.6 }$ & $\mathbf{42\pm 1 }$ & $60\pm 4   $ & $7.8\pm 0.2 $\\ \hline
\end{tabular}}
\caption{\miti{Average number of times each module was chosen for each task; most frequently chosen modules are in bold.}
}\label{tab:module_usage_counts}}
\end{table*} \vspace{-10pt}
\end{vldbpaper}

\begin{techreport}
    \begin{table*}[h]
        \centering
        \resizebox{\dimexpr\linewidth}{!}{%
        \begin{tabular}{|l|ll|ll|ll|ll|ll|ll|}
    \hline
    System     & \multicolumn{2}{l|}{BFS - Adult (Age)}            & \multicolumn{2}{l|}{DFS - Adult (Country)}      & \multicolumn{2}{l|}{RRQ - Synthetic (1D)}            & \multicolumn{2}{l|}{BFS - Taxi (Lat,Long)}                       & \multicolumn{2}{l|}{DFS - Taxi (Lat,Long)}                & \multicolumn{2}{l|}{IDEBench} \\ \cline{2-13} 
               & \multicolumn{1}{l|}{Cache Entries}      & Runtime (s)       & \multicolumn{1}{l|}{Cache Entries} & Runtime (s)       & \multicolumn{1}{l|}{Cache Entries}      & Runtime (s)    & \multicolumn{1}{l|}{Cache Entries}       & Runtime (s)         & \multicolumn{1}{l|}{Cache Entries} & Runtime (s)        & \multicolumn{1}{l|}{Cache Entries}  & Runtime (s) \\ \hline
    APEx       & \multicolumn{1}{l|}{-}               & 3.7$\pm$0.2 & \multicolumn{1}{l|}{-}          & 2.71$\pm$0.02 & \multicolumn{1}{l|}{-}               & 0.05$\pm$0 & \multicolumn{1}{l|}{-}                & 111$\pm$8 & \multicolumn{1}{l|}{-}          & 25$\pm$1 & \multicolumn{1}{l|}{-}           & 456$\pm$70       \\
    APEx Cache & \multicolumn{1}{l|}{123$\pm$1} & 3.7$\pm$0.2 & \multicolumn{1}{l|}{81$\pm$0}   & 2.71$\pm$0.02 & \multicolumn{1}{l|}{3118$\pm$4} & 0.05$\pm$0 & \multicolumn{1}{l|}{668$\pm$14} & 111$\pm$8 & \multicolumn{1}{l|}{2515$\pm$0} & 25$\pm$1 & \multicolumn{1}{l|}{6540$\pm$0}           & 456$\pm$70      \\
    Pioneer    & \multicolumn{1}{l|}{117$\pm$1} & 4.8$\pm$0.2ms       & \multicolumn{1}{l|}{80$\pm$0}   & 1.6$\pm$0.2ms       & \multicolumn{1}{l|}{3118$\pm$4} & 0.01$\pm$0 & \multicolumn{1}{l|}{-}                & -               & \multicolumn{1}{l|}{-}          & -              & \multicolumn{1}{l|}{-}           & -       \\
    CacheDP    & \multicolumn{1}{l|}{145$\pm$3} & 7.5$\pm$0.4 & \multicolumn{1}{l|}{81$\pm$0}   & 2.8$\pm$0.2  & \multicolumn{1}{l|}{1998$\pm$0}      & 0.05$\pm$0 & \multicolumn{1}{l|}{3669$\pm$0}       & 23$\pm$1  & \multicolumn{1}{l|}{3669$\pm$0} & 15.3$\pm$0.5 & \multicolumn{1}{l|}{23666$\pm$1000}           & 338$\pm$20       \\ \hline
    \end{tabular}
        }
        \caption{Cache size and workload runtime comparison. }
        \label{tab:overhead-eval}
    \end{table*}
\end{techreport}

\subsection{\miti{Ablation study}}\label{sec:ablation_study}
\begin{techreport}

\begin{table}[t]
\resizebox{\dimexpr\columnwidth}{!}{%
\begin{tabular}{|l|l|l|l|l|}
\hline
         & BFS & DFS & RRQ  & IDEBench \\ \hline
Free & $164.8\pm 0.9$ & $591\pm 6$ & $49801\pm 3$ & $216\pm 1$\\
MMM & $2.1\pm 0.1  $ & $2.25\pm 0.09  $ & $\mathbf{137\pm 3}  $ & $15.1\pm 0.2 $ \\
RP & $14.7\pm 0.5 $ & $21.1\pm 0.5 $ & $1.4\pm 0.1    $ & $\mathbf{31\pm 1} $ \\
SE & $\mathbf{18.5\pm 0.6 }$ & $\mathbf{42\pm 1 }$ & $60\pm 4   $ & $7.8\pm 0.2 $\\ \hline
\end{tabular}}
\caption{\miti{Average number of times each module was chosen for each task; most frequently chosen modules are in bold.}
}\label{tab:module_usage_counts}
\end{table}
\end{techreport}
\miti{
Each of our modules contribute differently to the success of our system across different workloads. We conduct an ablation study in two parts analyzing our modules. 
First, we analyze the frequencies at which each module is selected to answer a workload, 
and second, we run a study to quantify the impact of each module on the cumulative privacy budget.
We begin with our frequency analysis, noting that a module is chosen to answer a given workload if it is estimated to cost the lowest privacy budget. 
We only include the MMM, RP, and SE modules in this analysis, since the PQ module is not involved in the cost estimation stage. 
We present the number of times each module is chosen to answer a workload in each of the BFS, DFS, RRQ and IDEBench tasks, averaged over $N=100$ runs, in Table~\ref{tab:module_usage_counts}.
If MMM reports $\epsilon=0$ for a workload, \sysname simply uses MMM to answer the workload using cached responses, and it does not run RP or SE modules. 
We thus separately record the number of free workloads per task in the first row of Table~\ref{tab:module_usage_counts}. 
First, we observe that most workloads for each task are free, indicating that using solely the cached strategy responses, \sysname can successfully answer most workloads for these tasks. 
}

\miti{
Second, considering all non-free workloads, each of the modules are used the most frequently for at least one task. 
SE is chosen most frequently for the BFS and DFS tasks, answering $52\%$ and $64\%$ of non-free workloads respectively. 
Furthermore, for many workloads in these tasks, we observed that MMM had $\epsilon >0$ cost, but under SE, these workloads became free ($\epsilon=0$).
RP is chosen most frequently for the IDEBench task ($57\%$), whereas MMM is used most frequently for the RRQ task ($69\%$). 
Thus, we can see that each module plays a role in \sysname's performance in one or more tasks.
}

\miti{
We also run a study to quantify savings in the cumulative privacy budget due to each module. 
We rerun our single-attribute tasks (BFS, DFS, RRQ) while disabling each of our modules (MMM, SE, RP, PQ) one at a time, and present the cumulative privacy budget consumed by each such configuration in Figure~\ref{fig:ablation_study}.
The standard configuration consists of all modules turned on. 
(Turning an effective module off should lead to an increase in the cumulative privacy budget, in comparison to the standard configuration.)
First, we observe that the standard configuration performs the best in all three tasks, while considering CI overlaps. 
Therefore, data analysts need not pick which modules should be turned on in order to answer a workload sequence with the lowest privacy budget.
Second, the PQ module significantly lowers the cost for the BFS and RRQ task, proactively fetching all ($\approx 12$) queries in a BFS workload at the cost of one strategy node.
Third, the RP module also lowers the cost for BFS, when the same workload is repeated by other clients. 
}

\miti{Fourth, turning the SE module off only contributes to minor differences in the cumulative privacy budget ($\ConsumedPrivacyBudget$). 
However, the reader may expect that turning the SE module off would lead to a higher $\ConsumedPrivacyBudget$, since based on the frequency analysis, the SE module is most frequently chosen to answer non-zero workloads for the BFS and DFS tasks.
}
\begin{vldbpaper}
    \miti{We reconcile this discrepancy with the observation that the SE module provides savings on earlier workloads, through constrained inference, at the cost of a less accurate cache to answer future workloads. }
\end{vldbpaper}
\begin{techreport}
    \miti{
    We reconcile this discrepancy as follows.
    When the SE module is chosen, constrained inference reduces the privacy cost for paid strategy queries. 
    As a result, the cache entries for these queries remain at a lower accuracy than if MMM or RP had been used to answer them, and they may not be reusable for later workloads. 
    Thus \sysname may need to obtain noisy responses for these queries, later on, at a higher accuracy. The SE module thus provides savings on earlier workloads at the cost of a less accurate cache to answer future workloads. }
\end{techreport}
\miti{In summary, different tasks exploit different modules and the standard configuration incurs the least privacy budget, and thus data analysts need not turn off any modules. }

\begin{figure}[t]
\centering
  \includegraphics[scale=0.40]{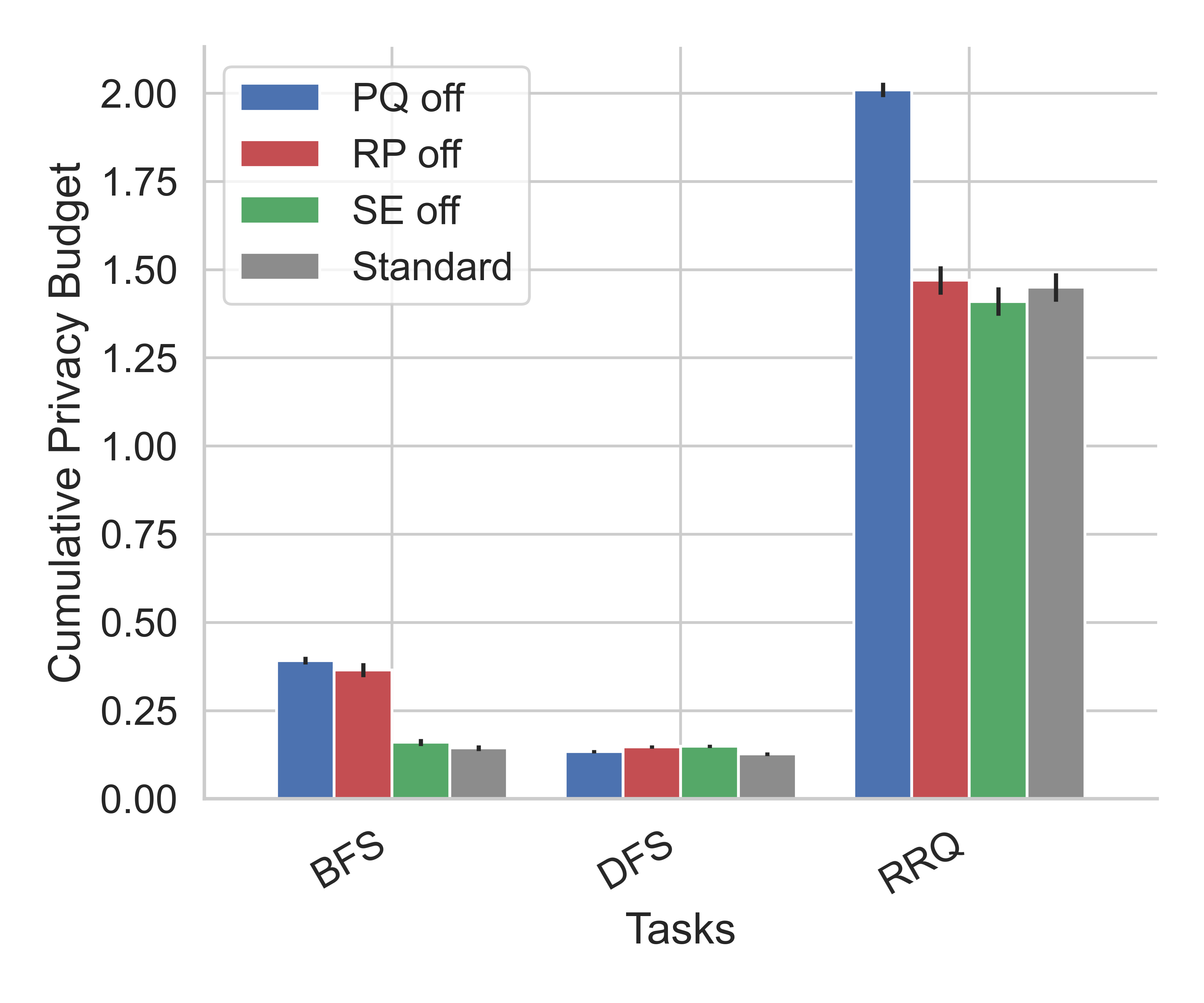}
  \vspace{-5mm}
\caption{Ablation study over PQ, RP, SE modules of \sysname }\label{fig:ablation_study}
\end{figure}

\section{Related work} 
\label{sec:related_work}
\miti{
    Constrained inference techniques have been applied in the non-interactive DP setting to improve the accuracy of noisy query responses~\cite{Zhang2016_privtree,qardaji2013understanding} and in synthetic data generators~\cite{dpautogan21,ge2020kamino,mckenna_graphical-model_2019} to infer consistent answers from a data model built through noisy measurement queries. 
    However, these systems do not provide any accuracy guarantee on the inferred responses. 
    If the analyst desires a more accurate response than the synthetic data can offer, no privacy budget remains to improve the query answer~\cite{benchsyn21}.
    Our work applies DP constrained inference in an interactive setting so that we can spend the privacy budget on queries that the analyst is interested in and meet their accuracy requirements.   
    On the other hand, existing accuracy-aware DP systems for data exploration~\cite{ge_apex:_2017,gupt12}, releasing data~\cite{gaboardi2016psi,vip22}, or programming frameworks~\cite{lobovesga2019_dpella, fitnessforuse21} do not exploit historical query answers to save privacy budget on a given query. 
    We design a cache structure and inference engine extending one of these accuracy-aware systems, APEx~\cite{ge_apex:_2017}. 
}

Peng et al.'s Pioneer~\cite{pioneer} is the most relevant work that uses historical query answers to obtain accurate responses to upcoming single range queries. 
\eat{However, our work differs in several ways. }
However, \sysname can handle workloads with multiple queries.
Second, it supports multiple, complex DP mechanisms and chooses the mechanism that uses the least privacy budget for each new workload. 
Third, our PQ module (Section~\ref{subsec:proactive}) proactively fetches certain query responses that can be used later at no additional cost.
Finally, \sysname can answer multi-attribute queries through our extended ST module (Section~\ref{sec:multi-attribute-case}). 

\miti{
Our key modules are built on top of prior work (e.g., Li et al.'s Matrix Mechanism~\cite{li2015matrix}, Koufogiannis et al.'s Relax Privacy Mechanism~\cite{relaxed_privacy}), such that existing interactive DP systems that make use of these mechanisms (e.g. PrivateSQL~\cite{privatesql}, APEx~\cite{ge_apex:_2017}) do not have to make significant changes; these systems can include a relatively light-weight cache structure and cache-aware version of the DP mechanisms. 
Integrating a structured, reusable cache with these mechanisms has its own technical challenges, such as the Cost Estimation problem (Section~\ref{subsubsec:simplified-priv-optimizer}), Full Rank Transformation problem (Section~\ref{sec:fullranktransform}), 
as well as optimally reusing the cache (Section~\ref{subsec:strategy-expander}) and filling it (Section~\ref{subsec:proactive}). 
}

\begin{techreport}
\label{sec:future-work}
\section{Future work:}
\textit{CacheDP} can be extended to answer top-$k$ or iceberg counting queries studied in APEx, as well as simple aggregates such as means, by integrating the query processing engine of APEx to transform these queries to raw counting queries. 
Beyond counting queries, providing differential privacy over complex SQL queries, such as joins and group by operators, is a challenging problem, as studied in prior work~\cite{privatesql,wei_and_yi,googledp}. 
The global sensitivity of SQL queries involving joins is unbounded. 
To tackle this challenge, the existing well-performed DP mechanisms~\cite{privatesql,wei_and_yi,googledp} for these queries require a data-dependent transformation (e.g., truncation of the data) in order to bound the sensitivity of the query. 
Thus,  the accuracy of these  mechanisms depend on the data, and searching the minimum privacy budget to achieve the desired accuracy bound is non-trivial, which is an important research direction.
\end{techreport}

\begin{vldbpaper}
    \vspace{-5mm}
\end{vldbpaper}
\section{Conclusion}
\label{sec:conclusion}
We build a usable interactive DP query engine, \sysname, that uses a structured DP cache to achieve privacy budget savings commonly seen in the non-interactive model. 
\sysname supports data analysts in answering data exploration workloads accurately, without requiring them to have any DP knowledge.
Our work provides researchers with a methodology to address common challenges while integrating DP mechanisms with a DP cache, such as, cache-aware privacy budget estimation (MMM), filling the cache at a low privacy budget (PQ), and maximizing cache reuse (SE).

\begin{vldbpaper}
    \vspace{-1em}
\end{vldbpaper}
\section{Acknowledgements}
We thank NSERC for funding our work through the Postgraduate Scholarship-Doctoral program, grant CRDPJ-534381, and a Discovery Grant. 
We also thank the Royal Bank of Canada for supporting our work. 
This work benefited from the use of the CrySP RIPPLE Facility at the University of Waterloo.

\begin{vldbpaper}
\balance
\end{vldbpaper}
\bibliographystyle{ACM-Reference-Format}
\bibliography{refs}
\begin{techreport}
\appendix

\section{Proofs}
\subsection{End-to-end Privacy Proof}
\subsubsection{Proof of Theorem~\ref{thm:end-to-end_privacy}} Recall the theorem states that \sysname, as defined in Algorithm~\ref{algo:end-to-end}, satisfies \TotalPrivacyBudget-DP.
\begin{proof} (sketch)
We begin by addressing the cost estimation phase (lines 4-10). The cost estimation phases are independent of the data. In addition, each DP mechanism (MMM or RP) with its corresponding chosen strategy ensures $\epsilon_i$-DP (Proposition~\ref{prop:privacy_mmm}, ~\cite[Theorem 1A]{relaxed_privacy}).  At line 11, we check if running the chosen DP mechanism (MMM or RP with the corresponding chosen strategy) will exceed the total privacy budget $\TotalPrivacyBudget$ by sequential composition~\cite{privacy_book}. We only run the DP mechanism if the total budget is sufficient. This ensures the overall Algorithm~\ref{algo:end-to-end} satisfies  \TotalPrivacyBudget-DP.


\end{proof}

\subsection{MMM Module Proofs}
\subsubsection{Proof of Proposition~\ref{prop:privacy_mmm}} Recall the proposition states that the \textsc{AnswerWorkload} interface of MMM (Algorithm~\ref{algo:mmm}) satisfies $\epsilon$-DP, where $\epsilon$ is the output of this interface.
\begin{proof} (sketch)
When the cache is empty, the privacy of MMM (with optional SE) follows from the matrix mechanisms privacy guarantee in Proposition~\ref{prop:mmdp}. When there are entries in the cache, we split the strategy into the free and paid matrix. The paid matrix is private by the same reasoning as above. The privacy of the free matrix responses and the final concatenation of free and paid responses follow by the post processing lemma of DP~\cite[Proposition 2.1]{privacy_book}. 
\end{proof}
\subsubsection{Proof of Proposition~\ref{prop:mmm-error}}
Given an instant strategy $\StrategyFull=(\StrategyMatrixFree||\StrategyMatrixPaid)$ with a vector of 
$k$ noise parameters $\BList=\BList_{\StrategyMatrixFree} || \BList_{\StrategyMatrixPaid}$, the error to a workload $\WorkloadFull$ using the \textsc{AnswerWorkload} interface of MMM (Algorithm~\ref{algo:mmm}) is 
$\|\WorkloadFull\StrategyFull^+ Lap(\BList)\|   $, where $Lap(\BList)$ draws independent noise from $Lap(\BList[1])$, $\ldots, Lap(\BList[k])$ respectively. We can simplify its expected total square error as
$ \|\WorkloadFull\StrategyFull^+ diag(\BList)\|_F^2$ where $diag(\BList)$ is a diagonal matrix with $diag(\BList)[i,i] = \BList[i]$. 
\begin{proof} The error to the MMM is 
$\|\WorkloadFull \StrategyFull^+(\StrategyFull\DataVector+ Lap(\BList)) - \WorkloadFull\DataVector\| 
= \|\WorkloadFull\StrategyFull^+ Lap(\BList)\| 
$. The expected total square error 
$\mathbb{E}[\|\WorkloadFull\StrategyFull^+ Lap(\BList)\|^2_2]$ can be expanded to $\sum_{i=1}^l  \mathbb{E} [\sum_{j=1}^k (\WorkloadFull\StrategyFull^+[i,j] Lap(\BList[j]))^2]$. As the $k$ noise variables are independent and has a zero mean, then we have the error expression equals to 
\[
\sum_{i=1}^l \sum_{j=1}^k (\WorkloadFull\StrategyFull^+[i,j])^2 \mathbb{E} [Lap(\BList[j]))^2] 
= \sum_{i=1}^l\sum_{j=1}^k (\WorkloadFull\StrategyFull^+[i,j])^2 \BList[j]^2
\]
which is equivalent to $ \|\WorkloadFull\StrategyFull^+ diag(\BList)\|_F^2$.
\end{proof}

\subsubsection{Proof of Theorem~\ref{thm:simplifiedce}} \label{app:mmm-simplified-ce-proof} The optimal solution to simplified CE problem incurs a smaller privacy cost $\epsilon$ than the privacy cost $\epsilon_{\StrategyMatrixFree = \emptyset}$ of the matrix mechanism without cache, i.e., MMM with $\StrategyMatrixFree=\emptyset$. 
\begin{proof} (sketch) 
Let $b^*$ be the noise parameter for $\StrategyFull$ in the matrix mechanism without cache to achieve the desired accuracy requirement. 
We can show that $b^*$ is a valid solution to the simplified CE problem:
setting 
$\BList_{\StrategyMatrixFree}= [c.b \in \Cache  ~|~ c.\mathbf{a} \in \Cache \cap \StrategyFull, c.b \leq b^*]$ 
and $\BList_{\StrategyMatrixPaid}=[b^*~|~ \mathbf{a} \in \StrategyMatrixPaid]$ satisfies the accuracy requirement.  As $\|\StrategyMatrixPaid\|_1 \leq \|\StrategyFull\|_1$, the privacy cost $\epsilon= \frac{\|P\|_1}{b^*}$ for $\PaidNoiseParameter=b^*$
is smaller than $\epsilon_{\StrategyMatrixFree = \emptyset}=\frac{\|\StrategyFull\|_1}{b^*}$. The optimal solution to the simplified CE problem has a smaller or the same privacy cost than a valid solution $\PaidNoiseParameter=b^*$.
\end{proof}

\subsection{Full-rank Transformer (FRT) Proof}
Recall Theorem~\ref{thm:fullrank} states that
Given a global strategy $\StrategyMatrixGlobal$ in a $k$-ary tree structure, and an instant strategy $\StrategyRaw \subseteq \StrategyMatrixGlobal$,  \textsc{transformStrategy} outputs a strategy $\StrategyFull$ that is  full rank  and supports $\StrategyRaw$. We begin by proving the following lemma.

\begin{lemma} \label{lem:add_one_bucket}
Given a global strategy $\StrategyMatrixGlobal$ in a $k$-ary tree structure, and an instant strategy $\StrategyRaw \subseteq \StrategyMatrixGlobal$, running \textsc{getTransformationMatrix}($\StrategyRaw$) in Algorithm~\ref{algo:strategy-transformer-functions}  increases the number of non-empty buckets in $\BucketMatrix$ by at most 1, for each row $\StrategyRaw[i]\in \StrategyRaw$. 
\end{lemma}
\begin{proof}
If the condition in line~\ref{line:tm_if_disjoint} is met the result follows trivially.
We consider the remaining case where $\StrategyRaw[i]$ intersects with at least one bucket.
We recall that all entries in $\StrategyRaw$ represent nodes in a $k$-ary tree.
Since adding and subtracting nodes on a $k$-ary tree always results in a combination of one or more nodes on a $k$-ary tree, we can conclude that at the end of each loop, all of the buckets are disjoint.
For each new row $\StrategyRaw[i]$, if this row intersects with the buckets in $\BucketMatrix$, then $\StrategyRaw[i]$ is either (i) a descendant node of one bucket in $\BucketMatrix$, or (ii) an ancestor node of one or more buckets in $\BucketMatrix$.
For the first case, assume $\StrategyRaw[i]$ is a descendant node of $\mathbb{t} \in \BucketMatrix$. It cannot be the descendant of other buckets, as the other buckets are disjoint with $\mathbb{t}$. 
Then $\mathbb{t}$ will be replaced by $\mathbb{t}\cap\StrategyRaw[i]$ and $\mathbb{t}-\StrategyRaw[i]$ (line~\ref{line:tm_add_intersect}), and hence the size of $\BucketMatrix$ increases by 1. 
For the second case, assume $\StrategyRaw[i]$ is the ancestor node of multiple buckets in $\BucketMatrix$. We denote these buckets as $\{\BucketMatrix[j_1],\ldots,\BucketMatrix[j_k]\}$, then all these buckets will remain the same (line~\ref{line:tm_add_intersect} just adds $\mathbb{t}$), and at most one additional bucket $\StrategyRaw[i] - \sum_{\mathbb{t}\in \BucketMatrix' \wedge \mathbb{t}\cdot \StrategyRaw[i]\neq 0} \mathbb{t}$ is added in line~\ref{line:tm_add_remain}.
\end{proof}

We now prove the main result (Theorem~\ref{thm:fullrank}) that \textsc{transformStrategy} ensures full rank matrices.
\begin{proof}
We begin by discussing how $\StrategyFull$ represents the queries in $\StrategyRaw$.
First, we note that $\StrategyFull$ and $\StrategyRaw$ have the same number of rows.
The only difference is that $\StrategyFull$ represents the queries on the data vector $\DataVector$ where as $\StrategyRaw$ uses $\DataVectorRaw$.
By construction, we have $\StrategyRaw \DataVectorRaw = \StrategyFull \BucketMatrix \DataVectorRaw = \StrategyFull \DataVector$.

Since $\StrategyFull$ is a representation of $\StrategyRaw$ using the only the buckets generated in \textsc{getTransformationMatrix}, the number of columns in $\StrategyFull$ is the same as the number of buckets, $|\BucketMatrix|$.
To show that $\StrategyFull$ is always full rank, we will first show that $|row(\StrategyFull)| \ge |col(\StrategyFull)|$, 
where $|row(\StrategyFull)|$ and $|col(\StrategyFull)|$ represent the number of rows and the number of columns in $\StrategyFull$.  Or equivalently, $|row(\StrategyFull)|\ge |\BucketMatrix|$.
This follows by inductively applying Lemma~\ref{lem:add_one_bucket}. 

Next we show that $rank(\StrategyFull) = |col(\StrategyFull)|$.
We once again show this by induction on the number of rows in $\StrategyRaw$.
In the base case ($|\StrategyRaw|=1$) applying \textsc{transformStrategy}, we get $\StrategyFull = [1]$ and thus $rank(\StrategyFull)=1$.
Now we assume that $rank(\StrategyFull) = |col(\StrategyFull)|$ for some $\StrategyFull$ obtained from $\StrategyRaw$.
We consider adding a new row $r_{new}$ to $\StrategyRaw$ and define $\bar{\StrategyRaw} = \StrategyRaw \cup r_{new}$.
Let $\bar{\StrategyFull}$ represent the result of applying \textsc{transformStrategy} to $\bar{\StrategyRaw}$.

When adding this row there are two possible cases:
First, consider the case where the newly added row, $r_{new}$, can be represented as a linear combination of $\BucketMatrix$. That is the buckets created in \textsc{getTransformationMatrix} are the same for $\StrategyRaw$ and $\bar{\StrategyRaw}$ 
In this case, $\bar{\StrategyFull} = \StrategyFull \cup r'_{new}$.
Thus $|col(\bar{\StrategyFull})|=rank(\bar{\StrategyFull})$ since the number of linearly independent columns (or buckets) did not increase or decrease by adding a row. 
On the other hand, if the newly added row, $r_{new}$, cannot be represented as a linear combination of $\BucketMatrix$, then $r_{new}$ must be linearly independent of $\StrategyRaw$.
Thus, $\bar{\StrategyFull}$ has one additional linearly independent row and thus 
$rank(\bar{\StrategyFull}) = rank(\StrategyFull)+1$.
Furthermore by 
Lemma~\ref{lem:add_one_bucket}, we know that adding a row can add at most one new bucket.
Since we assume  $r_{new}$ cannot be represented as a linear combination of $\BucketMatrix$, this means $|col(\bar{A})|= |col(A)|+1$.
Thus we have that $|col(\bar{A})| = rank(\bar{\StrategyFull})$.
Combining the fact that $|row(\StrategyFull)| \ge |col(\StrategyFull)|$ and  $rank(\StrategyFull) = |col(\StrategyFull)|$, it follows that $\StrategyFull$ is full rank.
\end{proof}

\subsection{PQ Module Proofs} 
\label{app:proacitveproof}
\subsubsection{Proof of Lemma~\ref{lemma:subtree-norm-root}}    Recall the lemma states that the $L_1$ norm of the $\StrategyRawPaid$~matrix is equal to the subtree norm of the root of the tree with marked nodes corresponding to $\StrategyRawPaid$:
    \begin{equation}
        \mathcal{S}_{\StrategyRawPaid}(\mathcal{T}\text{.root}) = \|\StrategyRawPaid\|_1
    \end{equation}
\begin{proof}
    Given that we form $\StrategyRawPaid$ as shown in Example~\ref{ex:frt}, $\|\StrategyRawPaid\|_1$ simply represents the maximum number of overlapping RCQs in $\StrategyRawPaid$.
    Overlapping RCQs on the tree $\mathcal{T}$ must occur on the same path, that is, they form an ancestor-descendant relationship, since the children of each node $n$ have non-overlapping ranges.
    Thus, the maximum number of overlapping RCQs across all tree paths, is equal to the $\|\StrategyRawPaid\|_1$.
\end{proof}

\subsubsection{Proof of Lemma~\ref{lemma:proac-invariant}} The lemma states that the proactive strategy $\StrategyRawProactive$ generated by \textsc{generateProactiveStrategy} for an input $\StrategyRawPaid$ satisfies the condition: 
    \begin{equation}
        \label{eqn:proac-invariant}
        \forall \text{ paths } p \in \mathcal{T}, 
        \sum_{v\in p} \mathcal{M}_{\StrategyRawPaid\cup \StrategyRawProactive}(v) 
        \leq     \mathcal{S}_{\StrategyRawPaid}(\mathcal{T}\text{.root}) 
        = \|\StrategyRawPaid\|_{1}
    \end{equation}
\begin{proof}
    This inequality holds for $\StrategyRawProactive=\emptyset$ by applying Lemma~\ref{lemma:subtree-norm-root} and Definition~\ref{def:subtree-norm}.
    \textsc{generateProactiveStrategy} only adds a node to $\StrategyRawProactive$ if it satisfies the following condition on line~\ref{line:proac-condition}, the remaining path has a length less than $r$. At the root, this condition is set to be less than $\|\StrategyRawPaid\|_{1}$.
\end{proof}

\subsubsection{Proof of Theorem~\ref{thm:proactive}}
The theorem states that, given a paid strategy matrix $\StrategyRawPaid$, 
Algorithm~\ref{algo:proactive} outputs $\StrategyRawProactive$ such that $\|\StrategyRawPaid\cup\StrategyRawProactive\|_1=\|\StrategyRawPaid\|_1$. 

\begin{proof}
Applying Lemma~\ref{lemma:subtree-norm-root} and Definition~\ref{def:subtree-norm} for the matrix $\StrategyRawPaid\cup\StrategyRawProactive$: 
\begin{equation}
    \label{eqn:lhs-proactive-lemma-proof}
        \|\StrategyRawPaid\cup\StrategyRawProactive\|_1 =
        \mathcal{S}_{\StrategyRawPaid\cup\StrategyRawProactive}(\mathcal{T}\text{.root}) = 
        \max_{p \in \text{subtree}(v)} \sum_{v\in p} \mathcal{M}_{\StrategyRawPaid\cup\StrategyRawProactive}(v)
\end{equation}
Rephrasing Lemma~\ref{lemma:proac-invariant}:
\begin{equation}
    \label{eqn:rhs-proactive-lemma-proof}
        \max_{p \in \text{subtree}(v)} \sum_{v\in p} \mathcal{M}_{\StrategyRawPaid\cup\StrategyRawProactive}(v) 
        = \|\StrategyRawPaid\|_1
\end{equation}
Using equations~\ref{eqn:lhs-proactive-lemma-proof} and~\ref{eqn:rhs-proactive-lemma-proof}, we get: 
\begin{equation}
    \label{eqn:proactive-lemma-proof-raw-form}
        \|\StrategyRawPaid\cup\StrategyRawProactive\|_1 =
        \|\StrategyRawPaid\|_1
\end{equation}
\end{proof}




\section{MMM Cache-aware tight bound} \label{app:tight-bound-with-free-paid}
Given the cached noise parameters, we propose a theoretical upper bound \TightNoiseParameter~for the candidate \CandidatePaidNoiseParameter~ required to satisfy the $(\alpha, \beta)$-accuracy guarantee.
\begin{equation}
    \label{eq:tight-b-bound-cache}
    \TightNoiseParameter \leq \frac{\sqrt{\alpha^2 \beta/2 - \|\WorkloadFull \StrategyFull^+Diag(\vec{b_{\StrategyMatrixFree}})\|_F^2}}{\|WA^+Diag(I_{|\StrategyMatrixPaid|})\|_F}
\end{equation}
Here,  $\vec{b_{\StrategyMatrixFree}}=[b_1,\ldots,b_{|\StrategyMatrixFree|}]$. 
In doing so, we generalize Ge et al.'s tight bound $\TightNoiseParameter$ for $\Cache=\phi$ (equation~\ref{eq:mm-loose-noise-param-bound}) to consider cached noise parameters. 

\section{Implementation Details}
We implement an initial prototype of \sysname in Python. We use the source code provided by Ge et al. to evaluate APEx\footnote{\url{https://github.com/cgebest/APEx}}. Since the authors of Pioneer did not publish their code, we implement it from scratch in python following the paper. We implement a simple composite plan using all hyperparamters as described in the paper. We show in Section~\ref{sec:privacy_budget_comp} that our implementation reproduces similar performance to the results given in the original paper~\cite{pioneer}[Section 7.2.1]. We make our full evaluation including our implementation of related work publicly available\footnote{\url{https://git.uwaterloo.ca/m2mazmud/cachedp-public.git}}.

We note that pioneer uses variance based accuracy where as APEx uses $(\alpha, \beta)$ accuracy requirements. This is not a problem for \sysname as we accept both types of accuracy requirement. To run Pioneer on workloads with an $(\alpha, \beta)$ requirement, we use an MC simulation to search for the variance that satisfies the accuracy requirement. To run APEx on workloads with a variance based accuracy requirement we utilize a tail bound on the Laplace distribution (since we find empirically that apex uses the Laplace Mechanism for all single range queries) to get the alpha beta.

\section{Evaluation of the SE heuristics}\label{sec:se_heuristic_analysis}
In this section, we reason about the effectiveness of the strategy expander heuristics, both experimentally and theoretically.
 \paragraph{Experiments:} We first quantify the probability of SE being selected over MMM using the results of our frequency analysis, from Table~\ref{tab:module_usage_counts}. We refer the reader to Section~\ref{sec:evaluation} for all experimental details.
    Out of all paid workloads for which MMM and SE were selected (second and fourth rows of the table), we consider the probability for SE to be selected.  
    SE is overwhelmingly more likely to be chosen over MMM on both BFS and DFS tasks, with the probability of selection being $90\%$ for BFS and $95\%$ for DFS. 
    On the other hand, MMM is around twice as likely to be chosen over SE on both RRQ and IDEBench tasks, as the SE selection probability is $31\%$ for RRQ and $34\%$ for IDEBench.
    We conclude that the likelihood of the SE module being chosen, and thus the success of our heuristics, depends on the workload sequence since it influences the contents of our cache.
    Averaging across all four tasks, the likelihood of the SE module being chosen over MMM, is around $62\%$.
    
\paragraph{Theoretical Analysis:} We analyze the conditions under which our heuristics result in SE module being selected. The SE module is only selected if it has a lower cost than the original strategy in MMM. The privacy cost for our mechanisms is inversely proportional to the error term, for a given $(\alpha, \beta)$ or $\alpha^2$-expected total square error. Thus, our heuristics are only successful if they lead to the error term for the SE module to be smaller than the error term for MMM. 
   
We recall Figure~\ref{fig:strat_expander_coutner_ex} included an example to show that expanding the strategy can lead to an increased error term. 
We analyze why such situations can arise and describe conditions for when our noise parameter-based heuristic can reduce the error. 
Our heuristic to choose $b_{\ell+1}<\PaidNoiseParameter$ strictly improves the error under the following condition: 
\begin{theorem}
Given a workload $\WorkloadFull$, a strategy $\StrategyFull$ of $\ell$ rows, and a noise vector $\BList$, adding a new row to $\StrategyFull$ to form $\StrategyFull_e$
reduces the error, i.e., 
    \begin{equation}\label{eqn:to-prove}
        \|\WorkloadFull\StrategyFull^+Diag(\BList)\|_F^2 \geq \|\WorkloadFull\StrategyFull_e^+Diag(\BList||b_{\ell+1})\|_F^2
    \end{equation} 
if all entries in $\BList$ equal $b^*$ and $b_{\ell+1} < \PaidNoiseParameter \leq b^*$, for any $b^*$.
\end{theorem}

\begin{proof}
    \noindent We recall the following theorem proved by Li et al.~\cite{li2015matrix} in their MM paper.
    \begin{equation}\label{eqn:li-etal-new}
        \|\WorkloadFull\StrategyFull^+\|_F^2 \geq \|\WorkloadFull\StrategyFull_e^+\|_F^2 
    \end{equation}
    For $\BList=[b^* \cdots b^*]$, we transform  Equation~\ref{eqn:to-prove} to:
    \begin{eqnarray}
    \|\WorkloadFull\StrategyFull^+Diag(\BList)\|_F^2 &=& \|\WorkloadFull\StrategyFull^+(b^* I)\|_F^2 = (b^*)^2\|\WorkloadFull\StrategyFull^+\|_F^2\\
    &\geq& (b^*)^2\|\WorkloadFull\StrategyFull_e^+\|_F^2 \\
    &=&\|\WorkloadFull\StrategyFull_e^+Diag(\BList||b^*)\|_F^2\\
    &\geq&\|\WorkloadFull\StrategyFull_e^+Diag(\BList||b_{\ell+1})\|_F^2 
    \end{eqnarray}

where the first inequality comes from applying Li et al.'s result (Equation~\ref{eqn:li-etal-new}) and the final inequality by applying the condition: $b_{\ell+1} \leq b^*$.
  \end{proof} 

When the noise parameters are not all equal, the sufficiency conditions become more complicated. For example, if we modify the above counterexample to have a slightly more accurate expanded row ($b_{\ell+1}=3b<4b$), we get a lower error than MMM:
\[
\|\WorkloadFull_{1}\StrategyFull_{1}^+ diag(\BList_{1})\| = 28b^2 ,\
    \|\WorkloadFull_{1}\StrategyFull_{1e}^+ diag(\BList_{1e})\| = 26.5b^2
    \]
Thus our greedy heuristic, which does not simply add entries less than $b_p$, but prioritizes the most accurate rows first, would be effective in this scenario. We also observe that distribution of the existing noise parameters in $\BList$ is important. If the noise parameters are all close to each other, strategy expansion is likely to be more beneficial. For example changing the second entry of $\BList_{1}$ from $b$ to $2b$ (i.e. $\BList_{1e} = [b, 2b, 5b, 4b]$), results in a smaller error than MMM:
\[
    \|\WorkloadFull_{1}\StrategyFull_{1}^+ diag(\BList_{1})\| = 31.0b^2 ,\
    \|\WorkloadFull_{1}\StrategyFull_{1e}^+ diag(\BList_{1e})\| = 30.2b^2
\]

It is evident that both the distribution of the noise parameters in $\BList$ and the new noise parameter $b_{\ell+1}$, influence the success of our accuracy heuristic. 
Additionally, we observe that a newly added row could satisfy our final heuristic by being a parent or a child of \emph{any} of the existing rows in $\StrategyFull_{1}$. We find that the structure of the expanded strategy $\StrategyFull_{1e}$ also significantly influences the error term. 
For example, changing the final row in $\StrategyFull_{1e}$ from $(1,1,1)$ to $(1,0,1)$ or $(0,1,1)$ also reduces the error for $\StrategyFull_{1e}$ 
to $26.5b^2$ and $20b^2$ respectively, under the exact same noise parameters as our counterexample. However, changing the final row to $(1,1,0)$ increases the error to $34.7b^2$. 

\paragraph{Conclusion:} We observe that whether an additional row selected by our noise parameter heuristics will succeed in decreasing the error for SE over MMM, depends on how the new row changes elements in $\WorkloadFull\StrategyFull_e^+$. However, we can guarantee that when the noise parameters in $\BList$ are sufficiently similar, strategy expander will reduce the error regardless of the structure. Our structure-based heuristic only allows strategy queries related through a parent-child relationship to an existing strategy query in $\StrategyFull$. Future research may model the success of this heuristic, by analyzing the relation between $\WorkloadFull\StrategyFull^+$ and $\WorkloadFull\StrategyFull_e^+$ for $\StrategyFull$ and $\StrategyFull_e$ that differ by a parent or child row.

\end{techreport}
\end{document}